%% file: permutation-graphs.tex
	\documentclass[
		paper=a4,%
		abstract=true,
		11pt,american,pagesize,%
		DIV=12,headinclude=true,headlines=0,footinclude=true,footlines=-5,twoside=semi,%
	]{scrartcl}
	\def\docclass{koma}

	\def\version{arxiv}

	\def\draftmode{false} 

\AtBeginDocument{
	\iflipics{
		\renewcommand\paragraph[1]{\subparagraph{#1.}}
	}{}
}

\usepackage[utf8]{inputenc}
\input{preamble}

\usepackage{hyperref}

\title{Succinct Permutation Graphs}

\iflipics{
	\author{Konstantinos Tsakalidis}{University of Liverpool, UK}{K.Tsakalidis@liverpool.ac.uk}{https://orcid.org/0000-0001-6470-9332}{}
	\author{Sebastian Wild}{University of Liverpool, UK}{Sebastian.Wild@liverpool.ac.uk}{https://orcid.org/0000-0002-6061-9177}{}
	\author{Viktor Zamaraev}{University of Liverpool, UK}{Viktor.Zamaraev@liverpool.ac.uk}{https://orcid.org/0000-0001-5755-4141}{}
	
	\authorrunning{K. Tsakalidis, S. Wild, and V. Zamaraev}
	\Copyright{Konstantinos Tsakalidis, Sebastian Wild, and Viktor Zamaraev}
	
		\ccsdesc[500]{Theory of computation~Data structures design and analysis}%
		\ccsdesc[300]{Theory of computation~Data compression}%

	\keywords{%
		succinct data structures, distance oracles, permutation graphs, 
		bipartite permutations graphs, circular permutation graphs, graph compression,
		graph encoding%
	}
}{}
\ifkoma{
	\newcommand\email[1]{\texttt{#1}}
	\author{%
		Konstantinos Tsakalidis%
			\footnote{University of Liverpool, UK, 
			\email{K.Tsakalidis\,@\,liverpool.ac.uk}}
	\and
		Sebastian Wild%
			\footnote{University of Liverpool, UK, 
			\email{Sebastian.Wild\,@\,liverpool.ac.uk}}
	\and
		Viktor Zamaraev%
			\footnote{University of Liverpool, UK, 
			\email{Viktor.Zamaraev\,@\,liverpool.ac.uk}}
	}
	
	\date{\small\today}
}{}
\iflncs{
	\author{%
		Konstantinos Tsakalidis\inst{1}\orcidID{0000-0001-6470-9332} \and
		Sebastian Wild\inst{1}\orcidID{0000-0002-6061-9177} \and
		Viktor Zamaraev\inst{1}\orcidID{0000-0001-5755-4141}%
}

	\authorrunning{K. Tsakalidis, S. Wild, and V. Zamaraev}
	\institute{%
		University of Liverpool, UK
		\email{\{K.Tsakalidis,\,Sebastian.Wild,\,Viktor.Zamaraev\}\,@\,liverpool.ac.uk}
	}
}{}
\ifspringer{
	\author*[1]{\fnm{Konstantinos} \sur{Tsakalidis}}\email{K.Tsakalidis\,@\,liverpool.ac.uk}
	
	\author*[1]{\fnm{Sebastian} \sur{Wild}}\email{Sebastian.Wild\,@\,liverpool.ac.uk}
	
	\author*[1]{\fnm{Viktor} \sur{Zamaraev}}\email{Viktor.Zamaraev\,@\,liverpool.ac.uk}
	
	\affil[1]{\orgdiv{Department of Computer Science}, \orgname{University of Liverpool}, \orgaddress{\street{Ashton Building, Ashton Street}, \city{Liverpool}, \postcode{L69 3BX}, \state{Merseyside}, \country{UK}}}
	
}{}
\acknowledgements{%
	We thank the TCS Open Problem Group of the Department of Computer Science of University of Liverpool 
	for creating the productive environment from which this research originated.
}

\begin{document}

\maketitle

%

\begin{abstract}\iflipics{\sloppy}{}
We present a succinct data structure for permutation graphs,
and their superclass of circular permutation graphs, \ie,
data structures using optimal space up to lower order terms.
Unlike concurrent work on circle graphs~\cite{AcanChakrabortyJoNakashimaSadakaneRao2022},
our data structure also supports distance and shortest-path queries, 
as well as adjacency and neighborhood queries, all in optimal time.
We present in particular the first succinct exact distance oracle for (circular) permutation graphs.
A second succinct data structure also supports degree queries in time independent of the neighborhood's size at the expense of an $\Oh(\log n/\log \log n)$-factor overhead in all running times.
Furthermore, we develop a succinct data structure for the class of bipartite permutation graphs.
We demonstrate how to run algorithms directly over our succinct representations for several problems on permutation graphs: \textsc{Clique}, \textsc{Coloring}, \textsc{Independent Set}, 
\textsc{Hamiltonian Cycle}, \textsc{All-Pair Shortest Paths}, and others.

Finally, we initiate the study of \emph{semi-distributed} graph representations; 
a concept that smoothly
interpolates between distributed (labeling schemes) and centralized (standard data structures).
We show how to turn some of our data structures into semi-distributed representations by storing only 
$\Oh(n)$ bits of additional global information, circumventing the lower bound on distance labeling schemes for 
permutation graphs.
\end{abstract}

%
%
%
%
%
%
%
%

%
%
%
%
%
%
%
\input{permutation-graphs-intro}
%
\input{permutation-graphs-prelims}
%
%

\input{permutation-graphs-ds}

\ifsubmission{

\section{Further results}

In the appendix, we show that succinct PGs can be used 
to implement various algorithms for PGs as if the graph was given in non-compressed format.
These algorithms heavily build on transitive orientations and topological sorts
of $G$ and $\overline G$; we show how our succinct PGs can provide these.
Moreover, we show that we can simulate access to the ``unrolled'' PG of a CPG 
with asymptotically no extra space.
This allows us to extend our data structures to CPGs.
Finally, using the insight that BPGs consist solely of type $A$ and type $B$ vertices,
we can get rid of $\pi^{-1}$ altogether, thus achieving the optimal $2n+o(n)$ bits on BPGs.

}{
	\input{permutation-graphs-algorithms}

	\input{permutation-graphs-bipartite}

	\input{permutation-graphs-circular}

	\input{permutation-graphs-semi-distributed}

}

\section{Conclusion}
\label{sec:conclusion}

We presented the first space-efficient data structures for permutation graphs (PGs), 
circular permutation graphs (CPGs),
and bipartite permutation graphs (BPGs). 
They use space close to the information-theoretic lower bound for these classes of graphs, 
while supporting many queries in optimal time.
The use of our data structures as space-efficient exact distance oracles 
improves the state of the art and proves a separation between standard, centralized data structures
and distributed graph labeling schemes for distance oracles in permutation graphs.
Our notion of semi-distributed graph representations interpolates between these two extremes;
an initial result shows that access to global memory is inherently more powerful even
if we cannot store the entire graph there.

There are several interesting directions for future research.
\ifsubmission{\begin{enumerateinline}}{\begin{enumerate}}
\item Is it possible to support degree queries in constant time and succinct space, 
together with the queries covered by our data structures?
With our current approach, this seems to require improvements to range searching in succinct grids,
but the queries are of a restricted form.
\item What is the least amount of global storage in a semi-distributed representation 
for distances in permutation graphs
that overcomes the lower bound for distance labeling schemes?
Is there a smooth trade-off between the ``amount of decentralization'' and total space, 
or does it exhibit a sharp threshold?
\item Comparability graphs of dimension $k$. 
	These graphs have representations with $k-1$ chord segments per vertex;
	PGs correspond to $k=2$. It is known~\cite{BazzaroGavoille2009} that 
	for $k\ge 3$, distance labels require $\Omega(n^{1/3})$ bits.
	Is a succinct distance oracle with efficient queries possible for these graphs?
\item Circle graphs.
	While navigational operations are possible~\cite{AcanChakrabortyJoNakashimaSadakaneRao2022},
	efficient distance queries remain an open problem.
\ifsubmission{\end{enumerateinline}}{\end{enumerate}}

\ifkoma{
	\myacknowledgements
}{}
%
%
%
%
%
\ifspringer{
	{
	\let\oldhref\href
	\renewcommand\href[1]{}
	\let\path\doi
	\bibliography{references}
	}
}{
	\bibliography{references}
}

\ifsubmission{
	\clearpage
	\appendix
	\ifkoma{\addpart{Appendix}}{}
	\section{Constant-space next neighbor}
	\label{app:rmq-next-neighbor}
	In this appendix, we the details of the range-minimum-query-based algorithm for 
	iterating through the result of a three-sided orthogonal range query 
	with constant extra space, and in constant amortized time per returned point.

\input{permutation-graphs-rmq-next-neighbor}\input{permutation-graphs-algorithms}\input{permutation-graphs-bipartite}\input{permutation-graphs-circular}\input{permutation-graphs-semi-distributed}
}

\end{document}

%% file: preamble.tex
\makeatletter
\usepackage{xifthen}

\newcommand\iflipics[2]{\ifthenelse{\equal{\docclass}{lipics}}{#1}{#2}}
\newcommand\ifkoma[2]{\ifthenelse{\equal{\docclass}{koma}}{#1}{#2}}
\newcommand\iflncs[2]{\ifthenelse{\equal{\docclass}{lncs}}{#1}{#2}}
\newcommand\ifspringer[2]{\ifthenelse{\equal{\docclass}{springer}}{#1}{#2}}
\ifthenelse{ \equal{\docclass}{lipics} \OR \equal{\docclass}{koma} \OR \equal{\docclass}{lncs} \OR \equal{\docclass}{springer} }{
	\PackageInfo{paper}{Building paper with docclass = \docclass} 
}{
	\PackageWarning{paper}{docclass = "\docclass", but must be one of "lipics", "koma", "lncs", "springer"}
}

\newcommand\ifmanuscript[2]{\ifthenelse{\equal{\version}{manuscript}}{#1}{#2}}
\newcommand\ifarxiv[2]{\ifthenelse{\equal{\version}{arxiv}}{#1}{#2}}
\newcommand\ifsubmission[2]{\ifthenelse{\equal{\version}{submission}}{#1}{#2}}
\newcommand\ifproceedings[2]{\ifthenelse{\equal{\version}{proceedings}}{#1}{#2}}
\ifthenelse{ 
	\equal{\version}{manuscript} 
	\OR \equal{\version}{arxiv} 
	\OR \equal{\version}{submission} 
	\OR \equal{\version}{proceedings} 
}{
	\PackageInfo{paper}{Building paper version = \version} 
}{
	\PackageWarning{paper}{version = "\version", but must be one of "manuscript", "arxiv", "submission", "proceedings"}
}

\newcommand\ifdraft[2]{\ifthenelse{\equal{\draftmode}{true}}{#1}{#2}}
\ifthenelse{ \equal{\draftmode}{true} \OR \equal{\draftmode}{false} }{
	\PackageInfo{paper}{Building paper with draftmode = \draftmode} 
}{
	\PackageWarning{paper}{draftmode = "\draftmode", but must be "true" or "false"}
}

\usepackage[T1]{fontenc}
\usepackage{lmodern}
\usepackage{slantsc}

\usepackage{babel}

\usepackage{array,multicol}
\usepackage{amsmath,amsfonts,amssymb,mathtools}
\usepackage{mleftright}\mleftright 
\usepackage{relsize,xspace,enumitem,booktabs,adjustbox,needspace,pbox,relsize}
\usepackage{graphicx}

\newlist{enumerateinline}{enumerate*}{1}
\setlist[enumerateinline]{
	label={(\arabic*)},
}

\usepackage{colonequals}
\usepackage{wref}
\usepackage[bibtex]{url-doi-arxiv}

\newdimen\makeboxdimen
\newcommand\makeboxlike[3][l]{%
\setbox0=\hbox{#2}%
\global\makeboxdimen=\wd0%
\setbox1=\hbox{\makebox[\makeboxdimen][#1]{%
\makebox[0pt][#1]{#3}%
}}%
\ht1=\ht0%
\dp1=\dp0%
\box1%
}

\newcommand\like[3][c]{%
	\mathchoice{
		\makeboxlike[#1]{%
			\ensuremath{\displaystyle\relax#2}%
		}{%
			\ensuremath{\displaystyle\relax#3}%
		}%
	}{
		\makeboxlike[#1]{%
			\ensuremath{\textstyle\relax#2}%
		}{%
			\ensuremath{\textstyle\relax#3}%
		}%
	}{
		\makeboxlike[#1]{%
			\ensuremath{\scriptstyle\relax#2}%
		}{%
			\ensuremath{\scriptstyle\relax#3}%
		}%
	}{
		\makeboxlike[#1]{%
			\ensuremath{\scriptscriptstyle\relax#2}%
		}{%
			\ensuremath{\scriptscriptstyle\relax#3}%
		}%
	}
}

\newcommand\plaincenter[1]{%
	\mbox{}\hfill#1\hfill\mbox{}%
}

\ifkoma{
	\setlength\parindent{1.5em}
	\usepackage[headsepline]{scrlayer-scrpage}
	\pagestyle{scrheadings}
	\clearscrheadfoot
	\AtBeginDocument{%
		\automark[section]{}%
	}
	\ohead{\pagemark}
	\rehead{\mytitle}
	\lohead{\headmark}
	\addtokomafont{caption}{\sffamily\small}
	\addtokomafont{captionlabel}{\sffamily\textbf}
	\setcapmargin{2em}
}{}

\AtBeginDocument{%
	\let\mytitle\@title%
}

\let\oldthebibliography\thebibliography
\renewcommand\thebibliography[1]{%
	\oldthebibliography{#1}%
	\pdfbookmark[1]{References}{}%
}

\usepackage{lscape} 

\ifkoma{
	\usepackage{float}
	\floatstyle{plain}
	\usepackage{newfloat}
	\DeclareFloatingEnvironment[%
			name=Algorithm,%
			placement=thb,%
		]{algorithm}
}{}
\iflipics{
	\usepackage{newfloat}
	\DeclareFloatingEnvironment[%
			name=Algorithm,%
			placement=thb,%
		]{algorithm}
}{}

\usepackage{placeins} 

\usepackage{wclrscode}

\usepackage{textcomp} 
\usepackage{listings}

\usepackage{tikz}

\usetikzlibrary{positioning,arrows.meta}
\usetikzlibrary{backgrounds,calc,trees,graphs,decorations,decorations.pathreplacing}

\pgfdeclarelayer{background}
\pgfsetlayers{background,main}

\iflipics{
	\newtheorem{fact}[theorem]{Fact}
	\newtheorem{conjecture}[theorem]{Conjecture}
	
	\newtheorem*{remarknonumber}{Remark}
}{}
\iflncs{
	\newtheorem{fact}[theorem]{Fact}

}{}
\ifkoma{
	\usepackage[amsmath,hyperref,thmmarks]{ntheorem}
	
	\theorembodyfont{\slshape}
	\theoremseparator{:}
	\newtheoremstyle{proofstyle}%
	  {\item[\theorem@headerfont\hskip\labelsep ##1\theorem@separator]}%
	  {\item[\theorem@headerfont\hskip\labelsep ##1 of ##3\theorem@separator]}
	
	\theorempreskip{\topsep} 

	\theoremsymbol{\adjustbox{scale=.8}{$\triangleleft\mkern-1mu$}}
	
	\newtheorem{theorem}{Theorem}[section]
	
	\theoremstyle{plain}
	\theorempreskip{\topsep}
	
	\newtheorem{proposition}[theorem]{Proposition}
	\newtheorem{lemma}[theorem]{Lemma}
	\newtheorem{conjecture}[theorem]{Conjecture}
	\newtheorem{corollary}[theorem]{Corollary}
	\newtheorem{definition}[theorem]{Definition}
	
	\theoremstyle{plain}
	\theorembodyfont{\upshape}
	
	\newtheorem{property}[theorem]{Property}
	\newtheorem{fact}[theorem]{Fact}
	\newtheorem{remark}[theorem]{Remark}
	\newtheorem{example}[theorem]{Example}
	\newtheorem{claim}[theorem]{Claim}
	
	\theorembodyfont{\upshape}
	\newtheorem{remarknonumber}[theorem]{Remark}
	
	\theoremsymbol{\raisebox{-.25ex}{$\Box$}}
	\qedsymbol{\raisebox{-.25ex}{$\Box$}}
	
	\theoremstyle{proofstyle}
	\newtheorem{proof}{Proof}
	
}{}
\ifspringer{
	\theoremstyle{thmstyleone}
	\newtheorem{theorem}{Theorem}[section]

	\newtheorem{lemma}[theorem]{Lemma}
	
	\newtheorem{corollary}[theorem]{Corollary}

	\theoremstyle{thmstyletwo}

	\newtheorem{remark}[theorem]{Remark}
	
	\newtheorem{claim}[theorem]{Claim}
	
	\theoremstyle{thmstylethree}
	\newtheorem{remarknonumber}{Remark}
	
}

\iflipics{
	\newenvironment{thmenumerate}[2][]{%
		\begin{enumerate}[
			label={\textsf{\textbf{\color{darkgray}{\makebox[\widthof{(a)}][c]{\textup{(\alph*)}}}}}},
			ref={\ref{#2}\kern.1em--\kern.1em(\alph*)},
			itemsep=0pt,
			topsep=.5ex,
			leftmargin=1.75em,
			#1
		]%
	}{%
		\end{enumerate}%
	}
}{
	\newenvironment{thmenumerate}[2][]{%
		\begin{enumerate}[
			label={\makebox[\widthof{(a)}][c]{\textup{(\alph*)}}},
			ref={\ref{#2}\kern.1em--\kern.1em(\alph*)},
			itemsep=0pt,
			#1
		]%
	}{%
		\end{enumerate}%
	}
}

\newcommand*\ie{\mbox{i.\hspace{.2ex}e.}}
\newcommand*\eg{\mbox{e.\hspace{.2ex}g.}}

\newcommand*\wrt{\mbox{w.\hspace{.2ex}r.\hspace{.2ex}t.}\xspace}

\newcommand*\withoutlossofgenerality{\mbox{w.\hspace{.23ex}l.\hspace{.2ex}o.\hspace{.17ex}g.}\xspace}

\newcommand\N{\mathbb N}
\newcommand\Z{\mathbb Z}

\newcommand\Oh{O}

\usepackage{fixmath}

\newcommand{\ESymbol}{\mathbb{E}}

\newcommand{\ProbSymbol}{\ensuremath{\mathbb{P}}}

\DeclarePairedDelimiterXPP\Prob[1]{\ProbSymbol}[]{}{%
	#1%
}
\DeclarePairedDelimiterXPP\E[1]{\ESymbol}[]{}{%
	#1%
}
\DeclarePairedDelimiterXPP\Eover[2]{\ESymbol_{#1}}[]{}{%
	#2%
}
\DeclarePairedDelimiterXPP\ProbIn[2]{\ProbSymbol_{#1}}[]{}{%
	#2%
}
\providecommand{\Prob}{} 
\providecommand{\ProbIn}{} 
\providecommand{\E}{} 
\providecommand{\Eover}{} 

\newcommand{\surroundedmath}[3]{
	\mathchoice{
		#1{#2{#3}#2}%
	}{
		#1{#3}%
	}{
		#1{#3}%
	}{
		#1{#3}%
	}%
}

\newcommand\wrel[1]{\surroundedmath{\mathrel}{\;}{#1}}
\newcommand\wwrel[1]{\surroundedmath{\mathrel}{\;\;}{#1}}
\newcommand\bin[1]{\surroundedmath{\mathbin}{\:}{#1}}
\newcommand\wbin[1]{\surroundedmath{\mathbin}{\;}{#1}}

\newcommand\ppe{\phantom{=}}

\ifkoma{
	\makeatletter
	\let\oldalign\align
	\let\endoldalign\endalign
	\renewenvironment{align}{%
		\begingroup%
		\let\oldhalign\halign
		\def\halign{%
			\let\oldbreak\\%
			\def\nonnumberbreak{\nonumber\oldbreak*}%
			\def\\{%
				\@ifstar{\nonnumberbreak}{\oldbreak}%
			}%
			\oldhalign%
		}
		\oldalign%
	}{%
		\endoldalign%
		\endgroup%
	}
}{}
\newcommand*\numberthis[1][]{\stepcounter{equation}\tag{\theequation}}

\allowdisplaybreaks[3]

\newcommand\splitaftercomma[1]{%
  \begingroup
  \begingroup\lccode`~=`, \lowercase{\endgroup
    \edef~{\mathchar\the\mathcode`, \penalty0 \noexpand\hspace{0pt plus .25em}}%
  }\mathcode`,="8000 #1%
  \endgroup
}

\def\mydots{\xleaders\hbox to.5em{\hfill.\hfill}\hfill}
\newlength\tmpLenNotations

\iflipics{
	\ifmanuscript{\hideLIPIcs}{}
	\ifarxiv{\hideLIPIcs}{}
	\ifsubmission{}{\nolinenumbers}
}{}

\ifdraft{}{%
	\usepackage{microtype}
}

\hypersetup{
	final,
	unicode=true, 
	bookmarks=true,
	bookmarksnumbered=true,
	bookmarksdepth=2,
	bookmarksopen=true,
	breaklinks=true,
	hidelinks,
}

\newsavebox\tmpbox

\iflipics{
	\let\oldparagraph\paragraph
	\renewcommand\paragraph[1]{%
		\oldparagraph*{#1}
	}
}{
	\let\oldparagraph\paragraph
	\renewcommand\paragraph[1]{%
		\oldparagraph{#1.}
	}
}

\let\epsilon\varepsilon

\def\myacknowledgements{}
\ifkoma{
	\newcommand\acknowledgements[1]{\def\myacknowledgements{\paragraph{Acknowledgements}#1}}
}{}
\iflncs{
	\newcommand\acknowledgements[1]{\def\myacknowledgements{\paragraph{Acknowledgements}#1}}
}{}
\ifspringer{
	\newcommand\acknowledgements[1]{\def\myacknowledgements{\bmhead{Acknowledgments}#1}}
}{}

\ifspringer{
	
	\renewenvironment{abstract}[1]{%
		\abstractfont%
		\abstracthead*{\abstractname}%
		#1\par%
	}{}
}{}

\newcommand\idtt[1]{\texttt{\upshape #1}\xspace}

\newcommand*\rankop{\idtt{rank}}
\newcommand*\selop{\idtt{select}}
\newcommand*\accessop{\idtt{access}}

\newcommand*\GDegree{\idtt{deg}}
\newcommand*\GAdjacent{\idtt{adj}}
\newcommand*\GNeighbor{\idtt{nbrhood}}
\newcommand*\GNextNeighbor{\idtt{next\_nbr}}
\newcommand*\GSPath{\idtt{spath}}
\newcommand*\GDistance{\idtt{dist}}
\newcommand*\GSPathFirst{\idtt{spath\_succ}}

\newcommand*\GridCount{\idtt{RCount}}
\newcommand*\GridReport{\idtt{RPoints}}
\newcommand*\GridXForY{\idtt{XForY}}
\newcommand*\GridYForX{\idtt{YForX}}

\newcommand*\RangeMin{\idtt{rmq}}
\newcommand*\RangeMax{\idtt{rMq}}
\newcommand*\NextAbove{\idtt{next\_above}}

\makeatother

%% file: permutation-graphs-intro.tex
\section{Introduction}

As a result of the rapid growth of data sets, 
memory requirements become a bottleneck in many applications;
in particular when data structures 
do no longer fit into faster levels of the memory hierarchy of computer systems.
Research on \emph{succinct data structures} 
has lead to optimal-space data structures for many types of data~\cite{Navarro2016},
significantly extending the size of data sets that can be analyzed efficiently
on commodity hardware.
A data structure is called \emph{succinct} when its space usage is optimal up to lower order terms, \ie,
optimal up to a factor of $1+o(1)$.

Graphs are one of the most widely used types of data.
In this paper, we study succinct representations of 
specific classes of graphs,
namely permutation graphs and related families of graphs.
A graph is a \emph{permutation graph} (PG) if it can be obtained as the 
intersection graph of chords (line segments) between two parallel lines~\cite{PnueliLempelEven1971},
\ie, the vertices corresponding to two such chords are adjacent, if and only if the chords intersect.
PGs are a well-studied class of graphs;
they are precisely the comparability graphs
of two-dimensional partial orders, and the class of comparability graphs
whose complement graph is also a comparability graph~\cite{dushnik1941partially} 
(see \wref{sec:preliminaries} for definitions of these concepts).
Many generally intractable graph problems can be solved efficiently 
on PGs,
for instance 
\textsc{Clique}~\cite{Moehring1985,McConnellSpinrad1999}, 
\textsc{Independent Set}~\cite{Moehring1985,McConnellSpinrad1999},
\textsc{Coloring}~\cite{Moehring1985,McConnellSpinrad1999},
\textsc{Clique Cover}~\cite{Moehring1985,McConnellSpinrad1999},
\textsc{Dominating Set}~\cite{ChaoHsuLee2000},
\textsc{Hamiltonian Cycle}~\cite{DeogunSteiner1994}, and
\textsc{Graph Isomorphism}~\cite{Colbourn1981}.
\textsc{All-Pair Shortest Paths} on PGs can be solved faster than in general graphs \cite{mondal2003optimal,BazzaroGavoille2009}. 
Moreover, PGs can be recognized in linear time~\cite{McConnellSpinrad1999}.

In this paper we study how to succinctly encode permutation graphs, while supporting
the following queries efficiently:
\begin{itemize}
	\item $\GAdjacent(u, v)$: whether vertices $u$ and $v$ are adjacent;
	\item $\GDegree(v)$: the degree of vertex $v$, \ie, the number of vertices adjacent to $v$;
	\item $\GNeighbor(v)$: the vertices adjacent to vertex $v$;
	\item $\GNextNeighbor(u,v)$: the successor of vertex $v$ in the adjacency list of vertex~$u$;
	\item $\GSPath(u, v)$: listing a shortest path from vertex $u$ to vertex $v$;
	\item $\GSPathFirst(u, v)$: the first vertex after vertex $u$ on a shortest path from $u$ to vertex $v$;
	\item $\GDistance(u, v)$: the length of the shortest path from vertex $u$ to vertex $v$.
\end{itemize}

\paragraph{Data structures} 
A succinct data structure is space optimal in the sense that it stores a given combinatorial object using asymptotically only the information-theoretic minimum of bits. %
Specifically, given a class of graphs $\mathcal C$ and denoting by $\mathcal C_n$ for the set of graphs $G\in\mathcal C$
on $|V(G)| = n$ vertices, 
a succinct data structure for $\mathcal C$ is allowed to spend $(1+o(1))\lg|\mathcal C_n|$ bits of space
when representing a graph in $\mathcal C_n$.
We present the first succinct data structures that support the above queries on a PG (\wref{thm:main}), 
as well as on its generalization, the  \emph{circular permutation graphs} (CPGs, see \wref{thm:succinct-cpg}). 
Moreover, we present the first succinct data structure for the special case of a \emph{bipartite permutation graph} (BPG, see \wref{thm:bipartite-ds}). %
\wref{tab:results} summarizes these results. 
\footnote{%
	Throughout this paper, running times assume the word-RAM model 
	with word size $w = \Theta(\log n)$, 
	where $n$ denotes the number of vertices of the input PG.
}

\begin{table}[thb]
\ifsubmission{%
	\newcommand\lame{}%
}{%
	\newcommand\lame{\color{black!50}}%
}%
	\caption{%
		Our data structure results for (variants of) permutation graphs with $n$ vertices.
		Space is in bits. Query times are $O(\cdot)$ bounds;
		$\mathit{deg}$ denotes the queried vertex' degree and $\mathit{dist}$
		the shortest-path distance between the queried vertices.
	}

	\small
	\adjustbox{max width=\linewidth}{%
	\begin{tabular}{rcc@{\qquad}cc}
	\toprule
		\multicolumn{1}{c}{} &                    \multicolumn{2}{c}{\bfseries permutation graphs}                     &               \bfseries bipartite               &            \bfseries circular            \\
		\addlinespace[2pt]
		\multicolumn{1}{c}{} &  \multicolumn{1}{c}{\bfseries (a) by grid}   &         \bfseries (b)  by array          &              \bfseries permutation              &          \bfseries permutation           \\
	\midrule
		                                                                              Space Usage &            $n\lg n + o(n \log n)$            &            $n\lg n + \Oh(n)$             &                    $2n+o(n)$                    &            $n\lg n + \Oh(n)$             \\
		                                                                      \strut  Lower Bound &   $\sim n\lg n$~\cite{BazzaroGavoille2009}   & $\sim n\lg n$~\cite{BazzaroGavoille2009} & $\sim 2n$~\cite{SaitohOtachiYamanakaUehara2012} & $\sim n\lg n$~\cite{BazzaroGavoille2009} \\
	\midrule
		           \GAdjacent &            $\log n / \log\log n$             &                   $1$                    &                     $1$                      &                   $1$                    \\
		             \GDegree &            $\log n / \log\log n$             &          \lame$\mathit{deg}+1$           &                     $1$                      &          \lame$\mathit{deg}+1$           \\
		           \GNeighbor &   ${(\mathit{deg}+1)\log n}/{\log\log n}$    &             $\mathit{deg}+1$             &               $\mathit{deg}+1$               &             $\mathit{deg}+1$             \\
		       \GNextNeighbor & \lame${(\mathit{deg}+1)\log n}/{\log\log n}$ &             $1$ (amortized)              &                     $1$                      &             $1$ (amortized)              \\
		              \GSPath &    $(\mathit{dist}+1)\log n / \log\log n$    &            $\mathit{dist}+1$             &              $\mathit{dist}+1$               &            $\mathit{dist}+1$             \\
		         \GSPathFirst &            $\log n / \log\log n$             &                   $1$                    &                     $1$                      &                   $1$                    \\
		           \GDistance &            $\log n / \log\log n$             &                   $1$                    &            \lame$\mathit{dist}+1$            &                   $1$                    \\
	\midrule
		              Theorem &         \wref[Thm.\!]{thm:main-grid}         &   \wref[Thm.\!]{thm:main-dist-oracle}    &       \wref[Thm.\!]{thm:bipartite-ds}        &     \wref[Thm.\!]{thm:succinct-cpg}      \\
	\bottomrule
	\end{tabular}
	}
	\label{tab:results}
\end{table}

To our knowledge, the only centralized data structures that store PGs are presented by 
Gustedt et al.~\cite{GustedtMorvanViennot1995} and by Crespelle and Paul~\cite{CrespellePaul2010}. 
The former are not succinct (using $\Oh(n\lg n)$ \emph{words} of space), 
but are parallelizable~\cite{GustedtMorvanViennot1995}. 
The latter support only \GAdjacent queries (in constant time), 
but are dynamic\ifsubmission{}{ (supporting insertions and deletions of vertices/chords and edges)}.
We are not aware of previous work on data structures for CPGs,
or on space-efficient data structures for BPGs.
\ifsubmission{}{\par}

Bazzaro and Gavoille~\cite{BazzaroGavoille2009} present \emph{distance labeling schemes} for PGs,
a distributed distance oracle, where the distance of two vertices can be computed solely from
the two labels of the vertices.
Their scheme uses labels of $\sim 9 \lg n$ bits per vertex%
\footnote{%
	By $\sim$ we denote a leading-term asymptotic approximation, \ie, $f(n) \sim g(n)$ iff $f(n)/g(n)\to 1$ as $n\to\infty$.
}, 
and their \GDistance queries take constant time. 
By concatenating all labels, their labeling scheme implies a data structure with matching time complexity 
and total space of $\sim 9 n \lg n$ bits. 
Our data structures (\wref{thm:main}) 
improve upon that space, while simultaneously supporting further queries besides \GDistance.
 
Interestingly, Bazzaro and Gavoille~\cite{BazzaroGavoille2009} further give a \emph{lower bound} of
$3 \lg n - \Oh(\lg \lg n)$ bits per vertex for \GDistance labeling schemes on PGs.
Comparing our data structures to this lower bound reveals a separation in terms of total space 
between their distributed and our centralized model:
giving up the distributed storage requirement, a data structure using the optimal $\sim n\lg n$ bits of space, 
\ie, $\lg n$ per vertex, becomes possible, 
proving that the centralized model is strictly more powerful.

\paragraph{Semi-distributed graph representations} 
To further explore the boundary of the above separation between 
standard centralized data structures and fully distributed labeling schemes,
we introduce a \emph{semi-distributed} model of computation for graph data structures 
that smoothly interpolates between these two extremes: 
in a $\langle L(n),D(n)\rangle$-space semi-distributed representation, 
each vertex locally stores a label of $L(n)$ bits, 
but all vertices also have access to a ``global'' data structure of $D(n)$ bits to support the queries.
Such a representation uses a total of $n L(n)+D(n)$ bits of space, 
but apart from the global part, only the labels of queried vertices are accessible to compute the answer.

The lower bound from~\cite{BazzaroGavoille2009} 
implies that when $D(n) = 0$, we must have $L(n)\ge 3 \lg n - \Oh(\lg \lg n)$
to support \GDistance on PGs, making the total space at least a factor 3 worse than the 
information-theoretic lower bound.
But what happens if we allow a small amount of global storage on top of the labels?
Is access to global storage inherently more powerful, 
even if insufficient encode the entire PG?
If so, what is the least amount of global storage that is necessary to overcome the labeling-scheme lower bound?
\ifsubmission{}{\par}
We do not comprehensively answer the latter question, but settle the former in
the affirmative: 
\ifsubmission{%
	we show that PGs admit a $\langle 2 \lg n, \Oh(n) \rangle$-space 
	semi-distributed representation that answers distance queries in $O(1)$ time.%
}{%
	we show that PGs admit a $\langle 2 \lg n, \Oh(n) \rangle$-space 
	semi-distributed representation that answers distance queries in constant time,
\ie, although the global space cannot distinguish all possible PGs,
	it suffices to circumvent the lower bound for labeling schemes in terms of total space and label size.
	Thus having access even to limited amounts of global space is inherently more powerful 
than a fully distributed data structure.
}

\paragraph{Applications}
Our data structures can replace the standard (space-inefficient) representation 
by adjacency lists in graph algorithms. 
For several known algorithms on PGs that make explicit use of their special structure
(namely, linear-time algorithms for computing 
minimum colorings, maximum cliques, maximum independent sets, or minimum clique covers), 
we show that they can be run with minimal extra space directly 
on top of our succinct representation.

Moreover, our data structures immediately yield an optimal-time all-pairs shortest-paths algorithm on PGs: 
For a PG with $n$ vertices and $m$ edges we can report all pairwise distances in $\Oh(n^2)$ time, 
matching the result of Mondal et al.~\cite{mondal2003optimal}; 
however, our approach is more flexible in that we can report the distances of any $k$ specified pairs 
of vertices in just $\Oh(n+m+k)$ total time. 
Furthermore, we can report the shortest paths (not just their lengths) in total time $\Oh(n+m+s)$, 
where $s$ is the size of the output; 
this does not immediately follow from~\cite{mondal2003optimal}.
The labeling scheme of~\cite{BazzaroGavoille2009} yields the same running times, but uses more space.

\paragraph{Further related work}
Similar to our work on PGs, succinct data structures that support the considered set of queries 
have been presented for \emph{chordal graphs}~\cite{MunroWu2018} and 
\emph{interval graphs}~\cite{AcanChakrabortyJoRao2020,HeMunroNekrichWildWu2020arxiv}. 
The latter also consider the special class of unit/proper interval graphs and 
the generalization to circular interval graphs. 

Concurrently%
\footnote{%
	The preprint~\cite{AcanChakrabortyJoNakashimaSadakaneRao2020arxiv} (now published as~\cite{AcanChakrabortyJoNakashimaSadakaneRao2022}) appeared 
	shortly after an initial version of this article~\cite{TsakalidisWildZamaraev2020arxiv} 
	was published on arXiv.%
}
to this work,
Acan et al.~\cite{AcanChakrabortyJoNakashimaSadakaneRao2022}
presented succinct data structures for \emph{circle graphs}
(\ie, the intersection graph of chords of a (single) circle)
and related classes (specifically $k$-polygon circle graphs and trapezoid graphs).
They show space lower bounds for these classes and data structures with
asymptotically matching space usage.
Since a PG is also a circle graph, their data structures can be applied to PGs,
but this is not known for CPGs.
Superficially, their grid-based representation~\cite[Thm.\,4.4]{AcanChakrabortyJoNakashimaSadakaneRao2022} 
is similar to ours, but the construction uses a different point set 
with different properties for queries:
Acan et al. support navigational operations 
\GAdjacent, \GDegree, and \GNeighbor, 
but none of their data structures offer \GDistance or \GSPath,
which are a main technical challenge of our work.
A further difference is that for general circle graphs, 
no succinct data structures with constant query time are known,
whereas for PGs, we can use our array-based data structure,
offering constant-time support for \GAdjacent, \GNextNeighbor, \GSPathFirst, \GDistance.

\paragraph{Outline}
The rest of this paper is organized as follows.
\wref{sec:preliminaries} collects previous results on PGs and succinct data structures.
In \wref{sec:data-structures}, we describe our main result: the succinct data structures for PGs.
\ifsubmission{%
	Due to space constraints, our other results are given in the appendix;
	\wref{app:algorithms} describes how to simulate various algorithms on top of our succinct representation.
	they extend the techniques established in \wref{sec:data-structures}.
	\wref{app:bipartite} discusses our data structure for bipartite PGs,
	and \wref{app:circular} extends our approach to circular PGs.
	Finally, \wref{app:semi-distributed} introduces semi-distributed graph representations
	and our corresponding results.
}{%
	Our other results extend the techniques established in that section.
	\wref{sec:algorithms} describes how to simulate various algorithms on top of our succinct representation.
	\wref{sec:bipartite} discusses our data structure for bipartite PGs,
	and \wref{sec:circular} extends our approach to circular PGs.
	Finally, \wref{sec:semi-distributed} introduces semi-distributed graph representations
	and our corresponding results.
}%
\wref{sec:conclusion} concludes the paper.

%% file: permutation-graphs-prelims.tex
\section{Preliminaries}
\label{sec:preliminaries}

We write $[n..m]$ for $\{n,\ldots,m\}\subset \Z$ and $[n]=[1..n]$.
We use standard notation for graphs, in particular (unless stated otherwise) 
$n$ denotes the number of vertices, $m$ the number of edges.
$N(v)$ is the neighborhood of $v$, \ie, the set of vertices adjacent to $v$.
In a directed graph $G=(V,E)$, we distinguish out-neigborhood $N^+(v) = \{u:(v,u)\in E\}$
and in-neigborhood $N^-(v) = \{u:(u,v)\in E\}$ of a vertex $v\in V$.
The complement graph of $G$ is denoted by $\overline G$.
We use the ``Iverson bracket'' notation: $[\mathit{cond}]$ is $1$ if $\mathit{cond}$ is true and $0$ otherwise.

\subsection{Permutation Graphs}

It is easy to see from the intersection model of a PG $G$ (as intersections
of chords between parallel lines) that only the relative order of upper (resp.\ lower) endpoints of the chords
are relevant (cf. \wref{fig:example-small}).
Hence, a graph $G$ is a PG if there exists a permutation~$\pi$
and a bijection between the vertices of $G$ and the elements of $\pi$,
such that two vertices are adjacent if and only if the corresponding elements are reversed by $\pi$;
that explains the name.

\ifsubmission{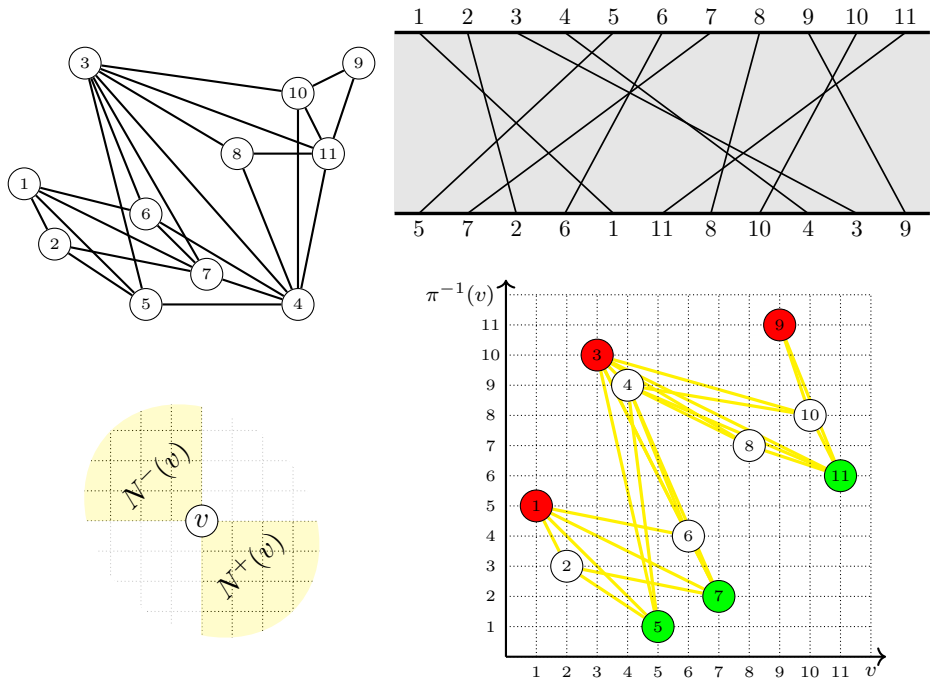
\begin{figure}[tb]}{\begin{figure}[htb]}
	\plaincenter{%
	\iflipics{%
		\begin{adjustbox}{max width=\linewidth,minipage=1.3\linewidth}%
	}{}%
	\ifkoma{%
		\begin{adjustbox}{varwidth=\linewidth,max width=\linewidth}%
	}{}%
	\ifspringer{%
		\begin{adjustbox}{minipage=1.3\linewidth,max width=\linewidth}%
	}{}%
	\plaincenter{%
		\begin{tikzpicture}[
				baseline=(v2),
				scale=.4,
				point/.style = {draw,fill=white,circle,minimum size=12pt,inner sep=0pt,font=\tiny},
				A point/.style = {point},
				B point/.style = {point},
				AB point/.style = {point},
				graph edge/.style = {thick},
			]
			\node[A point] (v1)  at (1,5) {1} ;
			\node[point]   (v2)  at (2,3) {2} ;
			\node[A point] (v3)  at (3,9) {3} ;
			\node[point]   (v4)  at (10,1) {4} ;
			\node[B point] (v5)  at (5,1) {5} ;
			\node[point]   (v6)  at (5,4) {6} ;
			\node[B point] (v7)  at (7,2) {7} ;
			\node[point]   (v8)  at (8,6) {8} ;
			\node[A point] (v9)  at (12,9) {9} ;
			\node[point]   (v10) at (10,8) {10} ;
			\node[B point] (v11) at (11,6) {11} ;		\begin{pgfonlayer}{background}
				\draw[graph edge] (v1) to (v2) ; 
				\draw[graph edge] (v1) to (v5) ; 
				\draw[graph edge] (v1) to (v6) ; 
				\draw[graph edge] (v1) to (v7) ; 
				\draw[graph edge] (v2) to (v5) ; 
				\draw[graph edge] (v2) to (v7) ; 
				\draw[graph edge] (v3) to (v4) ; 
				\draw[graph edge] (v3) to (v5) ; 
				\draw[graph edge] (v3) to (v6) ; 
				\draw[graph edge] (v3) to (v7) ; 
				\draw[graph edge] (v3) to (v8) ; 
				\draw[graph edge] (v3) to (v10) ; 
				\draw[graph edge] (v3) to (v11) ; 
				\draw[graph edge] (v4) to (v5) ; 
				\draw[graph edge] (v4) to (v6) ; 
				\draw[graph edge] (v4) to (v7) ; 
				\draw[graph edge] (v4) to (v8) ; 
				\draw[graph edge] (v4) to (v10) ; 
				\draw[graph edge] (v4) to (v11) ; 
				\draw[graph edge] (v6) to (v7) ; 
				\draw[graph edge] (v8) to (v11) ; 
				\draw[graph edge] (v9) to (v10) ; 
				\draw[graph edge] (v9) to (v11) ; 
				\draw[graph edge] (v10) to (v11) ; 
			\end{pgfonlayer}
		\end{tikzpicture}
		\hfill
		\scalebox{.8}{
	\begin{tikzpicture}[
			xscale=.8,
			yscale=3,
			normal/.style = {thick,black},
			forward/.style = {thick,green!50!black},
			backward/.style = {thick,orange!50!black},
			point/.style = {draw,fill=white,circle,minimum size=12pt,inner sep=0pt,font=\tiny},
		]
		\fill[black!10] (0.5,0) rectangle ++(11,-1) ;
		\draw[normal] (1,0) to (5,-1) ;
		\node[above] at (1,0) {1} ;
		\node[below] at (5,-1) {1} ;
		\draw[normal] (2,0) to (3,-1) ;
		\node[above] at (2,0) {2} ;
		\node[below] at (3,-1) {2} ;
		\draw[normal] (3,0) to (10,-1) ;
		\node[above] at (3,0) {3} ;
		\node[below] at (10,-1) {3} ;
		\draw[normal] (4,0) to (9,-1) ;
		\node[above] at (4,0) {4} ;
		\node[below] at (9,-1) {4} ;
		\draw[normal] (5,0) to (1,-1) ;
		\node[above] at (5,0) {5} ;
		\node[below] at (1,-1) {5} ;
		\draw[normal] (6,0) to (4,-1) ;
		\node[above] at (6,0) {6} ;
		\node[below] at (4,-1) {6} ;
		\draw[normal] (7,0) to (2,-1) ;
		\node[above] at (7,0) {7} ;
		\node[below] at (2,-1) {7} ;
		\draw[normal] (8,0) to (7,-1) ;
		\node[above] at (8,0) {8} ;
		\node[below] at (7,-1) {8} ;
		\draw[normal] (9,0) to (11,-1) ;
		\node[above] at (9,0) {9} ;
		\node[below] at (11,-1) {9} ;
		\draw[normal] (10,0) to (8,-1) ;
		\node[above] at (10,0) {10} ;
		\node[below] at (8,-1) {10} ;
		\draw[normal] (11,0) to (6,-1) ;
		\node[above] at (11,0) {11} ;
		\node[below] at (6,-1) {11} ;
		\draw[ultra thick] (0.5,0) -- ++(11,0);
		\draw[ultra thick] (0.5,-1) -- ++(11,0);
		
	\end{tikzpicture}
	}
	}
	\\[-4ex]
	\plaincenter{
	\begin{tikzpicture}[
			scale=.4,
			point/.style = {draw,fill=white,circle,minimum size=12pt,inner sep=0pt,font=\tiny},
			A point/.style = {point,fill=red},
			B point/.style = {point,fill=green},
			AB point/.style = {point,fill=cyan},
			graph edge/.style = {yellow,very thick},
		]
		\draw[densely dotted] (0,0) grid (12,12) ;
		\draw[->,thick] (0,0) -- (12.500000,0) ;
		\draw[->,thick] (0,0) -- (0,12.500000) ;
			\node at (1,-.5) {\tiny 1} ;
			\node at (2,-.5) {\tiny 2} ;
			\node at (3,-.5) {\tiny 3} ;
			\node at (4,-.5) {\tiny 4} ;
			\node at (5,-.5) {\tiny 5} ;
			\node at (6,-.5) {\tiny 6} ;
			\node at (7,-.5) {\tiny 7} ;
			\node at (8,-.5) {\tiny 8} ;
			\node at (9,-.5) {\tiny 9} ;
			\node at (10,-.5) {\tiny 10} ;
			\node at (11,-.5) {\tiny 11} ;
			\node at (12,-.5) {\scriptsize $v$} ;
			\node at (-.5,1) {\tiny 1} ;
			\node at (-.5,2) {\tiny 2} ;
			\node at (-.5,3) {\tiny 3} ;
			\node at (-.5,4) {\tiny 4} ;
			\node at (-.5,5) {\tiny 5} ;
			\node at (-.5,6) {\tiny 6} ;
			\node at (-.5,7) {\tiny 7} ;
			\node at (-.5,8) {\tiny 8} ;
			\node at (-.5,9) {\tiny 9} ;
			\node at (-.5,10) {\tiny 10} ;
			\node at (-.5,11) {\tiny 11} ;
			\node[anchor=east] at (0,12) {\scriptsize $\pi^{-1}(v)$} ;
		\node[A point] (v1) at (1,5) {1} ;
		\node[point] (v2) at (2,3) {2} ;
		\node[A point] (v3) at (3,10) {3} ;
		\node[point] (v4) at (4,9) {4} ;
		\node[B point] (v5) at (5,1) {5} ;
		\node[point] (v6) at (6,4) {6} ;
		\node[B point] (v7) at (7,2) {7} ;
		\node[point] (v8) at (8,7) {8} ;
		\node[A point] (v9) at (9,11) {9} ;
		\node[point] (v10) at (10,8) {10} ;
		\node[B point] (v11) at (11,6) {11} ;
		\begin{pgfonlayer}{background}
			\draw[graph edge] (v1) to (v2) ; 
			\draw[graph edge] (v1) to (v5) ; 
			\draw[graph edge] (v1) to (v6) ; 
			\draw[graph edge] (v1) to (v7) ; 
			\draw[graph edge] (v2) to (v5) ; 
			\draw[graph edge] (v2) to (v7) ; 
			\draw[graph edge] (v3) to (v4) ; 
			\draw[graph edge] (v3) to (v5) ; 
			\draw[graph edge] (v3) to (v6) ; 
			\draw[graph edge] (v3) to (v7) ; 
			\draw[graph edge] (v3) to (v8) ; 
			\draw[graph edge] (v3) to (v10) ; 
			\draw[graph edge] (v3) to (v11) ; 
			\draw[graph edge] (v4) to (v5) ; 
			\draw[graph edge] (v4) to (v6) ; 
			\draw[graph edge] (v4) to (v7) ; 
			\draw[graph edge] (v4) to (v8) ; 
			\draw[graph edge] (v4) to (v10) ; 
			\draw[graph edge] (v4) to (v11) ; 
			\draw[graph edge] (v6) to (v7) ; 
			\draw[graph edge] (v8) to (v11) ; 
			\draw[graph edge] (v9) to (v10) ; 
			\draw[graph edge] (v9) to (v11) ; 
			\draw[graph edge] (v10) to (v11) ; 
		\end{pgfonlayer}

			\begin{scope}[reset cm,shift={(-4,1.8)},scale=.4]
			\clip[rotate=-45] (0,0) circle (11em and 9em);
			\draw[densely dotted] (-5,-5) grid ++(12,12) ; ;
			\fill[white,opacity=.8] (0,0)++(-.01,-.01) rectangle ++(-5,-5) ;
			\fill[white,opacity=.8] (0,0)++(.01,.01) rectangle ++(5,5) ;
			\fill[yellow,opacity=.2] (0,0) rectangle ++(-5,5) ;
			\fill[yellow,opacity=.2] (0,0) rectangle ++(5,-5) ;
			\node[point] at (0,0) {\normalsize $v$} ;
			\node[rotate=45] at (-1.5,1.5) {$N^-(v)$} ;
			\node[rotate=45] at (1.6,-1.5) {$N^+(v)$} ;
		\end{scope}
	\end{tikzpicture}
	}
	\end{adjustbox}}
	\caption{%
		Example permutation graph (top left) from \cite{BazzaroGavoille2009} in different representations: 
		a representation as intersections of chords between two parallels (top right), 
		corresponding to the permutation $\pi=(5,7,2,6,1,11,8,10,4,3,9)$, and 
		the points $(v,\pi^{-1}(v))$ on a 2D grid (bottom right).
		A point in the grid can ``see'' (is adjacent to) 
		all points in the top left resp.\ lower right quadrant around it as illustrated on the bottom left~\cite{BazzaroGavoille2009}.
	}
	\label{fig:example-small}
\end{figure}

To avoid confusion in counting results, we carefully distinguish three related notions for PGs.
First, given a permutation $\pi:[n] \to [n]$, the \emph{ordered PG}
induced by $\pi$, denoted $G_\pi = (V,E)$, has vertices $V=[n]$ and (undirected) edges $\{i,j\}\in E$
for all $i>j$ with $\pi^{-1}(i)<\pi^{-1}(j)$, \ie, if $\pi$ has an \emph{inversion} $(i,j)$.
Given an ordered PG $G$, we can uniquely reconstruct the 
permutation $\pi$ with $G_\pi = G$: 
By setting $b_j$, for each vertex $j$, to
the number of its neighbors $i$ with $i>j$,
we obtain the inversion table $b_1,\ldots,b_n$ of the permutation, 
from which there is a well-known bijection to $\pi$ itself~\cite[\S5.1.1]{Knuth1998}.
Hence, ordered PGs and permutations are in bijection.
This yields a simple recognition algorithm for ordered PGs:
Compute $\pi$ as above and check if the given graph equals $G_\pi$.
The ordered PG $G_\pi$ can be characterized by its \emph{grid representation},
which is a collection of integer points in the plane associated with the vertices of $G_\pi$:
a vertex $v$ is associated with the unique point $(v, \pi^{-1}(v))$ (see \wref{fig:example-small}). 
A useful property of 
the grid representation is that the neighbors of the vertex $v$ are exactly those vertices whose
points are located in the top left or the lower right quadrant around the point of $v$.

A graph $G=([n],E)$ is a \emph{labeled PG}, written $G\in\mathcal P^n$, 
if there is set of $n$ chords between two parallel lines \emph{and} 
an assignment of vertices to chords, so that $\{i,j\}\in E$ iff chords $i$ and $j$ intersect.
In other words, $G\in \mathcal P^n$ iff there are \emph{two} permutations $\pi:[n]\to[n]$ and $\rho:[n]\to[n]$,
so that $\rho(G) = ([n],\rho(E)) = G_\pi$, 
where $\rho(E) = \bigl\{\{\rho(u),\rho(v)\} : \{u,v\} \in E\bigr\}$;
in short: $G$ is a labeled PG iff it is isomorphic to some ordered PG $G_\pi$.

The set of \emph{unlabeled PGs} of size $n$, denoted by $\mathcal P_n$,
is the family of equivalence classes of labeled graphs in $\mathcal P^n$ under (graph) isomorphism.

\ifsubmission{
	To contrast these three notions, consider a graph with a single edge. 
	This is a single unlabeled graph, 
	but it corresponds to $n-1$ ordered PGs (all $n-1$ permutations with a single inversion)
	and $\binom n2$ labeled graphs with a single edge.%
}{%
	To illustrate the notions of ordered, labeled, and unlabeled PGs, 
and to make the distinction between them clear, we consider a few simple examples.
Both the empty or complete \emph{unlabeled graph}
	correspond to a single ordered PG, 
namely with~$\pi$ the sorted (resp.\ reverse sorted) permutation.
Similarly, there is only one labeled empty or complete graph;
in this case, the three notions coincide.
However, the unlabeled graph with just a single edge
	corresponds to $n-1$ ordered PGs, namely all $n-1$ permutations with a single inversion;
and there are $\binom n2$ labeled graphs with a single edge.
	We can always select a representative (a labeled PG) 
	for an isomorphism class (the unlabeled PG)
	that is an ordered PG, but in general, there are more ordered PGs
	than unlabeled PGs.%
}

A graph is \emph{comparability} if its edges can be oriented such that if there are 
edge $(a,b)$ and $(b,c)$, then there is an edge $(a,c)$. 
We will use the following classical characterization of PGs.
\begin{theorem}[PG \& comparability, \cite{dushnik1941partially}]\label{th:permCharacterization}
	A graph $G$ is a PG if and only if both $G$ and $\overline{G}$ are comparability graphs.
\end{theorem}
Finally, for the construction of our data structures, we will assume that an \emph{ordered} 
PG $G_\pi$ is given; the following result allows to compute such from a given PG
in linear time.
\begin{theorem}[PG recognition, {\cite{McConnellSpinrad1999}}]
	There is an algorithm that 
	given a graph $G=(V,E)$ on $n=|V|$ vertices and $m=|E|$ edges
	computes in $\Oh(n+m)$ time two bijections $\pi:[n]\to[n]$ and $\rho:V\to[n]$
	with $\rho(G) = G_\pi$, or determines that $G$ is not a PG.
\end{theorem}
\ifsubmission{}{
	\subsection{Space Lower Bounds}
}
\label{sec:lower-bounds}
Recall that $\mathcal P_n$ denotes the set of unlabeled PGs on $n$ vertices.
We obtain information-theoretic lower bounds for storing an unlabeled PG from known counting results~\cite{BazzaroGavoille2009}.
\begin{corollary}
	$\lg|\mathcal P_n| \ge n\lg n - O(n \log \log n)$ bits are necessary to represent an unlabeled permutation 
	graph on $n$ vertices.
\end{corollary}
\begin{proof}
	Recall that we write $\mathcal P^n$ for the set of labeled PGs on $n$ vertices and $\mathcal P_n$
	for the set of unlabeled PGs on $n$ vertices.
	\cite[Thm.\,5.2]{BazzaroGavoille2009} shows that $\lg|\mathcal P^n| \ge 2 n \lg n - O(n \log \log n)$.
	Clearly $|\mathcal P^n| \le n! |\mathcal P_n|$ since there are at most $n!$ ways of assigning labels $[n]$.
	Using the Stirling approximation, $\lg(n!) = n \lg n - O(n)$, we obtain that
	$\lg|\mathcal P_n| \ge 2 n \lg n - O(n \log \log n) - \lg(n!) \ge n \lg n - O(n \log \log n)$.
\end{proof}
Up to lower order terms, this lower bound coincides with $\lg(n!)$,
so succinctly storing a given grid representation of an ordered PG in 
our data structures suffices for a succinct PG data structure.
Generalizing a construction from Acan et al.~\cite{AcanChakrabortyJoNakashimaSadakaneRao2022}, 
we can strengthen the above lower bound.
\begin{theorem}[Space lower bound]
\label{thm:lower-bound-refined}
	$\lg |\mathcal P_n| \ge n \lg n - O(n)$ bits are necessary to represent an unlabeled permutation 
	graph on $n$ vertices.
\end{theorem}
\begin{proof}[Proof Sketch]
We build on the proof of Thm.\ 4.2 of~\cite{AcanChakrabortyJoNakashimaSadakaneRao2022}; 
we reproduce the parts that need amendment here.
We construct a specific family of vertex-colored PGs that is large enough so that~-- 
even after discounting the overcounting due to counting \emph{colored} graphs~--
it corresponds to $2^{n \lg n - O(n)}$ distinct unlabeled PGs, yielding the claim.
We represent the colored graphs via their (colored) permutation diagram.
We begin with two parallel lines and place $n$ ``chord slots'' (points) on each line;
we will later connect these to $n$ disjoint chords.
Let $p_1,\ldots,p_n$ resp.\ $q_1,\ldots,q_n$ denote these points on the upper resp.\ lower line, 
numbered from left to right; cf.\ \wref{fig:lower-bound}.
\begin{figure}
	\plaincenter{%
	\begin{tikzpicture}[
			xscale=.5,
			yscale=3,
			dot/.style={circle,draw,fill=white,inner sep=1pt},
			every label/.style={font=\tiny},
		]
		\def\l{3}
		\def\k{4}
		\def\n{20}
		\foreach \y in {0,1} {
			\draw[thick] (0,\y) -- (\n+1,\y);
		}
		\foreach \i in {1,...,\n} {
			\node[dot,fill=black,label={90}:{$p_{\i}$}] (p\i) at (\i,1) {} ;
		}
		\foreach \i in {1,...,\n} {
			\node[dot,fill=black,label={-90}:{$q_{\i}$}] (q\i) at (\i,0) {} ;
		}
		\foreach[count=\ii] \i in {4,8,...,16} {
			\node[dot,inner sep=2pt] (pp\ii) at (\i,1) {} ;
		}
		\foreach[count=\ii] \i in {5,9,...,\n} {
			\node[dot,inner sep=2pt] (qq\ii) at (\i,0) {} ;
		}
		\foreach[count=\ii] \i in {17,...,20} {
			\node[dot,fill=gray,inner sep=2pt] (ppp\ii) at (\i,1) {} ;
		}
		\foreach[count=\ii] \i in {1,...,\k} {
			\node[dot,fill=gray,inner sep=2pt] (qqq\ii) at (\i,0) {} ;
		}
		\begin{scope}[black!10,very thick]
		\foreach \i/\j in {1/7,2/6,3/19,5/16,6/12,7/10,10/11,11/8,13/18,14/15,15/20} {
			\draw (p\i) to (q\j) ;
		}
		\end{scope}
		\begin{scope}[very thick]
		\draw[blue]         (pp1)  to node[inner sep=0pt,fill=white] {\tiny1} (qqq1);
		\draw[blue]         (pp2)  to node[inner sep=0pt,fill=white] {\tiny2} (qqq2);
		\draw[blue]         (pp3)  to node[inner sep=0pt,fill=white] {\tiny3} (qqq3);
		\draw[blue]         (pp4)  to node[inner sep=0pt,fill=white] {\tiny4} (qqq4);
		\draw[red!80!black] (ppp1) to node[inner sep=0pt,fill=white] {\tiny5} (qq1);
		\draw[red!80!black] (ppp2) to node[inner sep=0pt,fill=white] {\tiny6} (qq2);
		\draw[red!80!black] (ppp3) to node[inner sep=0pt,fill=white] {\tiny7} (qq3);
		\draw[red!80!black] (ppp4) to node[inner sep=0pt,fill=white] {\tiny8} (qq4);
		\draw (p9) to (q14) ;
		\end{scope}
		\foreach \i in {1,...,\k} {
			\draw[decoration={brace},decorate] ({1+(\i-1)*(\l+1)-.25},1.2) --node[above] {\scriptsize$A_\i$}  ++(\l-.5,0) ;
			\draw[decoration={brace,mirror},decorate] ({6+(\i-1)*(\l+1)-.25},-0.2) --node[below=1pt] {\scriptsize$B_\i$}  ++(\l-.5,0) ;
		}
	\end{tikzpicture}%
	}
	\caption{
		The colored PG construction from \wref{thm:lower-bound-refined} for $\ell=3$ and $k=4$,
		and hence $n=k\ell+2k = 20$.
		Special chords are shown in blue and red.
		The highlighted chord $(p_9,q_{14})$ intersects the special chords $[i..j] = \{3,4,5,6,7\}$ and has endpoints in
		$A_i = A_3$ and $B_{j-k} = B_{7-4}=B_3$.
	}
	\label{fig:lower-bound}
\end{figure}
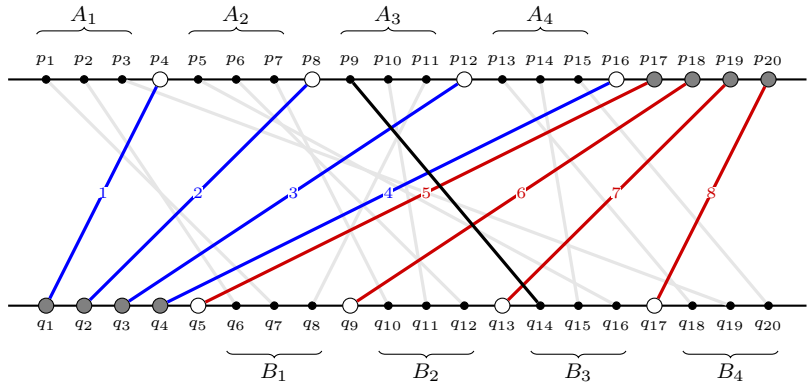
As in~\cite{AcanChakrabortyJoNakashimaSadakaneRao2022}, we fix parameters $k$ and $\ell$ so that 
$k\ell+2k = n$.
Now fix $2k$ special cords as follows:
The first $k$ special cords connect $q_1,\ldots,q_k$ to the points
$p_{\ell+1},p_{2(\ell+1)},p_{3(\ell+1)},\ldots,p_{k(\ell+1)}$,
the second $k$ special cords connect $p_{n-k+1},\ldots,p_{n}$ with 
$q_{k+1},q_{k+1+(\ell+1)},q_{k+1+2(\ell+1)},\ldots,q_n$.
Each of the $2k$ special cords is colored using a unique color in $[2k]$, assigned from left to right;
all other chords (added below) have color $0$.
We have so far used $4k$ of the chord slots; the remaining $2n-4k$ slots are partitioned by the special chords
into $2k$ intervals of $\ell$ chord slots each: $k$ on the upper line, $k$ on the lower line, each separated by 
an endpoint of a special chord. 
We name these intervals $A_1,\ldots,A_k$ on the upper line
and $B_{1},\ldots,B_{k}$ on the lower line (see \wref{fig:lower-bound}). 
We now consider \emph{matchings} of the remaining $k\ell$ slots
on the upper line with the remaining $k\ell$ slots on the lower line.
Each such matching corresponds to one way of adding the remaining $k\ell$ chords;
\wref{fig:lower-bound} shows an example (gray lines).
In general, different matchings can correspond to the same unlabeled colored graph, 
but we will see that this can only happen for \emph{bad matchings}~\cite{AcanChakrabortyJoNakashimaSadakaneRao2022}:
a matching is \emph{bad} if it contains 3 or more of chords connecting the same $A_i$ with the same $B_j$;
otherwise it is \emph{good}.
A good matching can be uniquely reconstructed from its induced colored PG: 
First, each colored vertex is unique and its color uniquely determines which special chord it corresponds to.
Next, each $0$-colored vertex must be adjacent to special chords with colors 
from a contiguous range $[i..j]\subset[2k]$;
its upper endpoint then lies on $A_{i}$ and its lower endpoint on $B_{j-k}$.
Hence we can uniquely reconstruct the intervals each chord's endpoints belong to.
Finally, if two chords $u$, $v$ both end in the same $A_i$, their relative order
is determined by whether or not they are adjacent.
Since the matching is good, there is at most one such pair $u$, $v$ where the relative order of 
endpoints on the bottom line is not already determined, so we can work out a total
order of the endpoints within $A_i$ from the colors and adjacencies.
The argument for two chords ending in the same $B_j$ is similar.
Using Lem.\ 4.1 of~\cite{AcanChakrabortyJoNakashimaSadakaneRao2022}, which shows
that for $k=n^{3/4+\epsilon}$, $\epsilon>0$ fixed,
a $1-o(1)$ fraction of all possible matchings is good,
we can now finish the proof as in~\cite{AcanChakrabortyJoNakashimaSadakaneRao2022}:
\begin{align*}
		|\mathcal P_n| \binom{n}{2k}(2k)!
	&\wrel\ge
		\text{\# PGs with $2k$ vertices assigned unique colors in $[2k]$}
\\[-1.7ex]	&\wrel\ge
		\text{\# colored PGs obtained from above construction}
\\	&\wrel\ge
		\text{\# good matchings}
\\	&\wrel=
		(1-o(1))(k\ell)! \;
	\wrel=
		(1-o(1))(n-2k)!\,;
\end{align*}
hence we have, denoting by $n^{\underline k} = \prod_{i=0}^{k-1}(n-i)$ the $k$th falling power of $n$, that
\begin{align*}
		\lg|\mathcal P_n| 
	&\wrel\ge 
		\lg((n-2k)!) + \lg(1-o(1)) - \lg(n^{\underline{2k}}) 
\\	&\wrel\ge 
		\lg((n-2k)!) - \lg(n^{2k}) - O(1)
\\	&\wrel=
		\bigl((n-2k) \lg(n(1-2k/n)) - (n-2k)\lg(e) + \Theta(\log n)\bigr) - 2k\lg(n) - o(1)
\\	&\wrel\ge
		n \lg n - \lg(e)n - O(k \log n).
\end{align*}
This concludes the proof.
\end{proof} 
\ifsubmission{}{
\begin{remarknonumber}
	We note that in the data-structures and graph-labeling-schemes communities, 
	the above approach for proving space optimality of graph representations via lower bounds on the number of unlabeled 
	graphs in the class is quite typical~\cite{GavoillePaul2008,BazzaroGavoille2009,AcanChakrabortyJoRao2020,munro2018succinct}:
	One establishes a lower bound on the number $|\mathcal X_n|$ of unlabeled graphs in a given class $\mathcal X$ by
	first deriving a lower bound on the number $|\mathcal X^n|$ of \emph{labeled} (or colored) graphs in the class, and then applying 
	the obvious relation $|\mathcal X^n| \le n! |\mathcal X_n|$ (or a similar one for partially colored graphs).
	The non-trivial part in this approach is the former one, and it usually boils down to an ad-hoc construction of
	a large family of labeled graphs.

	For leading-term estimates, a recent work of Sauermann~\cite{Sauermann2021} provides a uniform framework
	for deriving tight lower bounds on the number of labeled graphs in any \emph{semi-algebraic} graph class.
The family of semi-algebraic graph classes contains many geometric intersection graphs classes, including
interval graphs and PGs. %

\end{remarknonumber}
}

\subsection{Succinct Data Structures}

For the reader's convenience, we collect used results
on succinct data structures here.
First, we cite the compressed bit vectors of Pătrașcu~\cite{Patrascu2008}.

\begin{lemma}[Compressed bit vector]
\label{lem:compressed-bit-vectors}
	Let $B[1..n]$ be a bit vector of length~$n$, containing $m$ $1$-bits.
	For any constant $c>0$, there is a data structure using
	\(
			\lg \binom{n}{m} \wbin+ O\bigl(\frac{n}{\log^c n}\bigr)
		\wwrel\le 
			m \lg \bigl(\frac nm\bigr) \wbin+ O\bigl(\frac{n}{\log^c n}+m\bigr)
	\)
	bits of space that
	supports in $O(1)$ time operations 
	(for $i \in [1,n]$):
	\begin{enumerate}
		\item $\accessop(B, i)$: return $B[i]$, the bit at index $i$ in $B$.
		\item $\rankop_\alpha(B, i)$: return the number of bits with
		value $\alpha \in \{0,1\}$ in $B[1..i]$.
		\item $\selop_\alpha(B, i)$: return the index of the $i$-th
		bit with value $\alpha \in \{0,1\}$.
	\end{enumerate}
\end{lemma}

\ifsubmission{}{
	\begin{remarknonumber}[Simpler bitvectors]
		The result of P\u atra\textcommabelow scu has the best theoretical guarantees, 
		but requires rather complicated data structures.
		Compressed bitvectors with space
		\begin{align*}
				\lg \binom{n}{m} \wbin+ O\biggl(\frac{n\log \log n}{\log n}\biggr)
			&\wwrel\le 
				n H\Bigl(\frac{m}{n}\Bigr) \wbin+ O\biggl(\frac{n\log \log n}{\log n}\biggr) 
		\\	&\wwrel= 
				m \lg \Bigl(\frac nm\Bigr) \wbin+ O\biggl(\frac{n \log \log n}{\log n}+m\biggr)
		\end{align*}
		have been proposed by Raman, Raman, and Rao~\cite{RamanRamanRao2007}
		and implemented~\cite{GogBellerMoffatPetri2014}.
		For our application, indeed a plain (uncompressed) bitvector 
		with $\Oh(1)$-time support for rank and select and using $n+O(n/\log\log n)$ bits of space is sufficient
		(see, \eg,~\cite[\S4.2.2 \& \S4.3.3]{Navarro2016}, originally proposed in~\cite{Jacobson1989,Clark1996}).
	\end{remarknonumber}
}

Using wavelet trees, based on above bitvectors, we can also handle non-binary arrays.
\begin{lemma}[Wavelet trees for constant $\sigma$]
\label{lem:wavelet-trees}
	Let $S[1..n]$ be a static array with entries $S[i]\in\Sigma = [1..\sigma]$
	for $\sigma$ a fixed constant.
	There is a data structure using $\lg(\sigma) n + o(n)$ bits of space
	that supports the following queries in $O(\log \sigma) = O(1)$ time 
	(without access to $S$ at query time)
	\begin{enumerate}
		\item $\accessop(S, i)$: return $S[i]$, the symbol at index $i$ in $S$.
		\item $\rankop_\alpha(S, i)$: return the number of indices with
		value $\alpha \in \Sigma$ in $S[1..i]$.
		\item $\selop_\alpha(S, i)$: return the index of the $i$-th
		occurrence of value $\alpha \in \Sigma$ in $S$.
	\end{enumerate}
\end{lemma}
\begin{proof}
Wavelet trees~\cite[\S6.2]{Navarro2016} support these operations in the stated time.
For the case of a small fixed $\sigma$ that we need, we can use 
a separate compressed bitvector (\wref{lem:compressed-bit-vectors}) for each of the $O(\sigma)$ nodes in the wavelet tree. By the aggregation property of the entropy,
the overall space is bounded by $n H_0 + o(\sigma n) \le n \lg(\sigma) + o(n)$,
where $H_0$ is the (zeroth-order) empirical entropy of $S$ (see, \eg, \cite[\S6.2.4]{Navarro2016}).
\end{proof}

Given an array $A[1..n]$ of comparable elements, \ifsubmission{}{(\eg, numbers),}
range-minimum queries (resp.\ range-maximum queries) are defined for $1\le i\le j\le n$ by
\begin{align*}
		\RangeMin_A(i,j)
	&\wwrel=
		\mathop{\arg\min}\limits_{i\le k\le j} A[k],
\ifsubmission{&}{\\}
		\RangeMax_A(i,j)
	&\wwrel=
		\mathop{\arg\max}\limits_{i\le k\le j} A[k].
\end{align*}
\ifsubmission{\vspace{-2ex}}{%
In both cases, ties are broken by the index, \ie, the index of the 
leftmost minimum (resp.\ maximum) is returned.
}%

\begin{lemma}[{{RMQ index, \cite[Thm.\,3.7]{FischerHeun2011}}}]
\label{lem:rmq-indexing}
	For any constant $\epsilon>0$ the following holds.
	Given a static array $A[1..n]$ of comparable elements,
	there is a data structure using $\epsilon n$ bits of space on top of $A$ 
	that answers range-minimum queries in $O(1/\epsilon)$ time 
	(making as many queries to $A$).
\end{lemma}

\ifsubmission{}{%
	\noindent Clearly, the same data structure can also be used to answer range-maximum queries by
building the data structure \wrt the reverse ordering.
}

\ifsubmission{}{
	\begin{remarknonumber}[Sublinear RMQ]
		Indeed, $\epsilon$ can be chosen smaller than constant, yielding sublinear extra space,
		at the cost of increasing the query time to superconstant;
		we only need $\epsilon = \Omega(n^{-1+\delta})$ for some $\delta>0$.
	\end{remarknonumber}
}

\ifsubmission{}{%
Given a static set of points in the plane, \emph{orthogonal range reporting} asks
to find all points in the point set that lie inside a query rectangle $[x_1,x_2]\times[y_1,y_2]$.
Range \emph{counting} queries only report the number of such points.

\begin{lemma}[{{Succinct point grids, \cite[Thm.\,1]{BoseHeMaheshwariMorin2009}}}]
\label{lem:grid-bose-et-al}
	A set $N$ of $n$ points in an $n\times n$ integer grid can be represented
	using $n\lg n + o(n\log n)$ bits of space so that 
	\begin{enumerate}
	\item orthogonal-range-counting queries
	are answered in $\Oh(\log n / \log\log n)$ time, and 
	\item orthogonal-range-reporting queries are answered in $\Oh((k+1) \log n / \log \log n)$ time,
	where~$k$ is the output size.
	\end{enumerate}
\end{lemma}
}

\begin{lemma}[Permutation grid]
\label{lem:permutation-grid}
	Given a permutation $\pi:[n]\to[n]$, we can represent the point set 
	$P=P(\pi) = \{(x,\pi(x)):x\in[n]\}$ using $n\lg n + o(n\log n)$ bits of space so that 
	we can answer the following queries:
	\begin{enumerate}
	\item 
		\ifsubmission{}{orthogonal-range-counting queries,\/ }
		$\GridCount_P(x_1,x_2;y_1,y_2) = \bigl|P \cap [x_1,x_2]\times[y_1,y_2]\bigr|$ 
		in $\Oh(\log n / \log\log n)$ time;
	\item 
		\ifsubmission{}{orthogonal-range-reporting queries,\/ }
		$\GridReport_P(x_1,x_2;y_1,y_2) = P \cap [x_1,x_2]\times[y_1,y_2]$ 
		in $\Oh((k+1)\log n / \log\log n)$ time, where $k=\GridCount_P(x_1,x_2;y_1,y_2)$;
	\item 
		\ifsubmission{}{application of $\pi$,} 
		$\GridYForX_P(\pi)(x) = \pi(x)$ for $x\in[n]$
		in $\Oh(\log n / \log\log n)$ time;
	\item 
		\ifsubmission{}{inverse of $\pi$, }
		$\GridXForY_{P(\pi)}(y) = \pi^{-1}(y)$ for $y\in[n]$
		in $\Oh(\log n / \log\log n)$ time.
	\end{enumerate}
\end{lemma}
\begin{proof}
\ifsubmission{%
	This follows immediately using the succinct point-grid data structure
	of~\cite[Thm.\,1]{BoseHeMaheshwariMorin2009}, which answers range queries
	in $\Oh(\log n / \log\log n)$ time per vertex and range counting queries in the same time.
}{%
We use the grid data structure from \wref{lem:grid-bose-et-al} on $P$;
counting and reporting queries are immediate, and for others we use that
$\GridYForX_{P(\pi)}(x) = \GridReport_P(x,x;1,n).\texttt y$ and 
$\GridXForY_{P(\pi)}(y) = \GridReport_P(1,n;y,y).\texttt x$.
Here we write $Q.\texttt x$ to denote the projection of point set $Q$ to the $x$-coordinates
of the points.
}%
\end{proof}

\ifsubmission{}{
	\begin{remarknonumber}[Iterate over range]
	It is not clear if we can iterate over the result of \GridReport 
	with $\Oh(\log n / \log\log n)$ time per point instead of obtaining all points
		in one~go.
	\end{remarknonumber}
}

\ifsubmission{}{
	\begin{remarknonumber}[Simpler alternatives]
		At the slight expense of increasing running times by a $\Oh(\log\log n)$ factor,
		we can replace the grid data structure by a wavelet tree, which is likely to be
		favorable for an implementation~\cite{Navarro2016,AcanChakrabortyJoRao2020}.
	\end{remarknonumber}
}

A last ingredient for our data structures is a recent result
on succinct distance oracles for proper interval graphs.
Here, an interval graph is the intersection graph of a set of intervals on the real line,
and a proper interval graph is one that has an interval realization where no interval strictly
contains another one.

\begin{lemma}[{{Succinct proper interval graphs~\cite[Thm.\,12]{HeMunroNekrichWildWu2020arxiv}}}]
\label{lem:proper-interval-graphs}
	A proper interval graph on $n$ vertices can be represented in $3n + o(n)$ bits of space
	so that $\GDistance(u,v)$ for $u,v\in[n]$ can be computed in $\Oh(1)$ time, and 
	vertices are identified by the rank of the left endpoints of their interval in some
	realization of the proper interval graph.
	We can also answer \GAdjacent, \GDegree, \GNeighbor in $\Oh(1)$ time and
	$\GSPath(u,v)$ in $\Oh(\GDistance(u,v))$ time.
	For connected graphs, the space can be reduced to $2n+o(n)$ bits.
\end{lemma}

\ifsubmission{}{
	\begin{remarknonumber}[$O(1)$ time neighborhood]
	It might sound impossible to do \GNeighbor in constant time independent
	of the output size; this is possible in proper interval graphs since neighborhoods
	are contiguous intervals (of vertex labels) and thus can be encoded implicitly 
	in a constant number of words.
	\end{remarknonumber}
}

\begin{remarknonumber}[Routing]
	By inspection of the proof, the data structure 
	from~\cite{HeMunroNekrichWildWu2020arxiv} can also support 
	$\GSPathFirst(u,v)$ in constant time.
	Thus, not just can $\GSPath(u,v)$ be answered in optimal overall time, 
	but we can output the path step by step in optimal time per edge.
\end{remarknonumber}

%% file: permutation-graphs-ds.tex

\section{Data Structures for Permutation Graphs}
\label{sec:data-structures}

In this section, we assume a permutation $\pi:[n]\to[n]$ is given and
we describe how to answer queries on $G_\pi$,
\ie, we describe our data structures for \emph{ordered} PGs.
We present two approaches:
the first solution uses a grid data structure that can support all queries, 
albeit with superconstant running time;
the second solution stores $\pi$ as an array and achieves optimal query times
for all operations except \GDegree.
Our formal result is as follows.

\begin{theorem}[Succinct PG]
	\label{thm:main}
	A permutation graph can be represented 
	\begin{thmenumerate}{thm:main}\sloppy
		\item \label{thm:main-grid}
		using $n\lg n + o(n \lg n)$ bits of space
		while supporting \GAdjacent, \GDegree, \GDistance, \GSPathFirst in $\Oh(\log n / \log\log n)$ time, 
		$\GNeighbor(v)$ in $O((\GDegree(v)+1) \cdot \log n/\log\log n)$ time,
		and $\GSPath(u,v)$ in $O((\GDistance(u,v)+1)\log n/\log\log n)$ time; or
		\item \label{thm:main-dist-oracle}
		using $n\lg n + (6.17+\epsilon)n + o(n)$ bits of space (for any constant $\epsilon>0$)
		while supporting \GAdjacent, \GDistance, \GSPathFirst, \GNextNeighbor in $\Oh(1)$ time,
		$\GNeighbor(v)$, $\GDegree(v)$ in $O(\GDegree(v)+1)$ time,
		and $\GSPath(u,v)$ in $O(\GDistance(u,v)+1)$ time.
		The time for $\GNextNeighbor(v)$ is amortized $O(1)$ over
		iterating through $\GNeighbor(v)$.
	\end{thmenumerate}
\end{theorem}

\subsection{Grid-Based Data Structure}
\label{sec:ds-grid-based}

We first present the simpler grid-based data structure.
Here, we store $P(\pi) = \{(v,\pi^{-1}(v)):v\in[n]\}$ 
in the data structure of \wref{lem:permutation-grid} and
identify vertices with the $x$-coordinates of these points 
(the rank of the vertex' chord endpoint on the upper line).

\ifsubmission{First consider \GAdjacent.}{\paragraph{Adjacent}}
Given two vertices $u$ and $v$, \withoutlossofgenerality $u < v$. 
We compute $\pi^{-1}(u) = \GridYForX(u)$ and $\pi^{-1}(v) = \GridYForX(v)$;
then $\GAdjacent(u,v) = [\pi^{-1}(u) > \pi^{-1}(v)]$.

\ifsubmission{Next, we show how to implement \GNeighbor.}{\paragraph{Neighborhood}}
We separate the neighbors of a vertex $v$ into $\GNeighbor(v) = N^-(v)\cup N^+(v)$ where
$N^-(v) = \GNeighbor(v)\cap [1..v-1]$ and 
$N^+(v) = \GNeighbor(v)\cap [v+1..n]$.
Using the graphical representation of neighborhoods from \wref{fig:example-small}, 
we immediately obtain
$N^-(v) = \GridReport(1,v-1;\GridYForX(v),n)$ and 
$N^+(v) = \GridReport(v+1,n;1,\GridYForX(v))$.

\ifsubmission{Finally, consider \GDegree.}{\paragraph{Degree}}
Replacing the range-\emph{reporting} queries from \GNeighbor by range-\emph{counting} queries
yields $\GDegree(v) = |N^-(v)| + |N^+(v)|$.

\subsection{Array-Based Data Structure}
\label{sec:ds-array-based}

To improve the query time, we now give an alternative representation.
A key observation is that we never compute $\pi$; only $\pi^{-1}$ is needed.
Hence we simply store an array $\Pi[1..n]$ with $\Pi[i] = \pi^{-1}(i)$
using $n \lceil\lg n\rceil \le n \lg n + n$ bits of space.
At the expense of a slightly more complicated data structure, 
one can improve this space usage to $\lceil n \lg n\rceil = n \lg n + O(1)$
using the techniques of Dodis et al.~\cite{DodisParascuThorup2010},
still retaining access to $\Pi$ in constant time.
For legibility, we continue to write $\pi^{-1}(i)$ in operations, but it is understood that
this is indeed an access to $\Pi[i]$.

\ifsubmission{}{\par\paragraph{Adjacency}}
\GAdjacent queries only use $\pi^{-1}$, and thus they are solved exactly as above.

\paragraph{Neighborhood}
Like in our previous approach, we separately handle 
the neighbors $u$ of $v$ with $u<v$ (in $N^-(v)$) and with $u>v$ (in $N^+(v)$).
Even though we do not explicitly store the point set $P(\pi)$ in our data structure,
we can still answer the above range queries, 
because these are effectively two-sided range queries\ifsubmission{}{ (dominance queries)}:

For $N^-(v) = \GridReport(1,v-1;\pi^{-1}(v),n)$,
we maintain the range-maximum index from \wref{lem:rmq-indexing}
on $\Pi[1..n]$ using $\epsilon n$ bits of space.
We can then iterate through the vertices in $N^-(v)$ 
using the standard algorithm for three-sided orthogonal range reporting 
\ifsubmission{%
	using priority search trees.
}{%
	that uses priority search trees:
We compute $i = \RangeMax_\Pi(1,v-1)$; if $\pi^{-1}(i) \ge \pi^{-1}(v)$, we report $i$ as a neighbor
and recursively continue in the ranges $[1..i-1]$ and $[i+1..v-1]$. 
Otherwise, if $\pi^{-1}(i) < \pi^{-1}(v)$, we terminate the recursion.
(We also terminate recursive calls on empty ranges).
Each recursive call only takes constant time and either terminates or outputs a new neighbor of $v$,
so we can iterate through $N^-(v)$ with constant amortized time per vertex.
}

For $N^+(v) = \GridReport(v+1,n;1,\pi^{-1}(v))$, 
we use the same technique, reflected: we store a range-minimum index on $\Pi[1..n]$,
starting with the range $[v+1,n]$ and continue as long as the returned minimum is 
at most $\pi^{-1}(v)$.

\ifsubmission{}{\paragraph{Next neighbor}}

The above method can easily be used to \emph{iterate} over neighbors one at a time, instead
of generating and returning the full neighborhood. 
The order of iteration is implementation-defined (ultimately by the RMQ index),
but fixed for any $G_\pi$.
An easy argument shows that reporting the $k$th neighbor with the above algorithm can 
take $\Theta(k)$ time,
but amortized over the entire neighborhood of a vertex, iteration takes constant time per neighbor.
However, if done naively, it would require $O(k)$ extra working space to store the $k$ ranges 
wherein the $k$th neighbor might be found.

\ifsubmission{%
	Using a refined algorithm,
	we can reduce the extra space to constant.
	Moreover, we can support starting at an arbitrary given neighbor $w$
	and find $\GNextNeighbor(v,w)$ in the traversal.
	Basically, we can use queries on the Cartesian tree to simulate the 
	three-sided range search algorithm with out an explicit recursion.
	Due to space constraints, details are deferred to \wref{app:rmq-next-neighbor}.
}{%
	\par
	We can improve the extra space to $O(1)$ (words) 
	and support starting at an arbitrary given neighbor $w$
	to find $\GNextNeighbor(v,w)$ in the traversal.
	\input{permutation-graphs-rmq-next-neighbor}
}
\ifsubmission{}{

\begin{remark}[Easy degrees]
		We can compute $\GDegree(v)$ as $|\GNeighbor(v)|$ in $\Oh(\GDegree(v)+1)$ time,
		but this is not particularly efficient for high-degree vertices.
	We can obviously also add support for \GDegree in constant time by storing the degrees 
	of all vertices explicitly in an array. 
	This occupies an additional $n \lceil\lg n\rceil$ bits of space and is thus not succinct,
	but might in implementations be preferable to the grid data structure (and offers all queries
		in optimal time complexity).
\end{remark}
}

\subsection{Distance and Shortest Paths}
\label{sec:ds-distance}

Both of the above data structures can be augmented to support distance and shortest-path queries;
the only difference will be the running time to compute~$\pi^{-1}(v)$.
\ifsubmission{}{\par}
For that, we follow the idea of~\cite{BazzaroGavoille2009};
we sketch their approach here and give a more formal definition below.
A shortest path from $u$ to $v$ in a PG can always be found
using only \emph{left-to-right maxima} (``type $A$'' vertices) and 
\emph{right-to-left minima} (``type $B$'' vertices) 
as intermediate vertices; moreover, these are strictly alternating.
Hence, after removing an initial segment of at most 2 edges on either end of the path,
such a shortest path has either type $A(BA)^*A$ or $B(AB)^*B$.
For example, a shortest path from vertex 15 to vertex 25 in \wref{fig:example}
is $15$--$\color{red}14$--$\color{green}23$--$\color{red}22$--$25$.
Finally, how many intermediate $B$-vertices are needed to move from one $A$ vertex to another
is captured by a \emph{proper interval graph} $G_A$, and likewise for $B$-vertices in $G_B$.
We can hence reduce the shortest-path queries to proper interval graphs and 
use \wref{lem:proper-interval-graphs}. We present the details below. 

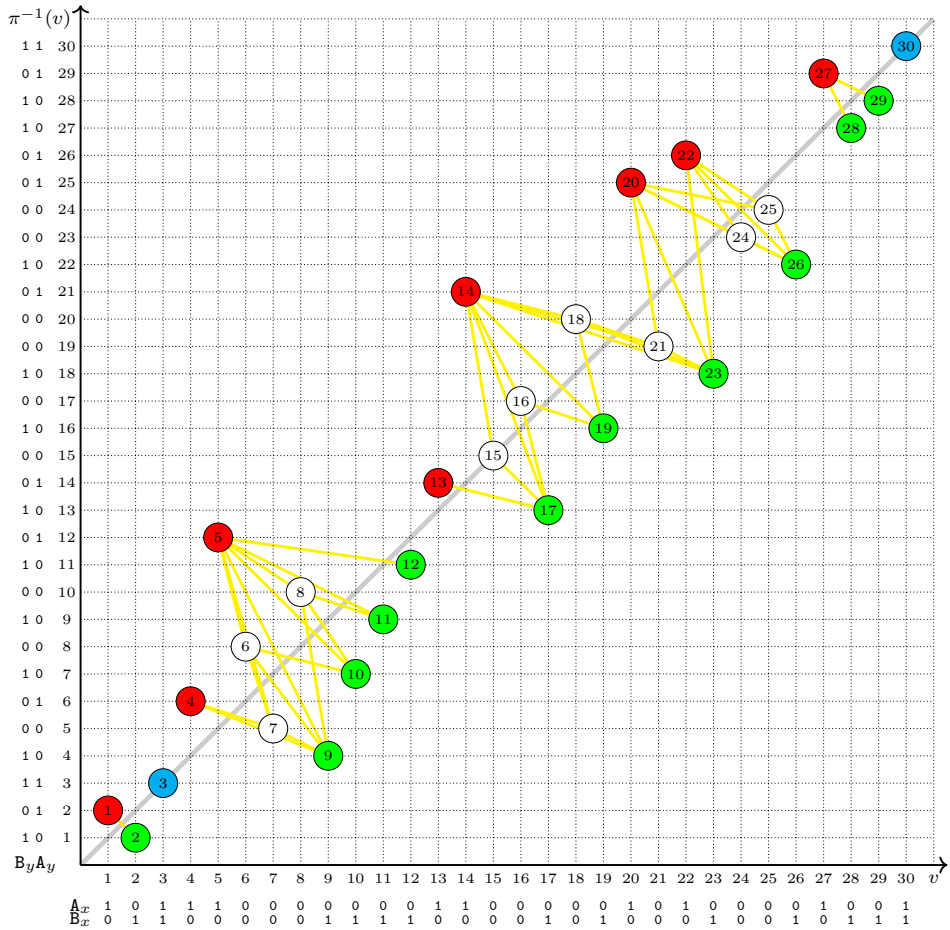
\begin{figure}[htb]
	\def\scaledwidth{.8\linewidth}%
	\iflipics{\def\picwidth{.8\linewidth}}{}%
	\ifkoma{\def\picwidth{\linewidth}}{}%
	\ifspringer{\def\picwidth{.8\linewidth}}{}%
	\ifspringer{\def\scaledwidth{.9\linewidth}}{}%
	\plaincenter{\adjustbox{max width=\scaledwidth}{\plaincenter{%
		\begin{tikzpicture}[
				scale=.4,
				point/.style = {draw,fill=white,circle,minimum size=12pt,inner sep=0pt,font=\tiny},
				A point/.style = {point,fill=red},
				B point/.style = {point,fill=green},
				AB point/.style = {point,fill=cyan},
				graph edge/.style = {yellow,very thick},
			]
			\draw[densely dotted] (0,0) grid (31,31) ;
			\draw[->,thick] (0,0) -- (31.500000,0) ;
			\draw[->,thick] (0,0) -- (0,31.500000) ;
				\node at (1,-.5) {\tiny 1} ;
				\node at (2,-.5) {\tiny 2} ;
				\node at (3,-.5) {\tiny 3} ;
				\node at (4,-.5) {\tiny 4} ;
				\node at (5,-.5) {\tiny 5} ;
				\node at (6,-.5) {\tiny 6} ;
				\node at (7,-.5) {\tiny 7} ;
				\node at (8,-.5) {\tiny 8} ;
				\node at (9,-.5) {\tiny 9} ;
				\node at (10,-.5) {\tiny 10} ;
				\node at (11,-.5) {\tiny 11} ;
				\node at (12,-.5) {\tiny 12} ;
				\node at (13,-.5) {\tiny 13} ;
				\node at (14,-.5) {\tiny 14} ;
				\node at (15,-.5) {\tiny 15} ;
				\node at (16,-.5) {\tiny 16} ;
				\node at (17,-.5) {\tiny 17} ;
				\node at (18,-.5) {\tiny 18} ;
				\node at (19,-.5) {\tiny 19} ;
				\node at (20,-.5) {\tiny 20} ;
				\node at (21,-.5) {\tiny 21} ;
				\node at (22,-.5) {\tiny 22} ;
				\node at (23,-.5) {\tiny 23} ;
				\node at (24,-.5) {\tiny 24} ;
				\node at (25,-.5) {\tiny 25} ;
				\node at (26,-.5) {\tiny 26} ;
				\node at (27,-.5) {\tiny 27} ;
				\node at (28,-.5) {\tiny 28} ;
				\node at (29,-.5) {\tiny 29} ;
				\node at (30,-.5) {\tiny 30} ;
				\node at (31,-.5) {\scriptsize $v$} ;
				\node at (-.5,1) {\tiny 1} ;
				\node at (-.5,2) {\tiny 2} ;
				\node at (-.5,3) {\tiny 3} ;
				\node at (-.5,4) {\tiny 4} ;
				\node at (-.5,5) {\tiny 5} ;
				\node at (-.5,6) {\tiny 6} ;
				\node at (-.5,7) {\tiny 7} ;
				\node at (-.5,8) {\tiny 8} ;
				\node at (-.5,9) {\tiny 9} ;
				\node at (-.5,10) {\tiny 10} ;
				\node at (-.5,11) {\tiny 11} ;
				\node at (-.5,12) {\tiny 12} ;
				\node at (-.5,13) {\tiny 13} ;
				\node at (-.5,14) {\tiny 14} ;
				\node at (-.5,15) {\tiny 15} ;
				\node at (-.5,16) {\tiny 16} ;
				\node at (-.5,17) {\tiny 17} ;
				\node at (-.5,18) {\tiny 18} ;
				\node at (-.5,19) {\tiny 19} ;
				\node at (-.5,20) {\tiny 20} ;
				\node at (-.5,21) {\tiny 21} ;
				\node at (-.5,22) {\tiny 22} ;
				\node at (-.5,23) {\tiny 23} ;
				\node at (-.5,24) {\tiny 24} ;
				\node at (-.5,25) {\tiny 25} ;
				\node at (-.5,26) {\tiny 26} ;
				\node at (-.5,27) {\tiny 27} ;
				\node at (-.5,28) {\tiny 28} ;
				\node at (-.5,29) {\tiny 29} ;
				\node at (-.5,30) {\tiny 30} ;
				\node[anchor=east] at (0,31) {\scriptsize $\pi^{-1}(v)$} ;
			\node[A point] (v1) at (1,2) {1} ;
			\node[B point] (v2) at (2,1) {2} ;
			\node[AB point] (v3) at (3,3) {3} ;
			\node[A point] (v4) at (4,6) {4} ;
			\node[A point] (v5) at (5,12) {5} ;
			\node[point] (v6) at (6,8) {6} ;
			\node[point] (v7) at (7,5) {7} ;
			\node[point] (v8) at (8,10) {8} ;
			\node[B point] (v9) at (9,4) {9} ;
			\node[B point] (v10) at (10,7) {10} ;
			\node[B point] (v11) at (11,9) {11} ;
			\node[B point] (v12) at (12,11) {12} ;
			\node[A point] (v13) at (13,14) {13} ;
			\node[A point] (v14) at (14,21) {14} ;
			\node[point] (v15) at (15,15) {15} ;
			\node[point] (v16) at (16,17) {16} ;
			\node[B point] (v17) at (17,13) {17} ;
			\node[point] (v18) at (18,20) {18} ;
			\node[B point] (v19) at (19,16) {19} ;
			\node[A point] (v20) at (20,25) {20} ;
			\node[point] (v21) at (21,19) {21} ;
			\node[A point] (v22) at (22,26) {22} ;
			\node[B point] (v23) at (23,18) {23} ;
			\node[point] (v24) at (24,23) {24} ;
			\node[point] (v25) at (25,24) {25} ;
			\node[B point] (v26) at (26,22) {26} ;
			\node[A point] (v27) at (27,29) {27} ;
			\node[B point] (v28) at (28,27) {28} ;
			\node[B point] (v29) at (29,28) {29} ;
			\node[AB point] (v30) at (30,30) {30} ;
			\begin{pgfonlayer}{background}
				\draw[line width=2pt,black!20] (0,0) -- (31,31) ;
				\draw[graph edge] (v1) to (v2) ; 
				\draw[graph edge] (v4) to (v7) ; 
				\draw[graph edge] (v4) to (v9) ; 
				\draw[graph edge] (v5) to (v6) ; 
				\draw[graph edge] (v5) to (v7) ; 
				\draw[graph edge] (v5) to (v8) ; 
				\draw[graph edge] (v5) to (v9) ; 
				\draw[graph edge] (v5) to (v10) ; 
				\draw[graph edge] (v5) to (v11) ; 
				\draw[graph edge] (v5) to (v12) ; 
				\draw[graph edge] (v6) to (v7) ; 
				\draw[graph edge] (v6) to (v9) ; 
				\draw[graph edge] (v6) to (v10) ; 
				\draw[graph edge] (v7) to (v9) ; 
				\draw[graph edge] (v8) to (v9) ; 
				\draw[graph edge] (v8) to (v10) ; 
				\draw[graph edge] (v8) to (v11) ; 
				\draw[graph edge] (v13) to (v17) ; 
				\draw[graph edge] (v14) to (v15) ; 
				\draw[graph edge] (v14) to (v16) ; 
				\draw[graph edge] (v14) to (v17) ; 
				\draw[graph edge] (v14) to (v18) ; 
				\draw[graph edge] (v14) to (v19) ; 
				\draw[graph edge] (v14) to (v21) ; 
				\draw[graph edge] (v14) to (v23) ; 
				\draw[graph edge] (v15) to (v17) ; 
				\draw[graph edge] (v16) to (v17) ; 
				\draw[graph edge] (v16) to (v19) ; 
				\draw[graph edge] (v18) to (v19) ; 
				\draw[graph edge] (v18) to (v21) ; 
				\draw[graph edge] (v18) to (v23) ; 
				\draw[graph edge] (v20) to (v21) ; 
				\draw[graph edge] (v20) to (v23) ; 
				\draw[graph edge] (v20) to (v24) ; 
				\draw[graph edge] (v20) to (v25) ; 
				\draw[graph edge] (v20) to (v26) ; 
				\draw[graph edge] (v21) to (v23) ; 
				\draw[graph edge] (v22) to (v23) ; 
				\draw[graph edge] (v22) to (v24) ; 
				\draw[graph edge] (v22) to (v25) ; 
				\draw[graph edge] (v22) to (v26) ; 
				\draw[graph edge] (v24) to (v26) ; 
				\draw[graph edge] (v25) to (v26) ; 
				\draw[graph edge] (v27) to (v28) ; 
				\draw[graph edge] (v27) to (v29) ; 
			\end{pgfonlayer}
				\node at (1,-1.5) {\tiny\ttfamily 1} ;
				\node at (2,-1.5) {\tiny\ttfamily 0} ;
				\node at (3,-1.5) {\tiny\ttfamily 1} ;
				\node at (4,-1.5) {\tiny\ttfamily 1} ;
				\node at (5,-1.5) {\tiny\ttfamily 1} ;
				\node at (6,-1.5) {\tiny\ttfamily 0} ;
				\node at (7,-1.5) {\tiny\ttfamily 0} ;
				\node at (8,-1.5) {\tiny\ttfamily 0} ;
				\node at (9,-1.5) {\tiny\ttfamily 0} ;
				\node at (10,-1.5) {\tiny\ttfamily 0} ;
				\node at (11,-1.5) {\tiny\ttfamily 0} ;
				\node at (12,-1.5) {\tiny\ttfamily 0} ;
				\node at (13,-1.5) {\tiny\ttfamily 1} ;
				\node at (14,-1.5) {\tiny\ttfamily 1} ;
				\node at (15,-1.5) {\tiny\ttfamily 0} ;
				\node at (16,-1.5) {\tiny\ttfamily 0} ;
				\node at (17,-1.5) {\tiny\ttfamily 0} ;
				\node at (18,-1.5) {\tiny\ttfamily 0} ;
				\node at (19,-1.5) {\tiny\ttfamily 0} ;
				\node at (20,-1.5) {\tiny\ttfamily 1} ;
				\node at (21,-1.5) {\tiny\ttfamily 0} ;
				\node at (22,-1.5) {\tiny\ttfamily 1} ;
				\node at (23,-1.5) {\tiny\ttfamily 0} ;
				\node at (24,-1.5) {\tiny\ttfamily 0} ;
				\node at (25,-1.5) {\tiny\ttfamily 0} ;
				\node at (26,-1.5) {\tiny\ttfamily 0} ;
				\node at (27,-1.5) {\tiny\ttfamily 1} ;
				\node at (28,-1.5) {\tiny\ttfamily 0} ;
				\node at (29,-1.5) {\tiny\ttfamily 0} ;
				\node at (30,-1.5) {\tiny\ttfamily 1} ;
				\node at (0,-1.5) {\scriptsize$\texttt A_x$} ;
				\node at (1,-2) {\tiny\ttfamily 0} ;
				\node at (2,-2) {\tiny\ttfamily 1} ;
				\node at (3,-2) {\tiny\ttfamily 1} ;
				\node at (4,-2) {\tiny\ttfamily 0} ;
				\node at (5,-2) {\tiny\ttfamily 0} ;
				\node at (6,-2) {\tiny\ttfamily 0} ;
				\node at (7,-2) {\tiny\ttfamily 0} ;
				\node at (8,-2) {\tiny\ttfamily 0} ;
				\node at (9,-2) {\tiny\ttfamily 1} ;
				\node at (10,-2) {\tiny\ttfamily 1} ;
				\node at (11,-2) {\tiny\ttfamily 1} ;
				\node at (12,-2) {\tiny\ttfamily 1} ;
				\node at (13,-2) {\tiny\ttfamily 0} ;
				\node at (14,-2) {\tiny\ttfamily 0} ;
				\node at (15,-2) {\tiny\ttfamily 0} ;
				\node at (16,-2) {\tiny\ttfamily 0} ;
				\node at (17,-2) {\tiny\ttfamily 1} ;
				\node at (18,-2) {\tiny\ttfamily 0} ;
				\node at (19,-2) {\tiny\ttfamily 1} ;
				\node at (20,-2) {\tiny\ttfamily 0} ;
				\node at (21,-2) {\tiny\ttfamily 0} ;
				\node at (22,-2) {\tiny\ttfamily 0} ;
				\node at (23,-2) {\tiny\ttfamily 1} ;
				\node at (24,-2) {\tiny\ttfamily 0} ;
				\node at (25,-2) {\tiny\ttfamily 0} ;
				\node at (26,-2) {\tiny\ttfamily 1} ;
				\node at (27,-2) {\tiny\ttfamily 0} ;
				\node at (28,-2) {\tiny\ttfamily 1} ;
				\node at (29,-2) {\tiny\ttfamily 1} ;
				\node at (30,-2) {\tiny\ttfamily 1} ;
				\node at (0,-2) {\scriptsize$\texttt B_x$} ;
				\node at (-1.5,1) {\tiny\ttfamily 0} ;
				\node at (-1.5,2) {\tiny\ttfamily 1} ;
				\node at (-1.5,3) {\tiny\ttfamily 1} ;
				\node at (-1.5,4) {\tiny\ttfamily 0} ;
				\node at (-1.5,5) {\tiny\ttfamily 0} ;
				\node at (-1.5,6) {\tiny\ttfamily 1} ;
				\node at (-1.5,7) {\tiny\ttfamily 0} ;
				\node at (-1.5,8) {\tiny\ttfamily 0} ;
				\node at (-1.5,9) {\tiny\ttfamily 0} ;
				\node at (-1.5,10) {\tiny\ttfamily 0} ;
				\node at (-1.5,11) {\tiny\ttfamily 0} ;
				\node at (-1.5,12) {\tiny\ttfamily 1} ;
				\node at (-1.5,13) {\tiny\ttfamily 0} ;
				\node at (-1.5,14) {\tiny\ttfamily 1} ;
				\node at (-1.5,15) {\tiny\ttfamily 0} ;
				\node at (-1.5,16) {\tiny\ttfamily 0} ;
				\node at (-1.5,17) {\tiny\ttfamily 0} ;
				\node at (-1.5,18) {\tiny\ttfamily 0} ;
				\node at (-1.5,19) {\tiny\ttfamily 0} ;
				\node at (-1.5,20) {\tiny\ttfamily 0} ;
				\node at (-1.5,21) {\tiny\ttfamily 1} ;
				\node at (-1.5,22) {\tiny\ttfamily 0} ;
				\node at (-1.5,23) {\tiny\ttfamily 0} ;
				\node at (-1.5,24) {\tiny\ttfamily 0} ;
				\node at (-1.5,25) {\tiny\ttfamily 1} ;
				\node at (-1.5,26) {\tiny\ttfamily 1} ;
				\node at (-1.5,27) {\tiny\ttfamily 0} ;
				\node at (-1.5,28) {\tiny\ttfamily 0} ;
				\node at (-1.5,29) {\tiny\ttfamily 1} ;
				\node at (-1.5,30) {\tiny\ttfamily 1} ;
				\node at (-1.25,0) {\scriptsize$\texttt A_y$} ;
				\node at (-2,1) {\tiny\ttfamily 1} ;
				\node at (-2,2) {\tiny\ttfamily 0} ;
				\node at (-2,3) {\tiny\ttfamily 1} ;
				\node at (-2,4) {\tiny\ttfamily 1} ;
				\node at (-2,5) {\tiny\ttfamily 0} ;
				\node at (-2,6) {\tiny\ttfamily 0} ;
				\node at (-2,7) {\tiny\ttfamily 1} ;
				\node at (-2,8) {\tiny\ttfamily 0} ;
				\node at (-2,9) {\tiny\ttfamily 1} ;
				\node at (-2,10) {\tiny\ttfamily 0} ;
				\node at (-2,11) {\tiny\ttfamily 1} ;
				\node at (-2,12) {\tiny\ttfamily 0} ;
				\node at (-2,13) {\tiny\ttfamily 1} ;
				\node at (-2,14) {\tiny\ttfamily 0} ;
				\node at (-2,15) {\tiny\ttfamily 0} ;
				\node at (-2,16) {\tiny\ttfamily 1} ;
				\node at (-2,17) {\tiny\ttfamily 0} ;
				\node at (-2,18) {\tiny\ttfamily 1} ;
				\node at (-2,19) {\tiny\ttfamily 0} ;
				\node at (-2,20) {\tiny\ttfamily 0} ;
				\node at (-2,21) {\tiny\ttfamily 0} ;
				\node at (-2,22) {\tiny\ttfamily 1} ;
				\node at (-2,23) {\tiny\ttfamily 0} ;
				\node at (-2,24) {\tiny\ttfamily 0} ;
				\node at (-2,25) {\tiny\ttfamily 0} ;
				\node at (-2,26) {\tiny\ttfamily 0} ;
				\node at (-2,27) {\tiny\ttfamily 1} ;
				\node at (-2,28) {\tiny\ttfamily 1} ;
				\node at (-2,29) {\tiny\ttfamily 0} ;
				\node at (-2,30) {\tiny\ttfamily 1} ;
				\node at (-2,0) {\scriptsize$\texttt B_y$} ;
		\end{tikzpicture}
	}}}
	\caption{%
		Example of a permutation graph with $n=30$ vertices, shown as the points $P(\pi)$.
		$A$-vertices are shown red, $B$-vertices are green and vertices that have both type $A$ and $B$ 
		(isolated vertices) are shown blue.
		Edges in $G_\pi$ are drawn yellow.
	}
	\label{fig:example}
\end{figure}

\ifsubmission{}{
\begin{figure}[htb]
	\def\scaledwidth{\linewidth}%
	\iflipics{\def\picwidth{.8\linewidth}}{}%
	\ifkoma{\def\picwidth{\linewidth}}{}%
	\ifspringer{\def\picwidth{.8\linewidth}}{}%
	\ifspringer{\def\scaledwidth{.9\linewidth}}{}%
	\plaincenter{\adjustbox{max width=\scaledwidth}{%
	\begin{tikzpicture}[
			scale=.4,
			point/.style = {draw,fill=white,circle,minimum size=12pt,inner sep=0pt,font=\tiny},
			A point/.style = {point,fill=red},
			B point/.style = {point,fill=green},
			AB point/.style = {point,fill=cyan},
			graph edge/.style = {yellow,very thick},
		]
		\draw[densely dotted] (0,0) grid (31,31) ;
		\draw[->,thick] (0,0) -- (31.500000,0) ;
		\draw[->,thick] (0,0) -- (0,31.500000) ;
			\node at (1,-.5) {\tiny 1} ;
			\node at (2,-.5) {\tiny 2} ;
			\node at (3,-.5) {\tiny 3} ;
			\node at (4,-.5) {\tiny 4} ;
			\node at (5,-.5) {\tiny 5} ;
			\node at (6,-.5) {\tiny 6} ;
			\node at (7,-.5) {\tiny 7} ;
			\node at (8,-.5) {\tiny 8} ;
			\node at (9,-.5) {\tiny 9} ;
			\node at (10,-.5) {\tiny 10} ;
			\node at (11,-.5) {\tiny 11} ;
			\node at (12,-.5) {\tiny 12} ;
			\node at (13,-.5) {\tiny 13} ;
			\node at (14,-.5) {\tiny 14} ;
			\node at (15,-.5) {\tiny 15} ;
			\node at (16,-.5) {\tiny 16} ;
			\node at (17,-.5) {\tiny 17} ;
			\node at (18,-.5) {\tiny 18} ;
			\node at (19,-.5) {\tiny 19} ;
			\node at (20,-.5) {\tiny 20} ;
			\node at (21,-.5) {\tiny 21} ;
			\node at (22,-.5) {\tiny 22} ;
			\node at (23,-.5) {\tiny 23} ;
			\node at (24,-.5) {\tiny 24} ;
			\node at (25,-.5) {\tiny 25} ;
			\node at (26,-.5) {\tiny 26} ;
			\node at (27,-.5) {\tiny 27} ;
			\node at (28,-.5) {\tiny 28} ;
			\node at (29,-.5) {\tiny 29} ;
			\node at (30,-.5) {\tiny 30} ;
			\node at (31,-.5) {\scriptsize $v$} ;
			\node at (-.5,1) {\tiny 1} ;
			\node at (-.5,2) {\tiny 2} ;
			\node at (-.5,3) {\tiny 3} ;
			\node at (-.5,4) {\tiny 4} ;
			\node at (-.5,5) {\tiny 5} ;
			\node at (-.5,6) {\tiny 6} ;
			\node at (-.5,7) {\tiny 7} ;
			\node at (-.5,8) {\tiny 8} ;
			\node at (-.5,9) {\tiny 9} ;
			\node at (-.5,10) {\tiny 10} ;
			\node at (-.5,11) {\tiny 11} ;
			\node at (-.5,12) {\tiny 12} ;
			\node at (-.5,13) {\tiny 13} ;
			\node at (-.5,14) {\tiny 14} ;
			\node at (-.5,15) {\tiny 15} ;
			\node at (-.5,16) {\tiny 16} ;
			\node at (-.5,17) {\tiny 17} ;
			\node at (-.5,18) {\tiny 18} ;
			\node at (-.5,19) {\tiny 19} ;
			\node at (-.5,20) {\tiny 20} ;
			\node at (-.5,21) {\tiny 21} ;
			\node at (-.5,22) {\tiny 22} ;
			\node at (-.5,23) {\tiny 23} ;
			\node at (-.5,24) {\tiny 24} ;
			\node at (-.5,25) {\tiny 25} ;
			\node at (-.5,26) {\tiny 26} ;
			\node at (-.5,27) {\tiny 27} ;
			\node at (-.5,28) {\tiny 28} ;
			\node at (-.5,29) {\tiny 29} ;
			\node at (-.5,30) {\tiny 30} ;
			\node[anchor=east] at (0,31) {\scriptsize $\pi^{-1}(v)$} ;
		\node[A point] (v1) at (1,18) {1} ;
		\node[A point] (v2) at (2,27) {2} ;
		\node[point] (v3) at (3,26) {3} ;
		\node[point] (v4) at (4,17) {4} ;
		\node[point] (v5) at (5,5) {5} ;
		\node[point] (v6) at (6,12) {6} ;
		\node[A point] (v7) at (7,30) {7} ;
		\node[point] (v8) at (8,9) {8} ;
		\node[B point] (v9) at (9,1) {9} ;
		\node[point] (v10) at (10,23) {10} ;
		\node[point] (v11) at (11,14) {11} ;
		\node[point] (v12) at (12,15) {12} ;
		\node[point] (v13) at (13,6) {13} ;
		\node[point] (v14) at (14,8) {14} ;
		\node[point] (v15) at (15,21) {15} ;
		\node[point] (v16) at (16,10) {16} ;
		\node[point] (v17) at (17,25) {17} ;
		\node[point] (v18) at (18,4) {18} ;
		\node[B point] (v19) at (19,2) {19} ;
		\node[point] (v20) at (20,13) {20} ;
		\node[point] (v21) at (21,16) {21} ;
		\node[point] (v22) at (22,29) {22} ;
		\node[point] (v23) at (23,7) {23} ;
		\node[B point] (v24) at (24,3) {24} ;
		\node[point] (v25) at (25,20) {25} ;
		\node[point] (v26) at (26,22) {26} ;
		\node[point] (v27) at (27,24) {27} ;
		\node[point] (v28) at (28,19) {28} ;
		\node[B point] (v29) at (29,11) {29} ;
		\node[B point] (v30) at (30,28) {30} ;
		\begin{pgfonlayer}{background}
			\draw[graph edge] (v1) to (v4) ; 
			\draw[graph edge] (v1) to (v5) ; 
			\draw[graph edge] (v1) to (v6) ; 
			\draw[graph edge] (v1) to (v8) ; 
			\draw[graph edge] (v1) to (v9) ; 
			\draw[graph edge] (v1) to (v11) ; 
			\draw[graph edge] (v1) to (v12) ; 
			\draw[graph edge] (v1) to (v13) ; 
			\draw[graph edge] (v1) to (v14) ; 
			\draw[graph edge] (v1) to (v16) ; 
			\draw[graph edge] (v1) to (v18) ; 
			\draw[graph edge] (v1) to (v19) ; 
			\draw[graph edge] (v1) to (v20) ; 
			\draw[graph edge] (v1) to (v21) ; 
			\draw[graph edge] (v1) to (v23) ; 
			\draw[graph edge] (v1) to (v24) ; 
			\draw[graph edge] (v1) to (v29) ; 
			\draw[graph edge] (v2) to (v3) ; 
			\draw[graph edge] (v2) to (v4) ; 
			\draw[graph edge] (v2) to (v5) ; 
			\draw[graph edge] (v2) to (v6) ; 
			\draw[graph edge] (v2) to (v8) ; 
			\draw[graph edge] (v2) to (v9) ; 
			\draw[graph edge] (v2) to (v10) ; 
			\draw[graph edge] (v2) to (v11) ; 
			\draw[graph edge] (v2) to (v12) ; 
			\draw[graph edge] (v2) to (v13) ; 
			\draw[graph edge] (v2) to (v14) ; 
			\draw[graph edge] (v2) to (v15) ; 
			\draw[graph edge] (v2) to (v16) ; 
			\draw[graph edge] (v2) to (v17) ; 
			\draw[graph edge] (v2) to (v18) ; 
			\draw[graph edge] (v2) to (v19) ; 
			\draw[graph edge] (v2) to (v20) ; 
			\draw[graph edge] (v2) to (v21) ; 
			\draw[graph edge] (v2) to (v23) ; 
			\draw[graph edge] (v2) to (v24) ; 
			\draw[graph edge] (v2) to (v25) ; 
			\draw[graph edge] (v2) to (v26) ; 
			\draw[graph edge] (v2) to (v27) ; 
			\draw[graph edge] (v2) to (v28) ; 
			\draw[graph edge] (v2) to (v29) ; 
			\draw[graph edge] (v3) to (v4) ; 
			\draw[graph edge] (v3) to (v5) ; 
			\draw[graph edge] (v3) to (v6) ; 
			\draw[graph edge] (v3) to (v8) ; 
			\draw[graph edge] (v3) to (v9) ; 
			\draw[graph edge] (v3) to (v10) ; 
			\draw[graph edge] (v3) to (v11) ; 
			\draw[graph edge] (v3) to (v12) ; 
			\draw[graph edge] (v3) to (v13) ; 
			\draw[graph edge] (v3) to (v14) ; 
			\draw[graph edge] (v3) to (v15) ; 
			\draw[graph edge] (v3) to (v16) ; 
			\draw[graph edge] (v3) to (v17) ; 
			\draw[graph edge] (v3) to (v18) ; 
			\draw[graph edge] (v3) to (v19) ; 
			\draw[graph edge] (v3) to (v20) ; 
			\draw[graph edge] (v3) to (v21) ; 
			\draw[graph edge] (v3) to (v23) ; 
			\draw[graph edge] (v3) to (v24) ; 
			\draw[graph edge] (v3) to (v25) ; 
			\draw[graph edge] (v3) to (v26) ; 
			\draw[graph edge] (v3) to (v27) ; 
			\draw[graph edge] (v3) to (v28) ; 
			\draw[graph edge] (v3) to (v29) ; 
			\draw[graph edge] (v4) to (v5) ; 
			\draw[graph edge] (v4) to (v6) ; 
			\draw[graph edge] (v4) to (v8) ; 
			\draw[graph edge] (v4) to (v9) ; 
			\draw[graph edge] (v4) to (v11) ; 
			\draw[graph edge] (v4) to (v12) ; 
			\draw[graph edge] (v4) to (v13) ; 
			\draw[graph edge] (v4) to (v14) ; 
			\draw[graph edge] (v4) to (v16) ; 
			\draw[graph edge] (v4) to (v18) ; 
			\draw[graph edge] (v4) to (v19) ; 
			\draw[graph edge] (v4) to (v20) ; 
			\draw[graph edge] (v4) to (v21) ; 
			\draw[graph edge] (v4) to (v23) ; 
			\draw[graph edge] (v4) to (v24) ; 
			\draw[graph edge] (v4) to (v29) ; 
			\draw[graph edge] (v5) to (v9) ; 
			\draw[graph edge] (v5) to (v18) ; 
			\draw[graph edge] (v5) to (v19) ; 
			\draw[graph edge] (v5) to (v24) ; 
			\draw[graph edge] (v6) to (v8) ; 
			\draw[graph edge] (v6) to (v9) ; 
			\draw[graph edge] (v6) to (v13) ; 
			\draw[graph edge] (v6) to (v14) ; 
			\draw[graph edge] (v6) to (v16) ; 
			\draw[graph edge] (v6) to (v18) ; 
			\draw[graph edge] (v6) to (v19) ; 
			\draw[graph edge] (v6) to (v23) ; 
			\draw[graph edge] (v6) to (v24) ; 
			\draw[graph edge] (v6) to (v29) ; 
			\draw[graph edge] (v7) to (v8) ; 
			\draw[graph edge] (v7) to (v9) ; 
			\draw[graph edge] (v7) to (v10) ; 
			\draw[graph edge] (v7) to (v11) ; 
			\draw[graph edge] (v7) to (v12) ; 
			\draw[graph edge] (v7) to (v13) ; 
			\draw[graph edge] (v7) to (v14) ; 
			\draw[graph edge] (v7) to (v15) ; 
			\draw[graph edge] (v7) to (v16) ; 
			\draw[graph edge] (v7) to (v17) ; 
			\draw[graph edge] (v7) to (v18) ; 
			\draw[graph edge] (v7) to (v19) ; 
			\draw[graph edge] (v7) to (v20) ; 
			\draw[graph edge] (v7) to (v21) ; 
			\draw[graph edge] (v7) to (v22) ; 
			\draw[graph edge] (v7) to (v23) ; 
			\draw[graph edge] (v7) to (v24) ; 
			\draw[graph edge] (v7) to (v25) ; 
			\draw[graph edge] (v7) to (v26) ; 
			\draw[graph edge] (v7) to (v27) ; 
			\draw[graph edge] (v7) to (v28) ; 
			\draw[graph edge] (v7) to (v29) ; 
			\draw[graph edge] (v7) to (v30) ; 
			\draw[graph edge] (v8) to (v9) ; 
			\draw[graph edge] (v8) to (v13) ; 
			\draw[graph edge] (v8) to (v14) ; 
			\draw[graph edge] (v8) to (v18) ; 
			\draw[graph edge] (v8) to (v19) ; 
			\draw[graph edge] (v8) to (v23) ; 
			\draw[graph edge] (v8) to (v24) ; 
			\draw[graph edge] (v10) to (v11) ; 
			\draw[graph edge] (v10) to (v12) ; 
			\draw[graph edge] (v10) to (v13) ; 
			\draw[graph edge] (v10) to (v14) ; 
			\draw[graph edge] (v10) to (v15) ; 
			\draw[graph edge] (v10) to (v16) ; 
			\draw[graph edge] (v10) to (v18) ; 
			\draw[graph edge] (v10) to (v19) ; 
			\draw[graph edge] (v10) to (v20) ; 
			\draw[graph edge] (v10) to (v21) ; 
			\draw[graph edge] (v10) to (v23) ; 
			\draw[graph edge] (v10) to (v24) ; 
			\draw[graph edge] (v10) to (v25) ; 
			\draw[graph edge] (v10) to (v26) ; 
			\draw[graph edge] (v10) to (v28) ; 
			\draw[graph edge] (v10) to (v29) ; 
			\draw[graph edge] (v11) to (v13) ; 
			\draw[graph edge] (v11) to (v14) ; 
			\draw[graph edge] (v11) to (v16) ; 
			\draw[graph edge] (v11) to (v18) ; 
			\draw[graph edge] (v11) to (v19) ; 
			\draw[graph edge] (v11) to (v20) ; 
			\draw[graph edge] (v11) to (v23) ; 
			\draw[graph edge] (v11) to (v24) ; 
			\draw[graph edge] (v11) to (v29) ; 
			\draw[graph edge] (v12) to (v13) ; 
			\draw[graph edge] (v12) to (v14) ; 
			\draw[graph edge] (v12) to (v16) ; 
			\draw[graph edge] (v12) to (v18) ; 
			\draw[graph edge] (v12) to (v19) ; 
			\draw[graph edge] (v12) to (v20) ; 
			\draw[graph edge] (v12) to (v23) ; 
			\draw[graph edge] (v12) to (v24) ; 
			\draw[graph edge] (v12) to (v29) ; 
			\draw[graph edge] (v13) to (v18) ; 
			\draw[graph edge] (v13) to (v19) ; 
			\draw[graph edge] (v13) to (v24) ; 
			\draw[graph edge] (v14) to (v18) ; 
			\draw[graph edge] (v14) to (v19) ; 
			\draw[graph edge] (v14) to (v23) ; 
			\draw[graph edge] (v14) to (v24) ; 
			\draw[graph edge] (v15) to (v16) ; 
			\draw[graph edge] (v15) to (v18) ; 
			\draw[graph edge] (v15) to (v19) ; 
			\draw[graph edge] (v15) to (v20) ; 
			\draw[graph edge] (v15) to (v21) ; 
			\draw[graph edge] (v15) to (v23) ; 
			\draw[graph edge] (v15) to (v24) ; 
			\draw[graph edge] (v15) to (v25) ; 
			\draw[graph edge] (v15) to (v28) ; 
			\draw[graph edge] (v15) to (v29) ; 
			\draw[graph edge] (v16) to (v18) ; 
			\draw[graph edge] (v16) to (v19) ; 
			\draw[graph edge] (v16) to (v23) ; 
			\draw[graph edge] (v16) to (v24) ; 
			\draw[graph edge] (v17) to (v18) ; 
			\draw[graph edge] (v17) to (v19) ; 
			\draw[graph edge] (v17) to (v20) ; 
			\draw[graph edge] (v17) to (v21) ; 
			\draw[graph edge] (v17) to (v23) ; 
			\draw[graph edge] (v17) to (v24) ; 
			\draw[graph edge] (v17) to (v25) ; 
			\draw[graph edge] (v17) to (v26) ; 
			\draw[graph edge] (v17) to (v27) ; 
			\draw[graph edge] (v17) to (v28) ; 
			\draw[graph edge] (v17) to (v29) ; 
			\draw[graph edge] (v18) to (v19) ; 
			\draw[graph edge] (v18) to (v24) ; 
			\draw[graph edge] (v20) to (v23) ; 
			\draw[graph edge] (v20) to (v24) ; 
			\draw[graph edge] (v20) to (v29) ; 
			\draw[graph edge] (v21) to (v23) ; 
			\draw[graph edge] (v21) to (v24) ; 
			\draw[graph edge] (v21) to (v29) ; 
			\draw[graph edge] (v22) to (v23) ; 
			\draw[graph edge] (v22) to (v24) ; 
			\draw[graph edge] (v22) to (v25) ; 
			\draw[graph edge] (v22) to (v26) ; 
			\draw[graph edge] (v22) to (v27) ; 
			\draw[graph edge] (v22) to (v28) ; 
			\draw[graph edge] (v22) to (v29) ; 
			\draw[graph edge] (v22) to (v30) ; 
			\draw[graph edge] (v23) to (v24) ; 
			\draw[graph edge] (v25) to (v28) ; 
			\draw[graph edge] (v25) to (v29) ; 
			\draw[graph edge] (v26) to (v28) ; 
			\draw[graph edge] (v26) to (v29) ; 
			\draw[graph edge] (v27) to (v28) ; 
			\draw[graph edge] (v27) to (v29) ; 
			\draw[graph edge] (v28) to (v29) ; 
		\end{pgfonlayer}
			\node at (1,-1.5) {\tiny\ttfamily 1} ;
			\node at (2,-1.5) {\tiny\ttfamily 1} ;
			\node at (3,-1.5) {\tiny\ttfamily 0} ;
			\node at (4,-1.5) {\tiny\ttfamily 0} ;
			\node at (5,-1.5) {\tiny\ttfamily 0} ;
			\node at (6,-1.5) {\tiny\ttfamily 0} ;
			\node at (7,-1.5) {\tiny\ttfamily 1} ;
			\node at (8,-1.5) {\tiny\ttfamily 0} ;
			\node at (9,-1.5) {\tiny\ttfamily 0} ;
			\node at (10,-1.5) {\tiny\ttfamily 0} ;
			\node at (11,-1.5) {\tiny\ttfamily 0} ;
			\node at (12,-1.5) {\tiny\ttfamily 0} ;
			\node at (13,-1.5) {\tiny\ttfamily 0} ;
			\node at (14,-1.5) {\tiny\ttfamily 0} ;
			\node at (15,-1.5) {\tiny\ttfamily 0} ;
			\node at (16,-1.5) {\tiny\ttfamily 0} ;
			\node at (17,-1.5) {\tiny\ttfamily 0} ;
			\node at (18,-1.5) {\tiny\ttfamily 0} ;
			\node at (19,-1.5) {\tiny\ttfamily 0} ;
			\node at (20,-1.5) {\tiny\ttfamily 0} ;
			\node at (21,-1.5) {\tiny\ttfamily 0} ;
			\node at (22,-1.5) {\tiny\ttfamily 0} ;
			\node at (23,-1.5) {\tiny\ttfamily 0} ;
			\node at (24,-1.5) {\tiny\ttfamily 0} ;
			\node at (25,-1.5) {\tiny\ttfamily 0} ;
			\node at (26,-1.5) {\tiny\ttfamily 0} ;
			\node at (27,-1.5) {\tiny\ttfamily 0} ;
			\node at (28,-1.5) {\tiny\ttfamily 0} ;
			\node at (29,-1.5) {\tiny\ttfamily 0} ;
			\node at (30,-1.5) {\tiny\ttfamily 0} ;
			\node at (0,-1.5) {\scriptsize$\texttt A_x$} ;
			\node at (1,-2) {\tiny\ttfamily 0} ;
			\node at (2,-2) {\tiny\ttfamily 0} ;
			\node at (3,-2) {\tiny\ttfamily 0} ;
			\node at (4,-2) {\tiny\ttfamily 0} ;
			\node at (5,-2) {\tiny\ttfamily 0} ;
			\node at (6,-2) {\tiny\ttfamily 0} ;
			\node at (7,-2) {\tiny\ttfamily 0} ;
			\node at (8,-2) {\tiny\ttfamily 0} ;
			\node at (9,-2) {\tiny\ttfamily 1} ;
			\node at (10,-2) {\tiny\ttfamily 0} ;
			\node at (11,-2) {\tiny\ttfamily 0} ;
			\node at (12,-2) {\tiny\ttfamily 0} ;
			\node at (13,-2) {\tiny\ttfamily 0} ;
			\node at (14,-2) {\tiny\ttfamily 0} ;
			\node at (15,-2) {\tiny\ttfamily 0} ;
			\node at (16,-2) {\tiny\ttfamily 0} ;
			\node at (17,-2) {\tiny\ttfamily 0} ;
			\node at (18,-2) {\tiny\ttfamily 0} ;
			\node at (19,-2) {\tiny\ttfamily 1} ;
			\node at (20,-2) {\tiny\ttfamily 0} ;
			\node at (21,-2) {\tiny\ttfamily 0} ;
			\node at (22,-2) {\tiny\ttfamily 0} ;
			\node at (23,-2) {\tiny\ttfamily 0} ;
			\node at (24,-2) {\tiny\ttfamily 1} ;
			\node at (25,-2) {\tiny\ttfamily 0} ;
			\node at (26,-2) {\tiny\ttfamily 0} ;
			\node at (27,-2) {\tiny\ttfamily 0} ;
			\node at (28,-2) {\tiny\ttfamily 0} ;
			\node at (29,-2) {\tiny\ttfamily 1} ;
			\node at (30,-2) {\tiny\ttfamily 1} ;
			\node at (0,-2) {\scriptsize$\texttt B_x$} ;
			\node at (-1.5,1) {\tiny\ttfamily 0} ;
			\node at (-1.5,2) {\tiny\ttfamily 0} ;
			\node at (-1.5,3) {\tiny\ttfamily 0} ;
			\node at (-1.5,4) {\tiny\ttfamily 0} ;
			\node at (-1.5,5) {\tiny\ttfamily 0} ;
			\node at (-1.5,6) {\tiny\ttfamily 0} ;
			\node at (-1.5,7) {\tiny\ttfamily 0} ;
			\node at (-1.5,8) {\tiny\ttfamily 0} ;
			\node at (-1.5,9) {\tiny\ttfamily 0} ;
			\node at (-1.5,10) {\tiny\ttfamily 0} ;
			\node at (-1.5,11) {\tiny\ttfamily 0} ;
			\node at (-1.5,12) {\tiny\ttfamily 0} ;
			\node at (-1.5,13) {\tiny\ttfamily 0} ;
			\node at (-1.5,14) {\tiny\ttfamily 0} ;
			\node at (-1.5,15) {\tiny\ttfamily 0} ;
			\node at (-1.5,16) {\tiny\ttfamily 0} ;
			\node at (-1.5,17) {\tiny\ttfamily 0} ;
			\node at (-1.5,18) {\tiny\ttfamily 1} ;
			\node at (-1.5,19) {\tiny\ttfamily 0} ;
			\node at (-1.5,20) {\tiny\ttfamily 0} ;
			\node at (-1.5,21) {\tiny\ttfamily 0} ;
			\node at (-1.5,22) {\tiny\ttfamily 0} ;
			\node at (-1.5,23) {\tiny\ttfamily 0} ;
			\node at (-1.5,24) {\tiny\ttfamily 0} ;
			\node at (-1.5,25) {\tiny\ttfamily 0} ;
			\node at (-1.5,26) {\tiny\ttfamily 0} ;
			\node at (-1.5,27) {\tiny\ttfamily 1} ;
			\node at (-1.5,28) {\tiny\ttfamily 0} ;
			\node at (-1.5,29) {\tiny\ttfamily 0} ;
			\node at (-1.5,30) {\tiny\ttfamily 1} ;
			\node at (-1.25,0) {\scriptsize$\texttt A_y$} ;
			\node at (-2,1) {\tiny\ttfamily 1} ;
			\node at (-2,2) {\tiny\ttfamily 1} ;
			\node at (-2,3) {\tiny\ttfamily 1} ;
			\node at (-2,4) {\tiny\ttfamily 0} ;
			\node at (-2,5) {\tiny\ttfamily 0} ;
			\node at (-2,6) {\tiny\ttfamily 0} ;
			\node at (-2,7) {\tiny\ttfamily 0} ;
			\node at (-2,8) {\tiny\ttfamily 0} ;
			\node at (-2,9) {\tiny\ttfamily 0} ;
			\node at (-2,10) {\tiny\ttfamily 0} ;
			\node at (-2,11) {\tiny\ttfamily 1} ;
			\node at (-2,12) {\tiny\ttfamily 0} ;
			\node at (-2,13) {\tiny\ttfamily 0} ;
			\node at (-2,14) {\tiny\ttfamily 0} ;
			\node at (-2,15) {\tiny\ttfamily 0} ;
			\node at (-2,16) {\tiny\ttfamily 0} ;
			\node at (-2,17) {\tiny\ttfamily 0} ;
			\node at (-2,18) {\tiny\ttfamily 0} ;
			\node at (-2,19) {\tiny\ttfamily 0} ;
			\node at (-2,20) {\tiny\ttfamily 0} ;
			\node at (-2,21) {\tiny\ttfamily 0} ;
			\node at (-2,22) {\tiny\ttfamily 0} ;
			\node at (-2,23) {\tiny\ttfamily 0} ;
			\node at (-2,24) {\tiny\ttfamily 0} ;
			\node at (-2,25) {\tiny\ttfamily 0} ;
			\node at (-2,26) {\tiny\ttfamily 0} ;
			\node at (-2,27) {\tiny\ttfamily 0} ;
			\node at (-2,28) {\tiny\ttfamily 1} ;
			\node at (-2,29) {\tiny\ttfamily 0} ;
			\node at (-2,30) {\tiny\ttfamily 0} ;
			\node at (-2,0) {\scriptsize$\texttt B_y$} ;
	\end{tikzpicture}
	}}
	\caption{%
		Another example of a permutation graph; the drawing is as in \wref{fig:example}.
		This graph is a typical graph when the $\pi$ is chosen uniformly at random.%
	}
	\label{fig:example2}
\end{figure}
}

\paragraph{Distance}
A vertex $v\in[n]$ is a \emph{type-$A$ vertex} iff $\pi^{-1}$ has a left-to-right maximum at position $v$,
\ie, when $\pi^{-1}(v) \ge \pi^{-1}(u)$ for all $u < v$.
Note that $1$ is always a left-to-right maximum.
Similarly, a vertex $v\in[n]$ is \emph{type $B$} iff $\pi^{-1}$ has a right-to-left minimum at $v$,
\ie, $\pi^{-1}(v) \le \pi^{-1}(u)$ for all $u > v$; vertex $n$ is always type $B$.
As in~\cite{BazzaroGavoille2009}, we use $A$ and $B$ 
to denote the set of $A$-vertices and $B$-vertices, respectively,
and we define: 
\ifsubmission{%
	\begin{align*}
	a^-(v) &\wrel= \like[l]{\max}{\min} (\GNeighbor(v)\cap A), &
	b^-(v) &\wrel= \like[l]{\max}{\min} (\GNeighbor(v)\cap B), \\
	a^+(v) &\wrel= \max (\GNeighbor(v)\cap A), &
	b^+(v) &\wrel= \max (\GNeighbor(v)\cap B). 
	\end{align*}
}{%
\begin{align*}
a^-(v) &\wrel= \like[l]{\max}{\min} (\GNeighbor(v)\cap A), \\
a^+(v) &\wrel= \max (\GNeighbor(v)\cap A), \\
b^-(v) &\wrel= \like[l]{\max}{\min} (\GNeighbor(v)\cap B), \\
b^+(v) &\wrel= \max (\GNeighbor(v)\cap B). 
\end{align*}
}
If we are computing a shortest path from $u$ to $v$, then 
either $u$ and $v$ are adjacent, or there is a shortest path whose first vertex 
after $u$ is one of $a^+(u)$ and $b^+(u)$, if $v > u$, or 
one of $a^-(u)$ and $b^-(u)$, if $v < u$.
It is therefore vital to be able to compute these four functions.
For that, we store four bitvectors with rank/select support (\wref{lem:compressed-bit-vectors})
that encode which points belong to $A$ (resp.\ $B$) given an $x$- (resp.\ $y$-)coordinate:
\ifsubmission{%
\begin{align*}
	\like{Wl[1..n]}{\mathtt A_x[1..n]} \quad\text{with}\quad \like{WW}{\mathtt A_x [u]} &\wrel= [u \in A] , \qquad\;&
	\like{Wl[1..n]}{\mathtt A_y[1..n]} \quad\text{with}\quad \like{WW}{\mathtt A_y [i]} &\wrel= [\pi(i) \in A],\\
	\like{Wl[1..n]}{\mathtt B_x[1..n]} \quad\text{with}\quad \like{WW}{\mathtt B_x [u]} &\wrel= [u \in B] ,&
	\like{Wl[1..n]}{\mathtt B_y[1..n]} \quad\text{with}\quad \like{WW}{\mathtt A_y [i]} &\wrel= [\pi(i) \in B].
\end{align*}
}{%
\begin{align*}
	\like{Wl[1..n]}{\mathtt A_x[1..n]} \quad\text{with}\quad \like{WW}{\mathtt A_x [u]} &\wrel= [u \in A] ,\\
	\like{Wl[1..n]}{\mathtt B_x[1..n]} \quad\text{with}\quad \like{WW}{\mathtt B_x [u]} &\wrel= [u \in B] ,\\
	\like{Wl[1..n]}{\mathtt A_y[1..n]} \quad\text{with}\quad \like{WW}{\mathtt A_y [i]} &\wrel= [\pi(i) \in A],\\
	\like{Wl[1..n]}{\mathtt B_y[1..n]} \quad\text{with}\quad \like{WW}{\mathtt A_y [i]} &\wrel= [\pi(i) \in B].
\end{align*}
}
\ifsubmission{%
	\wref{fig:example} shows an example of these bitvectors.
}{%
	\wref{fig:example} and \wref{fig:example2} show examples of these bitvectors.
}%
\ifsubmission{%
	We can now use these to compute the extremal $A$ neighbors of a vertex $v$ as follows: 
	\begin{align*}
		a^+(v) &\wrel= \selop_1\bigl(\mathtt A_x, \rankop_1(\mathtt A_x, v)\bigr) , \\
		a^-(v) &\wrel= \selop_1\bigl(\mathtt A_x, \rankop_1(\mathtt A_y, \pi^{-1}(v)-1)+1\bigr) .
	\end{align*}
	Extremal $B$ neighbors are found similarly.%
}{%
We can now use these to compute the extremal $A$ and $B$ neighbors of a vertex $v$ as follows: 
\begin{align*}
	a^+(v) &\wrel= \selop_1\bigl(\mathtt A_x, \rankop_1(\mathtt A_x, v)\bigr) , \\
	a^-(v) &\wrel= \selop_1\bigl(\mathtt A_x, \rankop_1(\mathtt A_y, \pi^{-1}(v)-1)+1\bigr) , \\
	b^+(v) &\wrel= \selop_1\bigl(\mathtt B_x, \rankop_1(\mathtt B_y, \pi^{-1}(v))\bigr) , \\
	b^-(v) &\wrel= \selop_1\bigl(\mathtt B_x, \rankop_1(\mathtt B_x, v-1)+1\bigr) .
\end{align*}
}
The computation takes $O(1)$ time plus at most one evaluation of $\pi^{-1}(v)$.

\begin{remark}[$\pi^{-1}$ for $A$/$B$-vertices]
\label{rem:pi-inv-of-A-B}
	We note here (for later reference) that for $a\in A$ we can compute 
	$\pi^{-1}(a) = \selop_1(\mathtt A_y, \rankop_1(\mathtt A_x, a))$ just from the bitvectors 
	without access to $\Pi$, because $\pi^{-1}$ is monotonically increasing on $A$;
	similarly for $b\in B$:
	$\pi^{-1}(b) = \selop_1(\mathtt B_y, \rankop_1(\mathtt B_x, b))$.
\end{remark}

In \cite[Thm.\,2.1]{BazzaroGavoille2009}, Bazzaro and Gavoille show that the distances/shortest paths
in a PG can now be found by testing for the special cases of distance
at most~$3$ 
(using $a^\pm$ or $b^\pm$) 
or by asking a distance query in a proper interval graph.
More specifically, let $u<v$.
\ifsubmission{\begin{enumerateinline}}{\begin{enumerate}}
\item If $\GAdjacent(u,v)$, the distance is $1$ and we are done.
\item Otherwise, if $\GAdjacent(a^+(u), v)$ or $\GAdjacent(b^+(u), v)$, 
	which can equivalently be written as $a^-(v) \le a^+(u) \,\vee\, b^-(v) \le b^+(u)$, 
	the distance is $2$ and we are done.
\item Otherwise, if $\GAdjacent(a^+(u),b^-(v))$ or $\GAdjacent(b^+(u),a^-(v))$,
	which is equivalent to $a^-(v) \le a^+(b^+(u)) \,\vee\, b^-(v) \le b^+(a^+(u))$,
	the distance is $3$ and we are done.
\item Otherwise, the distance is the minimum of the following four cases:\ifsubmission{\\}{\\[1ex]}
	$2 + 2\cdot \GDistance_{G_B}(b^+(u),b^-(v))$, \quad
	$3 + 2\cdot \GDistance_{G_B}(b^+(a^+(u),b^-(v))$, \\
	$2 + 2\cdot \GDistance_{G_A}(a^+(u),a^-(v))$, \quad
	$3 + 2\cdot \GDistance_{G_A}(a^+(b^+(u),a^-(v))$.
\ifsubmission{\end{enumerateinline}}{\end{enumerate}}
\ifspringer{\vspace{1ex}}{}
Here $G_A$ is the interval graph (intersection graph) 
defined by intervals $[b^-(v),b^+(v)]$ for all $v\in A$
and $G_B$ by intervals $[a^-(v),a^+(v)]$ for all $v\in B$.
In general, these intervals share endpoints, but they can be transformed 
into a proper realization by breaking ties by vertex $v$%
\ifsubmission{.}{%
	, \eg, for $G_A$, we
use $[b^-(v)-(n-v)\cdot \epsilon,b^+(v)+v\cdot \epsilon]$ instead of $[b^-(v),b^+(v)]$
for, say, $\epsilon=1/n^2$.
Then all endpoints are disjoint and no interval properly contains another;
moreover, the $i$th smallest left endpoint corresponds to the $i$th smallest
vertex in $A$.

}
\ifsubmission{}{\par}%
We compute the data structure of \wref{lem:proper-interval-graphs} for $G_A$ and $G_B$;
to map vertex $v\in A$ to the corresponding vertex in $G_A$, we simply
compute $\rankop_1(\texttt A_x, v)$; 
recall that the data structure of \wref{lem:proper-interval-graphs}
identifies vertices with the rank of their left endpoints.
With that, we can compute the four distances above and return the minimum.

The running time for \GDistance is the time needed for a constant number of extremal neighbor queries
($\Oh(1)$ for the array-based data structure, $\Oh(\log n / \log\log n)$ for the grid-based one),
a constant number of adjacency checks (same running times),
a constant number of rank-queries ($O(1)$ each),
and finally a constant number of \GDistance queries in proper interval graphs (again $O(1)$).
The running time for \GDistance is thus dominated by the time for evaluating $\pi^{-1}$.

\paragraph{Shortest paths}
Suppose $u<v$.
As noted by Bazzaro and Gavoille~\cite{BazzaroGavoille2009},
the above case distinction does not only determine the distance, but also determines
in each case a next vertex $w$ after $u$ on a shortest path from $u$ to $v$.
We output $u$ and unless $u=v$, we recursively call $\GSPath(w,v)$.

\ifsubmission{}{\par}
Since the running time for all checks above is dominated by $\pi^{-1}(v)$,
we can iterate through the vertices on $\GSPath(u,v)$ in
$\Oh(1)$ time per vertex for the array-based data structure,
and in $\Oh(\log n / \log\log n)$ time per vertex for the grid-based data structure.

\paragraph{Space}
The four bitvectors $\mathtt A_x$, $\mathtt B_x$, $\mathtt A_y$, and $\mathtt B_y$
require no more than $4n+o(n)$ bits of space including the support for rank and select operations.

\ifsubmission{%
	We can slightly improve upon this. After removing isolated vertices,
	any node is either an $A$-node, a $B$-node, or neither,
	which can be encoded using at most $\lg (3) n + o(n)$ bits of space
	per dimension ($x$ and $y$), for a total of at most $3.16993n+o(n)$ bits.
}{%
	\par
	When we allow ourselves to modify $\pi$, we can slightly improve upon this: 
We first move all isolated vertices to the largest indices. 
Note that any connected components can be freely permuted without changing the graph; 
in the point grid this has to be done by shifts along the $y=x$ line.
We now store the number~$w$ of isolated vertices.
Each of the remaining nodes, $[n-w]$, can either be an $A$-node, a $B$-node, or neither,
	which can be encoded as a string over $\{A,B,N\}$. 
	We store this string as a wavelet tree (\wref{lem:wavelet-trees}) with support for rank and select, 
	using at most $\lg (3) n + o(n)$ bits of space
per dimension ($x$ and $y$), for a total of at most $3.16993n+o(n)$ bits.

	(The data structure can sometimes achieve even better compression 
	since it compresses to the empirical entropy of the string).
}
\ifsubmission{}{\par}
$G_A$ and $G_B$ have no more than $n$ vertices in total, so the data structures
from \wref{lem:proper-interval-graphs} will use at most $3n+o(n)$ bits of space.
In addition to that, we need $\epsilon n$ bits of space for the range-maximum and range-minimum indices,
for a total of $(6.17+\epsilon) n + o(n)$ bits of space on top of storing $\Pi$. Assuming we using the data structure of~\cite{DodisParascuThorup2010} for the latter, the total space is 
$n \lg n + (6.17+\epsilon) n + o(n)$.

\ifsubmission{}{\medskip\noindent}%
This concludes the proof of \wref{thm:main}.

%% file: permutation-graphs-rmq-next-neighbor.tex
For that, we have to look into the black box that is the RMQ index from \wref{lem:rmq-indexing}.
Indeed, what we describe here is modification of the construction of 
Fischer and Heun~\cite[Thm.\,3.7]{FischerHeun2011}
that has the same asymptotic performance characteristics as in \wref{lem:rmq-indexing},
but allows to iterate over values above a threshold.

\begin{lemma}[RMQ index with next-above]
\label{lem:rmq-indexing-next}
	Let $A[1..n]$ be a static array of comparable elements. 
	For any constant $\epsilon>0$, 
	there is a data structure using $\epsilon n$ bits of space on top of $A$ 
	that supports the following queries in $O(1/\epsilon)$ time 
	(making as many queries to $A$) and using $O(1)$ words of working memory:
	\begin{thmenumerate}{lem:rmq-indexing-next}
	\item range-maximum queries, $\RangeMax_A(\ell,r)$,
	\item next-above queries, $\NextAbove_A(\ell,r,y;i)$, 
		enumerating $\{i\in[\ell,r] : A[i] \ge y\}$
		in amortized $O(1/\epsilon)$ time.\\
		Formally, \NextAbove implicitly defines a sequence $(i_j)_{j\ge 0}$ via
		$i_0 = \RangeMax_A(\ell,r)$ if $A[i_0] \ge y$ and $i_0 = \text{null}$ otherwise,
		and $i_{j+1} = \NextAbove_A(\ell,r,y;i_j)$ if $i_j \ne \text{null}$ and $i_{j+1} = \text{null}$ otherwise.
		Then we require $\{i_j : i_j \ne \text{null}\} = \{i\in[\ell,r] : A[i] \ge y\}$.
	\end{thmenumerate}
\end{lemma}
This index can be used to iterate over the result of 
3-sided orthogonal range queries with amortized constant delay and using
constant working memory by computing the sequence $(i_j)$.

\begin{proof}
A $2\epsilon n+o(n)$ bit RMQ index for an array $A[1..n]$ can be obtained by (conceptually) dividing
$A$ into $\epsilon n$ blocks of $\lceil 1/\epsilon\rceil$ elements each and storing the 
Cartesian tree~\cite{GabowBentleyTarjan1984,Vuillemin1980} of the block maxima
as a succinct binary tree~\cite[Thm.\,3]{DavoodiRamanRao2017} in $2\epsilon n+o(n)$ bits.
This tree data structure allows in constant time to 
(a) map between nodes and their corresponding block indices in $A$,
(b) map between nodes and preorder indices,
(c) find the lowest common ancestor (LCA) of two nodes, and
(d) return the number of descendants of a node.
We first discuss how to solve the problem for $\epsilon=1$, \ie, when all elements are part of the tree.
We will identify nodes in the Cartesian tree $T$ with their inorder number, \ie, the index in $A$.
To answer $\RangeMax_A(\ell,r)$, we simply use the Cartesian tree operations to find the nodes (of inorder index) $\ell$ and $r$ and return (the inorder index of) their
LCA.

To iterate through all indices $i\in[\ell,r]$ with $A[i] \ge y$, we will now show how
to compute the next such index, $\NextAbove_A(\ell, r, y; i)$, given only a current such index $i$ (and $\ell$, $r$ and $y$);
if no further such index exists, we will return ``null''.

First, we compute $i_0=\RangeMax_A(\ell,r)$.
We will iterate through indices in the order of a \emph{preorder} traversal of
the subtree rooted at $i_0$, starting from the current node $i$.
The challenge is to, in constant time, skip over parts of the tree that are
outside of the range $[\ell,r]$ or have all $A$-values below $y$.
More specifically, the first step is to find the next candidate index $s\in[\ell,r]$, for which $A[s]\ge y$ might hold,
given the current index $i$.
We initialize $s$ to the successor of $i$ in preorder.

Now, we repeat the following steps until we have either found the next index or have determined that none exists.
If $s$ is \emph{not} a descendant of $i_0$ in $T$,
then there are no more indices to report and we return null;
we can check this condition in constant time by comparing the preorder index of $s$ to the sum of the preorder index of $i_0$ and $i_0$'s subtree size.

If $s$ is within $i_0$'s subtree, we check whether $s\in[\ell,r]$; if not, $s$ is too far left or too far right,
and we have to find the next node (in preorder) that lies inside $[\ell,r]$.
If $s<\ell$ and $i>\ell$, then $s$ is the left child of $i$, and following right-child links from $s$ eventually brings us back into the range $[\ell,r]$ since node $i-1\in[\ell,r]$ must lie in $s$'s right subtree. 
In this case, we update $s$ to the LCA of $\ell$ and $i-1$ to obtain, in $O(1)$ time, the first node (in preorder) where this sequence of right-child links from $s$ enters the inorder range $[\ell,r]$ again.
If $s<\ell$ and $i = \ell$, $i$ is the leftmost node in the range and we have to skip its left subtree. We can do this by advancing from $s$ (in preorder) 
by as many nodes as $s$'s right child has descendants;
the tree data structure again supports this in constant time.
The symmetric case of $s>r$ is handled similarly. If $i < r$, we set $s$ to LCA of $i+1$ and $r$; if $s>r$ and $i=r$, $i$ was the last node in preorder with inorder index in range $[\ell,r]$, so we can return null.

In all cases, after $O(1)$ time, we either terminate or arrive at the next candidate node $s$.
If $A[s]\ge y$ we return $s$ and are done.
Otherwise, \ie, if $A[s] < y$, then $s$ and its entire subtree have to be skipped;
the tree data structure supports this in constant time (as above).
Then we repeat the above steps with the new $s$.

We note that the accesses to $A$ are the same as in the naive implementation of 
three-sided range reporting, and only constant time is needed between two such accesses;
hence the same time bounds hold.

When we use blocks of $c = \lceil 1/\epsilon\rceil$ elements and only construct $T$ based
on the block minima, we modify this procedure as follows.
When we are given a current index $i$, we first check the indices $j>i$ in $i$'s block.
If any $j$ has $A[j]\ge y$, we return it.
Only if none of the indices in $i$'s block are returned, we continue with the above procedure to
find the next candidate node $s$.
When we compare the candidate node ``$A[s]\ge y$'', we now iterate through the block
corresponding to node $s$ and compare each array entry with $y$.
When we find $i$ with $A[i]\ge y$, we return this index; if none of the elements in the block
where big enough, we continue as if $A[s]<y$ held.
\end{proof}


From the discussion above, it is clear that $\GNextNeighbor$
corresponds exactly to \NextAbove queries
(separately for $N^+$ and $N^-$),
and so using \wref{lem:rmq-indexing-next}, 
we can support $\GNextNeighbor(u,w)$ with constant words of extra 
working memory and amortized constant running time 
(amortized over the iteration over all neighbors of $u$).

%% file: permutation-graphs-algorithms.tex
\section{Algorithms on Succinct Permutation Graphs}
\label{sec:algorithms}
\label{app:algorithms}

Clearly, \GNextNeighbor is equivalent to an adjacency-list based representation
of a graph, so our succinct data structures can replace them in standard graph algorithms, like traversals.
Beyond that, there are a few more properties specific to PGs that known algorithms
for this class build on and which are not reflected in our list of standard operations.
Fortunately, as we will show in the following, our data structures are capable of
providing this more specialized access, as well;
we formulate these as remarks for later reference.

\begin{remark}[Transitive orientations \& topological sort]
\label{rem:directed-PG}
	A graph is a comparability graph iff it admits a transitive orientation,
	\ie, an orientation of all its edges so that if there is a directed path from $u$ to $v$,
	we must also have the ``shortcut edge'' $(u,v)$.
	In any ordered PG $G_\pi$, orienting all edges $\{u,v\}$ with $u<v$ as $(u,v)$
	yields such a transitive orientation as is immediate from the point-grid representation.
	Denote the resulting directed graph by $G^\to_\pi$.
	
	It follows that the partition of the neighborhood into $N^-(v)$ and $N^+(v)$ introduced 
	above coincides with in-neighborhood and out-neighborhood of $v$ in $G^\to_\pi$, respectively.
	Since both our data structures for PGs handle $N^-(v)$ and $N^+(v)$ separately, 
	our data structure can indeed answer \GAdjacent, \GNeighbor, \GDegree, \GDistance, and \GSPath 
	queries \wrt digraph $G^\to_\pi$ instead of $G_\pi$ at no extra cost and in the same running time.
	(Note that \GDistance and \GSPath are trivial in a transitively oriented digraph: 
	All shortest directed paths are single edges.)
	
	It is immediate from the definition that $1,\ldots,n$, \ie, 
	listing the vertices by (increasing) $x$-coordinate in the point grid,
	is a topological sort of the vertices in $G^\to_\pi$.
	It is also easy to see that the same is true for decreasing $y$-coordinate, 
	\ie, $\pi(n),\pi(n-1),\ldots,\pi(1)$ is a second topological sort of $G^\to_\pi$.
	Indeed, PGs are \emph{exactly} the comparability graphs of posets of dimension two,
	\ie, the edge set of $G^\to_\pi$ is obtained as the (set) intersection of 
	two linear orders (namely $1,\ldots,n$ and $\pi(n),\ldots,\pi(1)$).
\end{remark}

\begin{remark}[One data structure for $G$ and $\overline G$]
\label{rem:complement-PG}
	PGs are exactly the graphs where both $G$ 
	and the complement graph $\overline G$ are comparability graphs.
	That immediately implies that $\overline G$ is also a PG, when $G$ is such.
	
	We can extend our data structure with just $\Oh(n)$ additional bits of space
	so that we can also answer all queries in $\overline G$ that the data structure could
	answer for $G$; in fact, only the distance-related data structures 
	($\mathtt A_x$, $\mathtt A_y$, $\mathtt B_x$, $\mathtt B_y$ and $G_A$, $G_B$)
	need to be duplicated for $\overline G$.
\end{remark}

With these preparations, we can show how several known algorithms for PGs~\cite{McConnellSpinrad1999,Moehring1985}
can efficiently run directly on top of our data structure (without storing $G$ separately).

\paragraph{Maximum Clique \& Minimum Coloring}
While computing (the size of) a maximum clique is NP-complete for general graphs,
in comparability graphs, they can be found efficiently:
we transitively orient the graph and then find a longest (directed) path.
Note that any directed path in the transitive orientation is actually a clique in the
comparability graph. 

Since our data structures already maintain $G_\pi$ in oriented form (\wref{rem:directed-PG}),
the textbook dynamic-programming algorithm for longest paths in DAGs~\cite{SedgewickWayne2011}
suffices:
For each vertex $v$, we store the length of the longest directed path \emph{ending} 
in $v$ seen so far in an array $L[v]$.
We iterate through the vertices in a topological sort; say $v=1,\ldots,n$ (in that order).
To process vertex $v$, we iterate through its in-neighbors $N^-(v)$ and
compute $L[v] = \max \bigl( \{ L[u]+1 : u\in\N^-(v) \}\cup\{1\}\bigr) $.
Then, $\ell = \max_v L[v]$ is the length of the longest path in $G^\to_\pi$, and the path
can be compute by backtracing. The same $\ell$ vertices then form a clique in $G$.
As McConnell and Spinrad~\cite{McConnellSpinrad1999} noted, $L[v]$ is simultaneously 
a valid coloring for~$G$ with $\ell$ colors, so no larger clique can possibly exist.

The running time of above algorithm is $O(n+m)$, where $m$ is the number of edges in~$G_\pi$;
the extra space on top of our data structure is just $n$ words to store the colors.

\paragraph{Maximum Independent Set \& Minimum Clique Cover}

Clearly, a maximum independent set in $G$ is a maximum clique in $\overline G$,
and similarly, a minimum clique cover of $G$ equals a minimum coloring of $\overline G$.
As discussed in \wref{rem:complement-PG}, our data structure can without additional space
support to iterate through ${N^-}_{\overline G}(v)$, the in-neighbors of $v$ in $\overline G$,
which is enough to run the above max-clique/min-coloring algorithm on $\overline G$.

%
%
%
%
%
%
%
%
%
%
%
%
%
%
%
%
%
%
%
%
%
%
%
%
%
%
%
%
%
%
%
%
%
%
%
%
%
%
%
%
%
%
%
%

%% file: permutation-graphs-bipartite.tex
\section{Bipartite Permutation Graphs}
\label{sec:bipartite}
\label{app:bipartite}

Bipartite permutation graphs (BPGs) are permutation graphs that are also bipartite.
While our data structures for general PGs clearly apply to BPGs, their special structure
allows to substantially reduce the required space.

\begin{theorem}[Succinct BPG]
\label{thm:bipartite-ds}
	A bipartite permutation graph can be represented 
	\begin{thmenumerate}{bipartite-ds}
	\item 
		using $2n + o(n)$ bits of space
		while supporting \GAdjacent, \GDegree, \GSPathFirst in $\Oh(1)$ time and
		$\GNeighbor(v)$ in $\Oh(\GDegree(v))$ time,
	\item 
		using $5n + o(n)$ bits of space
		while supporting \GAdjacent, \GDegree, \GSPathFirst, \GDistance in $\Oh(1)$ time and
		$\GNeighbor(v)$ in $\Oh(\GDegree(v))$ time.
	\end{thmenumerate}
\end{theorem}
By iterating \GSPathFirst, we can answer $\GSPath(u,v)$ in optimal $\Oh(\GDistance(u,v)+1)$ time.

\subsection{Data Structure}

As already observed in~\cite{BazzaroGavoille2009}, BPGs consist of only $A$ and $B$ vertices.
Isolated vertices are formally of both type $A$ and $B$; thus it is convenient to assign them to the highest possible indices
and to exclude them from further discussion.
(All operations on them are trivial.)

All vertices being of type $A$ or $B$ means that \emph{every} vertex corresponds 
to a left-to-right maximum or a right-to-left minimum.
The permutation $\pi^{-1}$ thus consists of two shuffled increasing subsequences
and can be encoded using the bitvectors $\mathtt A_x$ and $\mathtt A_y$
(introduced in \wref{sec:ds-distance}) in just $2n$ bits.
We add rank and select support to both bitvectors (occupying $o(n)$ additional bits of space).
\wref{fig:bpg-example} shows an example.

\begin{figure}[htbp]
	\adjustbox{max width=\linewidth}{\scalebox{.8}{
	\begin{tikzpicture}[
			scale=.4,
			point/.style = {draw,fill=white,circle,minimum size=12pt,inner sep=0pt,font=\tiny},
			A point/.style = {point,fill=red},
			B point/.style = {point,fill=green},
			AB point/.style = {point,fill=cyan},
			graph edge/.style = {yellow,very thick},
		]
		\draw[densely dotted] (0,0) grid (41,41) ;
		\draw[->,thick] (0,0) -- (41.500000,0) ;
		\draw[->,thick] (0,0) -- (0,41.500000) ;
			\node at (1,-.5) {\tiny 1} ;
			\node at (2,-.5) {\tiny 2} ;
			\node at (3,-.5) {\tiny 3} ;
			\node at (4,-.5) {\tiny 4} ;
			\node at (5,-.5) {\tiny 5} ;
			\node at (6,-.5) {\tiny 6} ;
			\node at (7,-.5) {\tiny 7} ;
			\node at (8,-.5) {\tiny 8} ;
			\node at (9,-.5) {\tiny 9} ;
			\node at (10,-.5) {\tiny 10} ;
			\node at (11,-.5) {\tiny 11} ;
			\node at (12,-.5) {\tiny 12} ;
			\node at (13,-.5) {\tiny 13} ;
			\node at (14,-.5) {\tiny 14} ;
			\node at (15,-.5) {\tiny 15} ;
			\node at (16,-.5) {\tiny 16} ;
			\node at (17,-.5) {\tiny 17} ;
			\node at (18,-.5) {\tiny 18} ;
			\node at (19,-.5) {\tiny 19} ;
			\node at (20,-.5) {\tiny 20} ;
			\node at (21,-.5) {\tiny 21} ;
			\node at (22,-.5) {\tiny 22} ;
			\node at (23,-.5) {\tiny 23} ;
			\node at (24,-.5) {\tiny 24} ;
			\node at (25,-.5) {\tiny 25} ;
			\node at (26,-.5) {\tiny 26} ;
			\node at (27,-.5) {\tiny 27} ;
			\node at (28,-.5) {\tiny 28} ;
			\node at (29,-.5) {\tiny 29} ;
			\node at (30,-.5) {\tiny 30} ;
			\node at (31,-.5) {\tiny 31} ;
			\node at (32,-.5) {\tiny 32} ;
			\node at (33,-.5) {\tiny 33} ;
			\node at (34,-.5) {\tiny 34} ;
			\node at (35,-.5) {\tiny 35} ;
			\node at (36,-.5) {\tiny 36} ;
			\node at (37,-.5) {\tiny 37} ;
			\node at (38,-.5) {\tiny 38} ;
			\node at (39,-.5) {\tiny 39} ;
			\node at (40,-.5) {\tiny 40} ;
			\node at (41,-.5) {\scriptsize $v$} ;
			\node at (-.5,1) {\tiny 1} ;
			\node at (-.5,2) {\tiny 2} ;
			\node at (-.5,3) {\tiny 3} ;
			\node at (-.5,4) {\tiny 4} ;
			\node at (-.5,5) {\tiny 5} ;
			\node at (-.5,6) {\tiny 6} ;
			\node at (-.5,7) {\tiny 7} ;
			\node at (-.5,8) {\tiny 8} ;
			\node at (-.5,9) {\tiny 9} ;
			\node at (-.5,10) {\tiny 10} ;
			\node at (-.5,11) {\tiny 11} ;
			\node at (-.5,12) {\tiny 12} ;
			\node at (-.5,13) {\tiny 13} ;
			\node at (-.5,14) {\tiny 14} ;
			\node at (-.5,15) {\tiny 15} ;
			\node at (-.5,16) {\tiny 16} ;
			\node at (-.5,17) {\tiny 17} ;
			\node at (-.5,18) {\tiny 18} ;
			\node at (-.5,19) {\tiny 19} ;
			\node at (-.5,20) {\tiny 20} ;
			\node at (-.5,21) {\tiny 21} ;
			\node at (-.5,22) {\tiny 22} ;
			\node at (-.5,23) {\tiny 23} ;
			\node at (-.5,24) {\tiny 24} ;
			\node at (-.5,25) {\tiny 25} ;
			\node at (-.5,26) {\tiny 26} ;
			\node at (-.5,27) {\tiny 27} ;
			\node at (-.5,28) {\tiny 28} ;
			\node at (-.5,29) {\tiny 29} ;
			\node at (-.5,30) {\tiny 30} ;
			\node at (-.5,31) {\tiny 31} ;
			\node at (-.5,32) {\tiny 32} ;
			\node at (-.5,33) {\tiny 33} ;
			\node at (-.5,34) {\tiny 34} ;
			\node at (-.5,35) {\tiny 35} ;
			\node at (-.5,36) {\tiny 36} ;
			\node at (-.5,37) {\tiny 37} ;
			\node at (-.5,38) {\tiny 38} ;
			\node at (-.5,39) {\tiny 39} ;
			\node at (-.5,40) {\tiny 40} ;
			\node[anchor=east] at (0,41) {\scriptsize $\pi^{-1}(v)$} ;
		\node[A point] (v1) at (1,3) {1} ;
		\node[A point] (v2) at (2,4) {2} ;
		\node[B point] (v3) at (3,1) {3} ;
		\node[A point] (v4) at (4,5) {4} ;
		\node[B point] (v5) at (5,2) {5} ;
		\node[A point] (v6) at (6,13) {6} ;
		\node[A point] (v7) at (7,14) {7} ;
		\node[B point] (v8) at (8,6) {8} ;
		\node[A point] (v9) at (9,15) {9} ;
		\node[B point] (v10) at (10,7) {10} ;
		\node[A point] (v11) at (11,16) {11} ;
		\node[A point] (v12) at (12,18) {12} ;
		\node[B point] (v13) at (13,8) {13} ;
		\node[B point] (v14) at (14,9) {14} ;
		\node[B point] (v15) at (15,10) {15} ;
		\node[A point] (v16) at (16,22) {16} ;
		\node[A point] (v17) at (17,23) {17} ;
		\node[B point] (v18) at (18,11) {18} ;
		\node[B point] (v19) at (19,12) {19} ;
		\node[B point] (v20) at (20,17) {20} ;
		\node[A point] (v21) at (21,25) {21} ;
		\node[A point] (v22) at (22,28) {22} ;
		\node[B point] (v23) at (23,19) {23} ;
		\node[A point] (v24) at (24,34) {24} ;
		\node[A point] (v25) at (25,37) {25} ;
		\node[B point] (v26) at (26,20) {26} ;
		\node[B point] (v27) at (27,21) {27} ;
		\node[A point] (v28) at (28,38) {28} ;
		\node[B point] (v29) at (29,24) {29} ;
		\node[B point] (v30) at (30,26) {30} ;
		\node[B point] (v31) at (31,27) {31} ;
		\node[A point] (v32) at (32,39) {32} ;
		\node[B point] (v33) at (33,29) {33} ;
		\node[B point] (v34) at (34,30) {34} ;
		\node[B point] (v35) at (35,31) {35} ;
		\node[B point] (v36) at (36,32) {36} ;
		\node[B point] (v37) at (37,33) {37} ;
		\node[B point] (v38) at (38,35) {38} ;
		\node[A point] (v39) at (39,40) {39} ;
		\node[B point] (v40) at (40,36) {40} ;
		\begin{pgfonlayer}{background}
			\draw[line width=2pt,black!20] (0,0) -- (41,41) ;
			\draw[graph edge] (v1) to (v3) ; 
			\draw[graph edge] (v1) to (v5) ; 
			\draw[graph edge] (v2) to (v3) ; 
			\draw[graph edge] (v2) to (v5) ; 
			\draw[graph edge] (v4) to (v5) ; 
			\draw[graph edge] (v6) to (v8) ; 
			\draw[graph edge] (v6) to (v10) ; 
			\draw[graph edge] (v6) to (v13) ; 
			\draw[graph edge] (v6) to (v14) ; 
			\draw[graph edge] (v6) to (v15) ; 
			\draw[graph edge] (v6) to (v18) ; 
			\draw[graph edge] (v6) to (v19) ; 
			\draw[graph edge] (v7) to (v8) ; 
			\draw[graph edge] (v7) to (v10) ; 
			\draw[graph edge] (v7) to (v13) ; 
			\draw[graph edge] (v7) to (v14) ; 
			\draw[graph edge] (v7) to (v15) ; 
			\draw[graph edge] (v7) to (v18) ; 
			\draw[graph edge] (v7) to (v19) ; 
			\draw[graph edge] (v9) to (v10) ; 
			\draw[graph edge] (v9) to (v13) ; 
			\draw[graph edge] (v9) to (v14) ; 
			\draw[graph edge] (v9) to (v15) ; 
			\draw[graph edge] (v9) to (v18) ; 
			\draw[graph edge] (v9) to (v19) ; 
			\draw[graph edge] (v11) to (v13) ; 
			\draw[graph edge] (v11) to (v14) ; 
			\draw[graph edge] (v11) to (v15) ; 
			\draw[graph edge] (v11) to (v18) ; 
			\draw[graph edge] (v11) to (v19) ; 
			\draw[graph edge] (v12) to (v13) ; 
			\draw[graph edge] (v12) to (v14) ; 
			\draw[graph edge] (v12) to (v15) ; 
			\draw[graph edge] (v12) to (v18) ; 
			\draw[graph edge] (v12) to (v19) ; 
			\draw[graph edge] (v12) to (v20) ; 
			\draw[graph edge] (v16) to (v18) ; 
			\draw[graph edge] (v16) to (v19) ; 
			\draw[graph edge] (v16) to (v20) ; 
			\draw[graph edge] (v16) to (v23) ; 
			\draw[graph edge] (v16) to (v26) ; 
			\draw[graph edge] (v16) to (v27) ; 
			\draw[graph edge] (v17) to (v18) ; 
			\draw[graph edge] (v17) to (v19) ; 
			\draw[graph edge] (v17) to (v20) ; 
			\draw[graph edge] (v17) to (v23) ; 
			\draw[graph edge] (v17) to (v26) ; 
			\draw[graph edge] (v17) to (v27) ; 
			\draw[graph edge] (v21) to (v23) ; 
			\draw[graph edge] (v21) to (v26) ; 
			\draw[graph edge] (v21) to (v27) ; 
			\draw[graph edge] (v21) to (v29) ; 
			\draw[graph edge] (v22) to (v23) ; 
			\draw[graph edge] (v22) to (v26) ; 
			\draw[graph edge] (v22) to (v27) ; 
			\draw[graph edge] (v22) to (v29) ; 
			\draw[graph edge] (v22) to (v30) ; 
			\draw[graph edge] (v22) to (v31) ; 
			\draw[graph edge] (v24) to (v26) ; 
			\draw[graph edge] (v24) to (v27) ; 
			\draw[graph edge] (v24) to (v29) ; 
			\draw[graph edge] (v24) to (v30) ; 
			\draw[graph edge] (v24) to (v31) ; 
			\draw[graph edge] (v24) to (v33) ; 
			\draw[graph edge] (v24) to (v34) ; 
			\draw[graph edge] (v24) to (v35) ; 
			\draw[graph edge] (v24) to (v36) ; 
			\draw[graph edge] (v24) to (v37) ; 
			\draw[graph edge] (v25) to (v26) ; 
			\draw[graph edge] (v25) to (v27) ; 
			\draw[graph edge] (v25) to (v29) ; 
			\draw[graph edge] (v25) to (v30) ; 
			\draw[graph edge] (v25) to (v31) ; 
			\draw[graph edge] (v25) to (v33) ; 
			\draw[graph edge] (v25) to (v34) ; 
			\draw[graph edge] (v25) to (v35) ; 
			\draw[graph edge] (v25) to (v36) ; 
			\draw[graph edge] (v25) to (v37) ; 
			\draw[graph edge] (v25) to (v38) ; 
			\draw[graph edge] (v25) to (v40) ; 
			\draw[graph edge] (v28) to (v29) ; 
			\draw[graph edge] (v28) to (v30) ; 
			\draw[graph edge] (v28) to (v31) ; 
			\draw[graph edge] (v28) to (v33) ; 
			\draw[graph edge] (v28) to (v34) ; 
			\draw[graph edge] (v28) to (v35) ; 
			\draw[graph edge] (v28) to (v36) ; 
			\draw[graph edge] (v28) to (v37) ; 
			\draw[graph edge] (v28) to (v38) ; 
			\draw[graph edge] (v28) to (v40) ; 
			\draw[graph edge] (v32) to (v33) ; 
			\draw[graph edge] (v32) to (v34) ; 
			\draw[graph edge] (v32) to (v35) ; 
			\draw[graph edge] (v32) to (v36) ; 
			\draw[graph edge] (v32) to (v37) ; 
			\draw[graph edge] (v32) to (v38) ; 
			\draw[graph edge] (v32) to (v40) ; 
			\draw[graph edge] (v39) to (v40) ; 
		\end{pgfonlayer}
			\node at (1,-1.5) {\tiny\ttfamily 1} ;
			\node at (2,-1.5) {\tiny\ttfamily 1} ;
			\node at (3,-1.5) {\tiny\ttfamily 0} ;
			\node at (4,-1.5) {\tiny\ttfamily 1} ;
			\node at (5,-1.5) {\tiny\ttfamily 0} ;
			\node at (6,-1.5) {\tiny\ttfamily 1} ;
			\node at (7,-1.5) {\tiny\ttfamily 1} ;
			\node at (8,-1.5) {\tiny\ttfamily 0} ;
			\node at (9,-1.5) {\tiny\ttfamily 1} ;
			\node at (10,-1.5) {\tiny\ttfamily 0} ;
			\node at (11,-1.5) {\tiny\ttfamily 1} ;
			\node at (12,-1.5) {\tiny\ttfamily 1} ;
			\node at (13,-1.5) {\tiny\ttfamily 0} ;
			\node at (14,-1.5) {\tiny\ttfamily 0} ;
			\node at (15,-1.5) {\tiny\ttfamily 0} ;
			\node at (16,-1.5) {\tiny\ttfamily 1} ;
			\node at (17,-1.5) {\tiny\ttfamily 1} ;
			\node at (18,-1.5) {\tiny\ttfamily 0} ;
			\node at (19,-1.5) {\tiny\ttfamily 0} ;
			\node at (20,-1.5) {\tiny\ttfamily 0} ;
			\node at (21,-1.5) {\tiny\ttfamily 1} ;
			\node at (22,-1.5) {\tiny\ttfamily 1} ;
			\node at (23,-1.5) {\tiny\ttfamily 0} ;
			\node at (24,-1.5) {\tiny\ttfamily 1} ;
			\node at (25,-1.5) {\tiny\ttfamily 1} ;
			\node at (26,-1.5) {\tiny\ttfamily 0} ;
			\node at (27,-1.5) {\tiny\ttfamily 0} ;
			\node at (28,-1.5) {\tiny\ttfamily 1} ;
			\node at (29,-1.5) {\tiny\ttfamily 0} ;
			\node at (30,-1.5) {\tiny\ttfamily 0} ;
			\node at (31,-1.5) {\tiny\ttfamily 0} ;
			\node at (32,-1.5) {\tiny\ttfamily 1} ;
			\node at (33,-1.5) {\tiny\ttfamily 0} ;
			\node at (34,-1.5) {\tiny\ttfamily 0} ;
			\node at (35,-1.5) {\tiny\ttfamily 0} ;
			\node at (36,-1.5) {\tiny\ttfamily 0} ;
			\node at (37,-1.5) {\tiny\ttfamily 0} ;
			\node at (38,-1.5) {\tiny\ttfamily 0} ;
			\node at (39,-1.5) {\tiny\ttfamily 1} ;
			\node at (40,-1.5) {\tiny\ttfamily 0} ;
			\node at (0,-1.5) {\scriptsize$\texttt A_x$} ;
			\node at (-1.5,1) {\tiny\ttfamily 0} ;
			\node at (-1.5,2) {\tiny\ttfamily 0} ;
			\node at (-1.5,3) {\tiny\ttfamily 1} ;
			\node at (-1.5,4) {\tiny\ttfamily 1} ;
			\node at (-1.5,5) {\tiny\ttfamily 1} ;
			\node at (-1.5,6) {\tiny\ttfamily 0} ;
			\node at (-1.5,7) {\tiny\ttfamily 0} ;
			\node at (-1.5,8) {\tiny\ttfamily 0} ;
			\node at (-1.5,9) {\tiny\ttfamily 0} ;
			\node at (-1.5,10) {\tiny\ttfamily 0} ;
			\node at (-1.5,11) {\tiny\ttfamily 0} ;
			\node at (-1.5,12) {\tiny\ttfamily 0} ;
			\node at (-1.5,13) {\tiny\ttfamily 1} ;
			\node at (-1.5,14) {\tiny\ttfamily 1} ;
			\node at (-1.5,15) {\tiny\ttfamily 1} ;
			\node at (-1.5,16) {\tiny\ttfamily 1} ;
			\node at (-1.5,17) {\tiny\ttfamily 0} ;
			\node at (-1.5,18) {\tiny\ttfamily 1} ;
			\node at (-1.5,19) {\tiny\ttfamily 0} ;
			\node at (-1.5,20) {\tiny\ttfamily 0} ;
			\node at (-1.5,21) {\tiny\ttfamily 0} ;
			\node at (-1.5,22) {\tiny\ttfamily 1} ;
			\node at (-1.5,23) {\tiny\ttfamily 1} ;
			\node at (-1.5,24) {\tiny\ttfamily 0} ;
			\node at (-1.5,25) {\tiny\ttfamily 1} ;
			\node at (-1.5,26) {\tiny\ttfamily 0} ;
			\node at (-1.5,27) {\tiny\ttfamily 0} ;
			\node at (-1.5,28) {\tiny\ttfamily 1} ;
			\node at (-1.5,29) {\tiny\ttfamily 0} ;
			\node at (-1.5,30) {\tiny\ttfamily 0} ;
			\node at (-1.5,31) {\tiny\ttfamily 0} ;
			\node at (-1.5,32) {\tiny\ttfamily 0} ;
			\node at (-1.5,33) {\tiny\ttfamily 0} ;
			\node at (-1.5,34) {\tiny\ttfamily 1} ;
			\node at (-1.5,35) {\tiny\ttfamily 0} ;
			\node at (-1.5,36) {\tiny\ttfamily 0} ;
			\node at (-1.5,37) {\tiny\ttfamily 1} ;
			\node at (-1.5,38) {\tiny\ttfamily 1} ;
			\node at (-1.5,39) {\tiny\ttfamily 1} ;
			\node at (-1.5,40) {\tiny\ttfamily 1} ;
			\node at (-1.25,0) {\scriptsize$\texttt A_y$} ;
	\end{tikzpicture}
	}}
	\caption{%
		An exemplary bipartite permutation graph, shown as the grid $P(\pi)$.
	}
	\label{fig:bpg-example}
\end{figure}
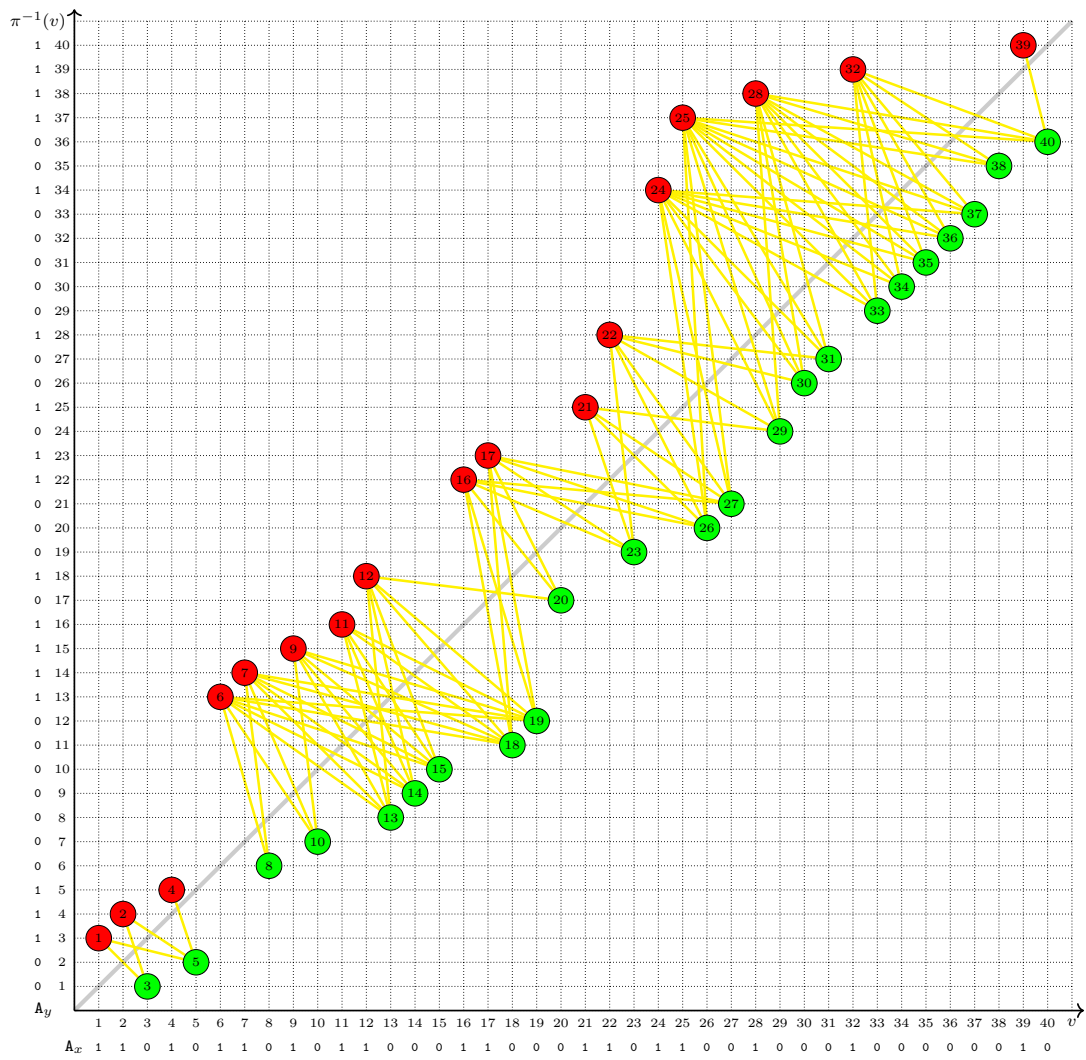

The key operation is to simulate access to $\pi^{-1}(v)$ based on the above representation:
\begin{align*}
		\pi^{-1}(v)
	&\wwrel=
		\begin{cases*}
		\selop_1(\mathtt A_y, \rankop_1(\mathtt A_x, v)) & if $\mathtt A_x[v]=1$ \\
		\selop_0(\mathtt A_y, \rankop_0(\mathtt A_x, v)) & if $\mathtt A_x[v]=0$ \\
		\end{cases*}
\end{align*}
Computation of $\pi^{-1}$ is thus supported in constant time.
That immediately allows to compute $\GAdjacent(u,v)$ as before;
moreover, 
$a^-(v)$, $a^+(v)$, are directly supported, too.
For $b^-(v)$, $b^+(v)$, we exploit that in BPGs, $\mathtt B_x[v] = 1-\mathtt A_x[v]$ so 
$b^+(v) = \selop_0(\mathtt A_x, \rankop_0(\mathtt A_y, \pi^{-1}(v)))$,
and similarly for $b^-(v)$.

It is easy to see that for a $B$-vertex $v$, its neighbors are exactly 
all $A$-vertices in $[a^-(v),a^+(v)]$; similarly for $A$-vertex $v$, we have
$N(v) = [b^-(v),b^+(v)]\cap B$.
We can iterate through these (in sorted order) using rank/select on $\mathtt A_x$,
so \GNeighbor can be answered in constant time per neighbor.

The degree of a vertex can computed in $\Oh(1)$ time. If $v$ is a $B$-vertex,
$\GDegree(v) = \rankop_1(\mathtt A_x, a^+(v)) - \rankop_1(\mathtt A_x, a^-(v))-1$,
and similarly for an $A$-vertex.

Finally, shortest paths in BPGs are particularly simple since there is only
one candidate successor vertex left:
Let $u<v$ and assume $u$ is an $A$-vertex.
Then either $u$ and $v$ are adjacent, or $\GSPathFirst(u,v) = b^+(u)$.
The situation where $u$ is a $B$-vertex is symmetric.
Computing $\GDistance(u,v)$ faster than $\Theta(\GDistance(u,v))$
seems only possible using the distance oracles for $G_A$ and $G_B$,
which require $3n+o(n)$ additional bits of space.
The query itself is as for general PGs.

\medskip\noindent
This concludes the proof of \wref{thm:bipartite-ds}.

\subsection{Space Lower Bound}
\label{sec:lower-bound-bipartite}

A known counting result for unlabeled BPGs implies that our data structure 
from \wref{thm:bipartite-ds} is succinct.
Let us denote by $b_n$ the number of unlabeled BPGs and 
by $\overline b_n$ the number of unlabeled \emph{connected} BPGs. 
Saitoh et al.~\cite[Thm.\,3.14]{SaitohOtachiYamanakaUehara2012} showed that
\begin{align*}
		\overline b_n 
	&\wwrel= 
		\begin{cases}
			\frac14\bigl(C_{n-1} + C_{n/2-1} + \binom n{n/2}\bigr) & \text{if $n$ is even} \\
			\frac14\bigl(C_{n-1} + \binom {n-1}{(n-1)/2}\bigr) & \text{if $n$ is odd} \\
		\end{cases} 
\\	&\wwrel=
		C_{n-2}(1+o(1)),
\end{align*}
for $n\ge 2$, where $C_n$ is the $n$th Catalan number.
Hence $\lg b_n \ge \lg \overline b_n = 2n - O(\log n)$ bits 
are necessary to represent an unlabeled BPG on $n$ vertices.
This is asymptotically equivalent to the amount of space required by our data structure.

\subsection{Algorithms}

Our data structure for BPGs can be used to solve the \textsc{Hamiltonian Path} and the \textsc{Hamiltonian Cycle} problems
in $\Oh(n)$ time with no extra space.
A Hamilton path (resp. Hamiltonian cycle) in a graph is a simple path (resp. simple cycle) which contains every vertex of the graph. 
Given a graph $G$, the \textsc{Hamiltonian Path} (resp.\ \textsc{Hamiltonian Cycle}) problem asks whether the graph $G$ contains a Hamiltonian path 
(resp.\ Hamiltonian cycle).
These problems are NP-complete even when restricted to several special classes of bipartite graphs, but can be solved efficiently in the class
of BPGs (see~\cite{spinrad1987bipartite} and references therein).
We will show how our data structure can be used to execute the algorithms from \cite{spinrad1987bipartite} in $\Oh(n)$ time without using extra
space.

In order to explain the algorithms and their execution on the data structure, we need to introduce some preliminaries from \cite{spinrad1987bipartite}.
A \emph{strong ordering} of the vertices of a bipartite graph $G = (A, B, E)$ consists of an ordering of $A$ and an ordering of $B$ 
such that for all $\{a, b\}$, $\{a', b'\}$ in $E$, where $a$, $a'$ are in $A$ and $b$, $b'$ are in $B$, $a < a’$ and $b > b’$ imply $\{a, b'\}$ 
and $\{a', b\}$ are in $E$. 
The algorithms are based on the following characterization of BPGs.
\begin{theorem}[Strong ordering, \cite{spinrad1987bipartite}]\label{th:BPcharacterization}
	A graph $G=(A,B,E)$ is BPG if and only if there exists a strong ordering of $A \cup B$.
\end{theorem}

Let $G=(A,B,E)$ be a BPG, where $A = \{ a_1, a_2, \ldots, a_k \}$, $B = \{ b_1, b_2, \ldots, b_s \}$, and
the vertices are indexed according to a strong ordering of $A \cup B$. Then using the characterization from Theorem \ref{th:BPcharacterization},
the following results were proved in \cite{spinrad1987bipartite}.

\begin{theorem}[Hamiltonian path, \cite{spinrad1987bipartite}]\label{thm:BPHP} 
Graph $G$ contains a Hamiltonian path if and only if 
\begin{itemize}
	\item either $s = k-1$ and $a_1, b_1, a_2, b_2, \ldots, b_{k-1}, a_k$ is a Hamiltonian path,
	\item or $s = k$ and $a_1, b_1, a_2, b_2, \ldots, b_{k-1}, a_k, b_k$ is a Hamiltonian path,
	\item or $s = k+1$ $b_1, a_1, b_2, a_2, \ldots, a_{k}, b_{k+1}$ is a Hamiltonian path,
	\item or $s = k$ and $b_1, a_1, b_2, a_2, \ldots, a_{k-1}, b_k, a_k$ is a Hamiltonian path.
\end{itemize}
\end{theorem}

\begin{theorem}[Hamiltonian cycle, \cite{spinrad1987bipartite}]\label{thm:BPHC} 
Graph $G$ contains a Hamiltonian cycle if and only if $k = s \geq 2$ and $a_i, b_i, a_{i+1}, b_{i+1}$ is a cycle of length four for $1 \leq i \leq k-1$.
\end{theorem}

In order to make use of these results, we will show that in our data structure, vertices of a given ordered BPG are stored in a \emph{strong ordering}.
Recall, that given a permutation $\pi:[n] \to [n]$, the ordered PG induced by $\pi$, denoted $G_\pi = (V,E)$, 
has vertices $V=[n]$ and edges $\{i,j\}\in E$ for all $i>j$ with $\pi^{-1}(i)<\pi^{-1}(j)$.

\begin{claim}\label{cl:strongOrdering}
	Let $G_{\pi}=(A,B,E)$ be an ordered BPG, then the ordering $1 < 2 < \ldots < n-1 < n$ 
	(restricted to $A$ and $B$, respectively) is a strong ordering of $A \cup B$.
\end{claim}
\begin{proof}
	As before, we assume that $A$ is the set of $A$-vertices and $B$ is the set of $B$-vertices of $G$.
	Let $a,a' \in A$ and $b,b' \in B$ be such that $\{ a, b \}$ and $\{ a', b' \}$ are in $E$, and $a < a'$ and $b > b'$.
	We will show that in this case $\{ a, b' \}$ and $\{ a', b \}$ are also in $E$. 
	By definition, we need to establish: 
	\begin{enumerate}
		\item[(1)] $a < b'$ and $\pi^{-1}(a) > \pi^{-1}(b')$; and
		\item[(2)] $a' < b$ and $\pi^{-1}(a') > \pi^{-1}(b)$.
	\end{enumerate}
	We will show only (1), as (2) is proved similarly.
	Since $\{ a', b' \} \in E$ and $a'$ is an $A$-vertex, we have that $a' < b'$ and hence $a < b'$ (as, by assumption, $a < a'$).
	To prove the second part of (1), we note that $\pi^{-1}(b) < \pi^{-1}(a)$ and $b > a$ because $\{ a, b \} \in E$.
	Furthermore, since $\{ b', b \} \not\in E$ and $b' < b$, we have that $\pi^{-1}(b') < \pi^{-1}(b)$. Consequently, $\pi^{-1}(b') < \pi^{-1}(a)$.
\end{proof}

\paragraph{Hamiltonian Path}
Using \wref{thm:BPHP} and Claim \ref{cl:strongOrdering} the problem can be solved by going in constant time from the
first $A$-vertex $a_1$ to its first $B$-neighbor $b_1 = b^-(a_1)$, then going in constant time from $b_1$ to its first $A$-neighbor 
$a_2 = a^-(b_1)$, and so on until we can no longer move. If we made $n$ moves, then we have visited all the vertices of the graph following a Hamiltonian path. 
Otherwise, we try to do the same but this time starting from the first $B$-vertex. 
Similarly, if we made $n$ moves, then the graph has a Hamiltonian path. 
If both attempts fail, the graph does not contain a Hamiltonian path. 
This algorithm works in $\Oh(n)$ time.

\paragraph{Hamiltonian Cycle}
First, we check that the number of $A$-vertices is equal to the number of $B$-vertices. 
If so, we check next if the graph contains a 
Hamiltonian path using the previous algorithm (this will ensure that $A$- and $B$-vertices alternate). 
In the case of success,
at the final stage of the algorithm, we iterate through $A$-vertices following
the strong ordering, and for every $A$-vertex $a_i$ calculate in constant time the 
vertices $b_{i,1} = b^-(a_i)$, $a_{i,2} = a^-(b_{i,1})$, $b_{i,2} = b^-(a_{i,2})$ and 
check if the vertices $a_i$ and $b_{i,2}$ are adjacent, 
(\ie, whether all the four vertices induce a cycle on four vertices),
which is equivalent to $\pi^{-1}(a_i) > \pi^{-1}(b_{i,2})$. 
\wref{thm:BPHC} and \wref[Claim]{cl:strongOrdering} imply that 
the graph contains a Hamiltonian cycle if and only if 
all stages of the algorithm were successful. 
Overall, the algorithm works in $\Oh(n)$ time.

\FloatBarrier

%% file: permutation-graphs-circular.tex

\section{Circular Permutation Graphs}
\label{sec:circular}
\label{app:circular}

Circular permutation graphs (CPGs) are a natural generalization of PGs
first introduced by Rotem and Urrutia~\cite{RotemUrrutia1982}.
In this section, we show how to extend our data structure to CPGs.

\subsection{Preliminaries}

CPGs results from PGs by allowing \emph{circular/cyclic} permutation diagrams, 
\ie, in the intersecting chords representation, we connect the right and left end of the 
gray ribbon to form a cylinder.
The cylinder can be smoothly transformed into two concentric circles with chords 
in the annular region between them; \wref{fig:small-circular-pg} shows an example.

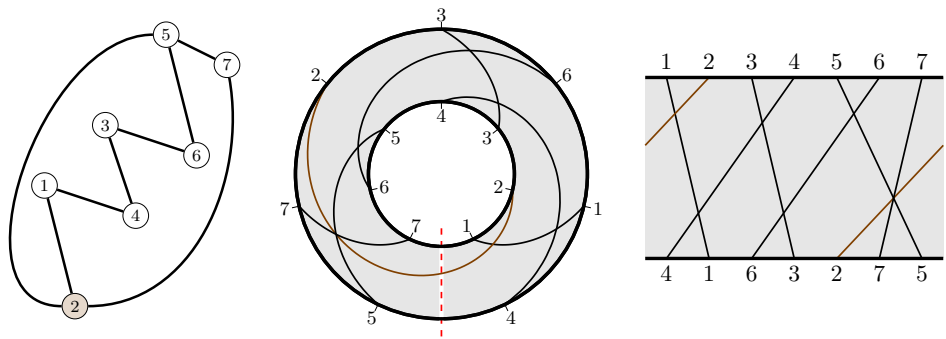
\begin{figure}[htbp]
	\centering
	\adjustbox{max width=\linewidth}{\scalebox{.8}{\kern-3em%
		\begin{tikzpicture}[
				scale=.5,
				baseline=10ex,
				point/.style = {draw,fill=white,circle,minimum size=12pt,inner sep=0pt,font=\scriptsize},
				copy point/.style = {point,fill=black!30,draw=black!80,text=black!80},
				graph edge/.style = {yellow,very thick},
			]
				\node[point] (v1) at (8,9) {1} ;
				\node[point,fill=orange!50!black!20] (v2) at (9,5) {2} ;
				\node[point] (v3) at (10,11) {3} ;
				\node[point] (v4) at (11,8) {4} ;
				\node[point] (v5) at (12,14) {5} ;
				\node[point] (v6) at (13,10) {6} ;
				\node[point] (v7) at (14,13) {7} ;
				\draw[very thick] (v5) to[in=180,out=180] (v2) ; 
				\draw[very thick] (v7) to[in=0,out=-80] (v2) ; 
				\draw[very thick] (v1) to (v2) ; 
				\draw[very thick] (v1) to (v4) ; 
				\draw[very thick] (v3) to (v4) ; 
				\draw[very thick] (v3) to (v6) ; 
				\draw[very thick] (v5) to (v6) ; 
				\draw[very thick] (v5) to (v7) ; 
		\end{tikzpicture}\quad
		\begin{tikzpicture}[
				normal/.style = {thick,black},
				forward/.style = {thick,green!50!black},
				backward/.style = {thick,orange!50!black},
			]
			\begin{scope}[xscale=.7,yscale=3,]
			\fill[black!10] (0.5,0) rectangle ++(7,-1) ;
			\draw[normal] (1,0) to (2,-1) ;
			\node[above] at (1,0) {1} ;
			\node[below] at (2,-1) {1} ;
			\draw[backward] (2,0) to (0.5,-0.375000) ;
			\draw[backward] (7+.5,-0.375000) to (5,-1) ;
			\node[above] at (2,0) {2} ;
			\node[below] at (5,-1) {2} ;
			\draw[normal] (3,0) to (4,-1) ;
			\node[above] at (3,0) {3} ;
			\node[below] at (4,-1) {3} ;
			\draw[normal] (4,0) to (1,-1) ;
			\node[above] at (4,0) {4} ;
			\node[below] at (1,-1) {4} ;
			\draw[normal] (5,0) to (7,-1) ;
			\node[above] at (5,0) {5} ;
			\node[below] at (7,-1) {5} ;
			\draw[normal] (6,0) to (3,-1) ;
			\node[above] at (6,0) {6} ;
			\node[below] at (3,-1) {6} ;
			\draw[normal] (7,0) to (6,-1) ;
			\node[above] at (7,0) {7} ;
			\node[below] at (6,-1) {7} ;
			\draw[ultra thick] (0.5,0) -- ++(7,0);
			\draw[ultra thick] (0.5,-1) -- ++(7,0);
			\end{scope}
			\begin{scope}[shift={(-3,-1.6)},scale=1.2]
				\draw[ultra thick,fill=black!10] (0,0) circle (2) ;
				\draw[ultra thick,fill=white] (0,0) circle (1) ;
				\draw[line width=2pt, white] (0,-.5) -- ++(0,-2);
				\draw[dashed,thick,red] (0,-.75) -- ++(0,-1.5);
				\foreach \x in {1,...,7} {
					\node at (\x*360/7-90-360/14:.8) {\smaller \x} ;
					\draw (\x*360/7-90-360/14:1) coordinate (i\x) -- (\x*360/7-90-360/14:.9) ;
				}
				\foreach \x/\y in {1/4,2/1,3/6,4/3,5/2,6/7,7/5} {
					\node at (\x*360/7-90-360/14:2.2) {\smaller \y} ;
					\draw (\x*360/7-90-360/14:2) coordinate (o\y) -- (\x*360/7-90-360/14:2.1) ;
				}
				\draw[normal] plot[smooth,samples=20,domain=-64.28:-12.85] (\x:{1+(\x+64.28)/(51.43)}) ;
				\draw[backward] plot[smooth,samples=40,domain=-218.57:-12.85] (\x:{2-(\x+218.57)/(205.72)}) ;
				\draw[normal] plot[smooth,samples=20,domain=38.57:90] (\x:{1+(\x-38.57)/(51.43)}) ;
				\draw[normal] plot[smooth,samples=20,domain=-64.28:90] (\x:{2-(\x+64.28)/(154.28)}) ;
				\draw[normal] plot[smooth,samples=20,domain=141.43:244.29] (\x:{1+(\x-141.43)/(102.85)}) ;
				\draw[normal] plot[smooth,samples=20,domain=192.85:38.57] (\x:{2-(\x-38.57)/(154.28)}) ;
				\draw[normal] plot[smooth,samples=20,domain=244.28:192.85] (\x:{2-(\x-192.85)/(51.43)}) ;
				\draw[ultra thick] (0,0) circle (2)   (0,0) circle (1) ;
			\end{scope}
		\end{tikzpicture}%
	}}
	
	\caption{%
		Small circular permutation graph on 7 vertices (left) that is not a standard permutation graph,
		shown as the intersection of chords
		between concentric circles (middle), and as intersections of chords 
		on a cylinder that has been cut open (note that chord 2 wraps around the cut).%
	}
	\label{fig:small-circular-pg}
\end{figure}

By cutting the annulus open again, we obtain the \emph{permutation diagram with crossings},
\ie, where some chords \emph{cross} the cut and continue from the opposite end;
(\wref{fig:small-circular-pg} right).
This induces a linear order of the endpoints on both circles 
(in counterclockwise direction starting at the cut) and hence a permutation $\pi:[n]\to[n]$ as before;
\eg, for \wref{fig:small-circular-pg}, we have $\pi=(4,1,6,3,2,7,5)$.
Note that for CPGs, though, $\pi$ no longer uniquely determines a graph
because chords between circles can wrap around the inner circle in clockwise or counterclockwise
direction and this affects intersections.
The representation becomes unique again upon adding an assignment of chord types $t:[n]\to\{N,F,B\}$ to $\pi$
with the following meaning:
$N$-chords do \underline not cross the cut at all. 
$F$-chords do cross the cut, namely in \underline forward direction, 
\ie, when following the chord from the upper endpoint to the lower endpoint, we move to the right.
Finally, $B$-chords also cross the cut, but in \underline backward direction, \ie, following the chord top down
moves us to the left.
A larger example with all types of crossings is shown in \wpref{fig:circular-example}.

Note that every PG is also a CPG (setting $t(v)=N$ for all vertices), so the lower bounds
from \wref{sec:lower-bounds} applies here as well.

\begin{remark}[Improper diagrams]
	The original definition of CPGs required the permutation diagram to be ``proper'', 
	meaning that no two chords intersect more than once.
	All our permutation diagrams are required to be proper in this sense.
	(Later works~\cite{Sritharan1996} achieved a similar effect by defining vertices adjacent
	iff their chords intersect \emph{exactly} once.)
	
	Note that monotonic/straight chords and forbidding double crossings of the cut are \emph{not} sufficient:
	not all combinations of $\pi$ and $t$ lead to a proper permutation diagram.
	Indeed, the pair $(\pi,t)$ is valid iff no pair of chords $u$, $v$ has
	one of the following forbidden combinations of crossing type and relative location:
	\begin{enumerate}
		\item $u<v$, $\pi^{-1}(u)>\pi^{-1(v)}$ (inversion), $t(v)=N$, and $t(u)=F$.
		\item $u<v$, $\pi^{-1}(u)>\pi^{-1(v)}$ (inversion), $t(v)=B$, and $t(u)=N$.
		\item $u<v$, $\pi^{-1}(u)<\pi^{-1(v)}$ (no inversion), and $N\ne t(v)\ne t(u) \ne N$.
	\end{enumerate}
	Each of these cases implies a double crossing
	and a chord length $>n$ after ``pulling one chord straight'' (by turning the two circles against each other).
\end{remark}

Sritharan~\cite{Sritharan1996} gave a linear-time algorithm for recognizing CPGs, 
which also computes the circular permutation diagram if the input is a CPG.

\subsection{Ordered CPGs and the Thrice-Unrolled PG}

In analogy to ordered PG $G_\pi$, we define the \emph{ordered CPG} $G_{\pi,t}$ 
for a (valid combination of) permutation $\pi:[n]\to[n]$ and 
chord types $t:[n]\to\{N,F,B\}$.

\begin{figure}[htbp]
	\centering
	\adjustbox{max width=\linewidth}{%
	\begin{tikzpicture}[
			scale=.4,
			point/.style = {draw,fill=white,circle,minimum size=12pt,inner sep=0pt,font=\tiny},
			copy point/.style = {point,fill=black!30,draw=black!80,text=black!80},
			graph edge/.style = {yellow,very thick},
			normal/.style = {thick,black},
			forward/.style = {thick,green!50!black},
			backward/.style = {thick,orange!50!black},
	]
		\fill[black!30,opacity=.2] (0,0) rectangle (7+.3,21+1);
		\fill[black!30,opacity=.2] (2*7+.7,0) rectangle (3*7+1,21+1);
		\fill[orange,opacity=.1] (0,0) rectangle (21+1,7+.3);
		\fill[green!80!black,opacity=.1] (0,2*7+.7) rectangle (21+1,3*7+1);
		\draw[densely dotted] (0,0) grid (22,22) ;
		\draw[->,thick] (0,0) -- (22.500000,0) ;
		\draw[->,thick] (0,0) -- (0,22.500000) ;
			\node at (1,-.5) {\tiny 1} ;
			\node at (8,-.5) {\tiny 1} ;
			\node at (15,-.5) {\tiny 1} ;
			\node at (2,-.5) {\tiny 2} ;
			\node at (9,-.5) {\tiny 2} ;
			\node at (16,-.5) {\tiny 2} ;
			\node at (3,-.5) {\tiny 3} ;
			\node at (10,-.5) {\tiny 3} ;
			\node at (17,-.5) {\tiny 3} ;
			\node at (4,-.5) {\tiny 4} ;
			\node at (11,-.5) {\tiny 4} ;
			\node at (18,-.5) {\tiny 4} ;
			\node at (5,-.5) {\tiny 5} ;
			\node at (12,-.5) {\tiny 5} ;
			\node at (19,-.5) {\tiny 5} ;
			\node at (6,-.5) {\tiny 6} ;
			\node at (13,-.5) {\tiny 6} ;
			\node at (20,-.5) {\tiny 6} ;
			\node at (7,-.5) {\tiny 7} ;
			\node at (14,-.5) {\tiny 7} ;
			\node at (21,-.5) {\tiny 7} ;
			\node at (22,-.5) {\scriptsize $v$} ;
			\node at (-.5,1) {\tiny 1} ;
			\node at (-.5,8) {\tiny 1} ;
			\node at (-.5,15) {\tiny 1} ;
			\node at (-.5,2) {\tiny 2} ;
			\node at (-.5,9) {\tiny 2} ;
			\node at (-.5,16) {\tiny 2} ;
			\node at (-.5,3) {\tiny 3} ;
			\node at (-.5,10) {\tiny 3} ;
			\node at (-.5,17) {\tiny 3} ;
			\node at (-.5,4) {\tiny 4} ;
			\node at (-.5,11) {\tiny 4} ;
			\node at (-.5,18) {\tiny 4} ;
			\node at (-.5,5) {\tiny 5} ;
			\node at (-.5,12) {\tiny 5} ;
			\node at (-.5,19) {\tiny 5} ;
			\node at (-.5,6) {\tiny 6} ;
			\node at (-.5,13) {\tiny 6} ;
			\node at (-.5,20) {\tiny 6} ;
			\node at (-.5,7) {\tiny 7} ;
			\node at (-.5,14) {\tiny 7} ;
			\node at (-.5,21) {\tiny 7} ;
			\node[anchor=east] at (0,22) {\scriptsize $\pi^{-1}(v)$} ;
		\node[point] (v1) at (8,9) {1} ;
		\node[copy point] (l1) at (1,2) {1} ;
		\node[copy point] (r1) at (15,16) {1} ;
		\node[point] (v2) at (9,5) {2} ;
		\node[copy point] (r2) at (16,12) {2} ;
		\node[point] (v3) at (10,11) {3} ;
		\node[copy point] (l3) at (3,4) {3} ;
		\node[copy point] (r3) at (17,18) {3} ;
		\node[point] (v4) at (11,8) {4} ;
		\node[copy point] (l4) at (4,1) {4} ;
		\node[copy point] (r4) at (18,15) {4} ;
		\node[point] (v5) at (12,14) {5} ;
		\node[copy point] (l5) at (5,7) {5} ;
		\node[copy point] (r5) at (19,21) {5} ;
		\node[point] (v6) at (13,10) {6} ;
		\node[copy point] (l6) at (6,3) {6} ;
		\node[copy point] (r6) at (20,17) {6} ;
		\node[point] (v7) at (14,13) {7} ;
		\node[copy point] (l7) at (7,6) {7} ;
		\node[copy point] (r7) at (21,20) {7} ;
		\begin{pgfonlayer}{background}
			\draw[graph edge] (l1) to (l4) ; 
			\draw[graph edge] (l3) to (l4) ; 
			\draw[graph edge] (l3) to (l6) ; 
			\draw[graph edge] (l5) to (l6) ; 
			\draw[graph edge] (l5) to (l7) ; 
			\draw[graph edge] (l5) to (v2) ; 
			\draw[graph edge] (l7) to (v2) ; 
			\draw[graph edge] (v1) to (v2) ; 
			\draw[graph edge] (v1) to (v4) ; 
			\draw[graph edge] (v3) to (v4) ; 
			\draw[graph edge] (v3) to (v6) ; 
			\draw[graph edge] (v5) to (v6) ; 
			\draw[graph edge] (v5) to (v7) ; 
			\draw[graph edge] (v5) to (r2) ; 
			\draw[graph edge] (v7) to (r2) ; 
			\draw[graph edge] (r1) to (r2) ; 
			\draw[graph edge] (r1) to (r4) ; 
			\draw[graph edge] (r3) to (r4) ; 
			\draw[graph edge] (r3) to (r6) ; 
			\draw[graph edge] (r5) to (r6) ; 
			\draw[graph edge] (r5) to (r7) ; 
		\end{pgfonlayer}
		\begin{scope}[shift={(-8.5,5.5)},overlay]
			\fill[fill=white,opacity=.85,rounded corners=20pt] (10,8)++(-4,-4) rectangle ++(9,11) ;
			\node[point] (v1) at (8,9) {1} ;
			\node[point,fill=orange!50!black!20] (v2) at (9,5) {2} ;
			\node[point] (v3) at (10,11) {3} ;
			\node[point] (v4) at (11,8) {4} ;
			\node[point] (v5) at (12,14) {5} ;
			\node[point] (v6) at (13,10) {6} ;
			\node[point] (v7) at (14,13) {7} ;
			\draw[very thick] (v5) to[in=180,out=180] (v2) ; 
			\draw[very thick] (v7) to[in=0,out=-80] (v2) ; 
			\draw[very thick] (v1) to (v2) ; 
			\draw[very thick] (v1) to (v4) ; 
			\draw[very thick] (v3) to (v4) ; 
			\draw[very thick] (v3) to (v6) ; 
			\draw[very thick] (v5) to (v6) ; 
			\draw[very thick] (v5) to (v7) ; 
		\end{scope}
		
		\begin{scope}[shift={(0,-2.5)}, every node/.style={font=\tiny},yscale=5]
			\fill[black!10] (0.5,0) rectangle ++(7,-1) ;
			\draw[normal] (1,0) to (2,-1) ;
			\node[above] at (1,0) {1} ;
			\node[below] at (2,-1) {1} ;
			\draw[backward] (7+.5,-0.375000) to (5,-1) ;
			\node[below] at (5,-1) {2} ;
			\draw[normal] (3,0) to (4,-1) ;
			\node[above] at (3,0) {3} ;
			\node[below] at (4,-1) {3} ;
			\draw[normal] (4,0) to (1,-1) ;
			\node[above] at (4,0) {4} ;
			\node[below] at (1,-1) {4} ;
			\draw[normal] (5,0) to (7,-1) ;
			\node[above] at (5,0) {5} ;
			\node[below] at (7,-1) {5} ;
			\draw[normal] (6,0) to (3,-1) ;
			\node[above] at (6,0) {6} ;
			\node[below] at (3,-1) {6} ;
			\draw[normal] (7,0) to (6,-1) ;
			\node[above] at (7,0) {7} ;
			\node[below] at (6,-1) {7} ;
			\draw[ultra thick] (0.5,0) -- ++(7,0);
			\draw[ultra thick] (0.5,-1) -- ++(7,0);
		\end{scope}
		\begin{scope}[shift={(14,-2.5)}, every node/.style={font=\tiny},yscale=5]
			\fill[black!10] (0.5,0) rectangle ++(7,-1) ;
			\draw[normal] (1,0) to (2,-1) ;
			\node[above] at (1,0) {1} ;
			\node[below] at (2,-1) {1} ;
			\draw[backward] (2,0) to (0.5,-0.375000) ;
			\node[above] at (2,0) {2} ;
			\draw[normal] (3,0) to (4,-1) ;
			\node[above] at (3,0) {3} ;
			\node[below] at (4,-1) {3} ;
			\draw[normal] (4,0) to (1,-1) ;
			\node[above] at (4,0) {4} ;
			\node[below] at (1,-1) {4} ;
			\draw[normal] (5,0) to (7,-1) ;
			\node[above] at (5,0) {5} ;
			\node[below] at (7,-1) {5} ;
			\draw[normal] (6,0) to (3,-1) ;
			\node[above] at (6,0) {6} ;
			\node[below] at (3,-1) {6} ;
			\draw[normal] (7,0) to (6,-1) ;
			\node[above] at (7,0) {7} ;
			\node[below] at (6,-1) {7} ;
			\draw[ultra thick] (0.5,0) -- ++(7,0);
			\draw[ultra thick] (0.5,-1) -- ++(7,0);
		\end{scope}
		\begin{scope}[shift={(7,-2.5)}, every node/.style={font=\tiny},yscale=5]
			\fill[black!10] (0.5,0) rectangle ++(7,-1) ;
			\draw[normal] (1,0) to (2,-1) ;
			\node[above] at (1,0) {1} ;
			\node[below] at (2,-1) {1} ;
			\draw[backward] (2,0) to (0.5,-0.375000) ;
			\draw[backward] (7+.5,-0.375000) to (5,-1) ;
			\node[above] at (2,0) {2} ;
			\node[below] at (5,-1) {2} ;
			\draw[normal] (3,0) to (4,-1) ;
			\node[above] at (3,0) {3} ;
			\node[below] at (4,-1) {3} ;
			\draw[normal] (4,0) to (1,-1) ;
			\node[above] at (4,0) {4} ;
			\node[below] at (1,-1) {4} ;
			\draw[normal] (5,0) to (7,-1) ;
			\node[above] at (5,0) {5} ;
			\node[below] at (7,-1) {5} ;
			\draw[normal] (6,0) to (3,-1) ;
			\node[above] at (6,0) {6} ;
			\node[below] at (3,-1) {6} ;
			\draw[normal] (7,0) to (6,-1) ;
			\node[above] at (7,0) {7} ;
			\node[below] at (6,-1) {7} ;
			\draw[ultra thick] (0.5,0) -- ++(7,0);
			\draw[ultra thick] (0.5,-1) -- ++(7,0);
			\draw[line width=1pt, white] (0.5,0.1) -- ++(0,-1.2);
			\draw[thin,dashed,red] (0.5,0.1) -- ++(0,-1.2);
			\draw[line width=1pt, white] (7.5,0.1) -- ++(0,-1.2);
			\draw[thin,dashed,red] (7.5,0.1) -- ++(0,-1.2);
		\end{scope}
	\end{tikzpicture}%
	}
	
	\caption{%
		The circular permutation graph from \wref{fig:small-circular-pg} and its thrice-unrolled PG $G_3$ as
		a permutation diagram and in the grid representation.%
	}
	\label{fig:cpg-unrolled}
\end{figure}
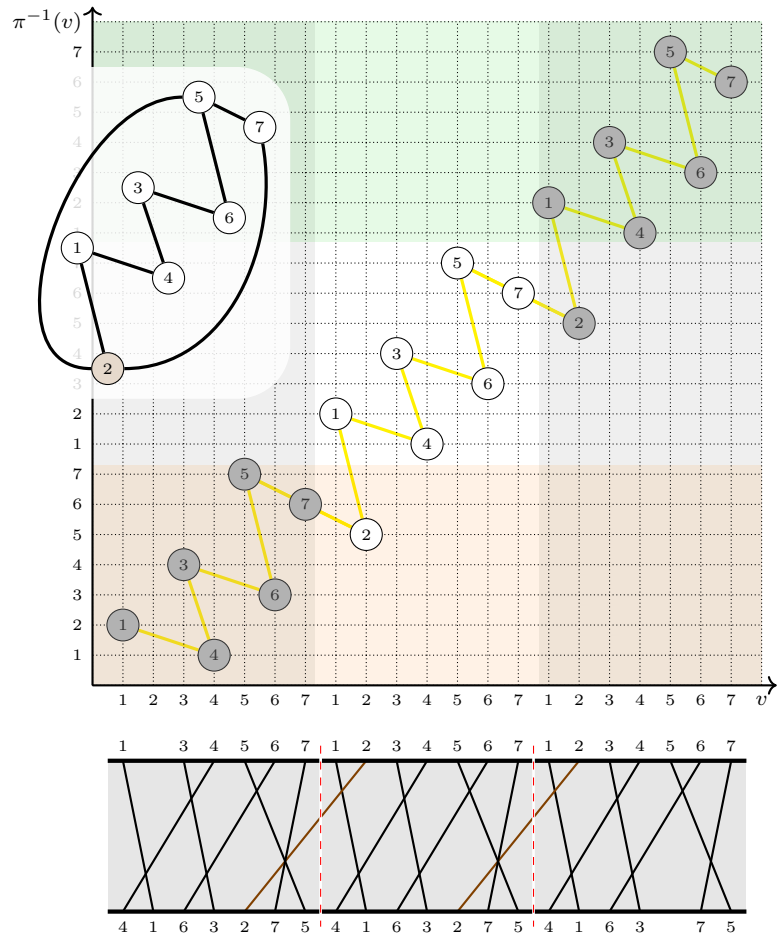

From now on, we assume such a graph $G_{\pi,t}$ is given.
In preparation of our succinct data structure for CPGs,
we again define a planar point set based on which we support all queries:
\begin{align*}
		P(\pi,t)
	&\wwrel=\phantom{\bin\cup}
		\bigl\{(v+k n,\,\pi^{-1}(v)+k n) : v \in [n], k\in\{0,1,2\}, t(v) = N\bigr\} 
\\	&\wwrel\ppe \bin\cup 
		\bigl\{(v+k n,\,\pi^{-1}(v)+(k+1) n) : v \in [n], k\in\{0,1\}, t(v) = F\bigr\} 
\\	&\wwrel\ppe \bin\cup 
		\bigl\{(v+k n,\,\pi^{-1}(v)+(k-1) n) : v \in [n], k\in\{1,2\}, t(v) = B\bigr\} 
\end{align*}
$P(\pi,t)$ lies in a $3n\times 3n$ grid and $2n\le |P(\pi,t)|\le 3n$. 
Intuitively, $P(\pi,t)$ is obtained by
\emph{unrolling} the circular permutation diagram of $G_{\pi,t}$
\emph{three times}: We record the times at which we see a chord's endpoints
during this unrolling process and output a point for these times.
We only output chords when we have seen both endpoints during this process, 
so each noncrossing chord is output three times, whereas the crossing chords
are only present twice.
See \wref{fig:cpg-unrolled} for an example.

Clearly, $P(\pi,t)$ corresponds to the grid representation of a (larger) 
PG, denoted by $G_3 = G_3(\pi,t)$, which ``contains'' $G_{\pi,t}$
in the sense detailed in \wref{lem:cpg-neighborhood} below.
Denote the vertices in $G_3$ by $\ell_j$, $c_j$, and $r_j$, $j\in[n]$, 
respectively, 
where $\ell_j$ is the vertex corresponding to point $(x,y)$ with $x=j \in[n]$,
$c_j$ is the vertex for $(x,y)$ with $x=j+n \in (n..2n]$ and
$r_j$ is the vertex for $(x,y)$ with $x=j+2n\in(2n..3n]$.
(Note that in general not all $\ell_j$ (resp.\ $r_j$) will be present.)
We call $c_v$ the main copy of vertex $v$ in $G_{\pi,t}$, 
and $\ell_v$ and $r_v$ are the left (resp.\ right) copies of $v$.

\begin{lemma}[Neighborhood from $G_3$]
\label{lem:cpg-neighborhood}
	Let $v$ be a vertex in $G_{\pi,t}$ and $c_v$ its main copy in $G_3(\pi,t)$.
	Then $v$'s neighbors (in $G_{\pi,t}$) can be deduced from $c_v$'s neighbors in in $G_3(\pi,t)$ as follows:
	\begin{align*}
			N^-(v)
		&\wwrel=
			\{w : c_w \in N^-(c_v)\vee \ell_w \in N^-(c_v) \},
	\\
			N^+(v)
		&\wwrel=
			\{w : c_w \in N^+(c_v)\vee r_w \in N^+(c_v) \}.
	\end{align*}
\end{lemma}
\begin{proof}
	First note that by construction, any edge in $G_3$ between copies of $u$ and $v$ in $G_3$ 
	(\ie, any edge between $\ell_u$, $c_u$, $r_u$, resp.\ $\ell_v$, $c_v$, $r_v$)
	implies an edge in $G$ between $u$ and $v$.
	Hence we never report non-neighbors in the set for $N^-(v)$ and $N^+(v)$ above.
	Moreover, for any combination of $\ell$, $c$, $r$ where both copies of $u$ and $v$ exist,
	these copies are adjacent in~$G_3$.
	It remains to show that any edge in $G$ is witnessed by at least one pair of copies.
	For that, consider the permutation diagram of $G_3$ and note that
	it contains a complete copy of the permutation diagram with crossings of $G$ 
	in its middle third (see \wref{fig:cpg-unrolled}), 
	so every neighbor of $v$ in $G$ can be witnessed from $c_v$ in $G_3$.
\end{proof}

\begin{remark}[Thrice or twice?]
	It follows directly from the definition of a proper permutation diagram 
	that the upper endpoints of all backward-crossing chords must 
	precede all upper endpoints of forward-crossing chords,
	and vice versa for lower endpoints.
	As a consequence, we can remove further copies from $G_3$ without affecting \wref{lem:cpg-neighborhood};
	 one can show that at most two copies of every noncrossing chord are always sufficient.
	Since the size of $G_3$ will only affect lower-order terms of space,
	we omit this optimization here for ease of presentation.
\end{remark}

\subsection{Succinct CPGs}
With this preparation, we can now describe our succinct data structure for CPGs.
Conceptually, we store our succinct PG data structure for $G_3$ and reduce the queries to it.
For the space-dominant part, \ie, the inverse permutation $\pi^{-1}$,
we store it implicitly, exploiting the special structure of $G_3$.

\begin{theorem}[Succinct CPGs]
\label{thm:succinct-cpg}
	An (unlabeled) circular permutation graph on $n$ vertices can be represented
	using $n \lg n + \Oh(n)$ bits of space while supporting
	\GAdjacent, \GDistance, \GSPathFirst in $\Oh(1)$ time and
	$\GNeighbor(v)$ and $\GDegree(v)$ in $\Oh(\GDegree(v)+1)$ time.
\end{theorem}
As always, we can add constant-time degree support at the expense of another 
$n\lceil \lg n\rceil$ bits of space.

We are now ready to give the proof of \wref{thm:succinct-cpg}.
Let a valid pair $(\pi,t)$ be given and consider $G_{\pi,t}$.
As for PGs, we store the array $\Pi[1..n]$ with $\Pi[i] = \pi^{-1}(i)$;
additionally, we store the sequence $t = t(1),\ldots,t(n)$ over alphabet $\{N,F,B\}$
for constant-time access; (two bitvectors suffice for the claimed space).

For the operations, we will show how to simulate access to the grid representation
of $G_3$; 
the reader will find it useful to consult the larger example CPG 
in \wref{fig:circular-example} when following the description.

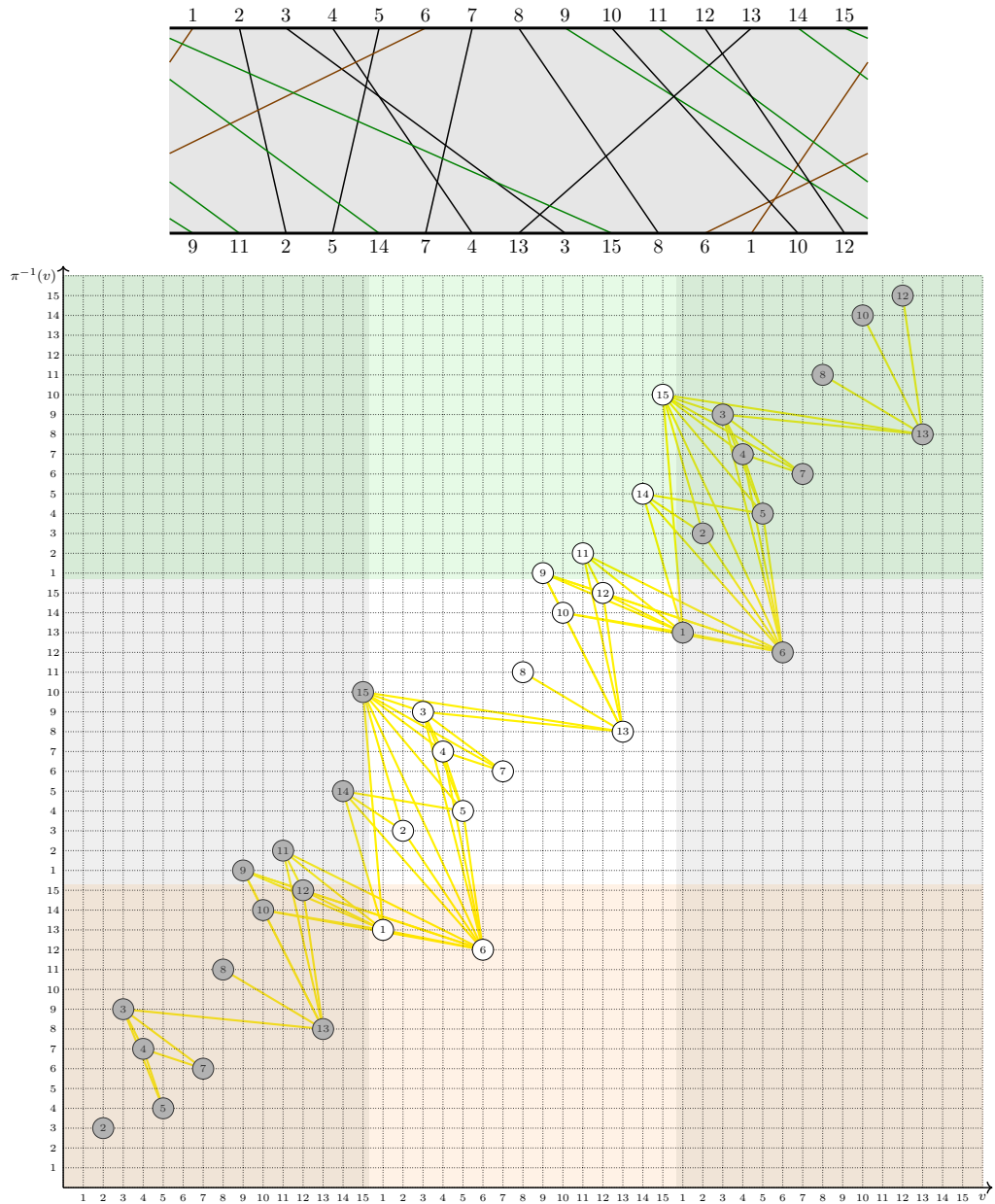
\begin{figure}[thbp]
	\ifspringer{\centering\begin{adjustbox}{varwidth=\linewidth,scale=1.1}}{}%
	{\plaincenter{\scalebox{.7}{%
	\qquad\begin{tikzpicture}[
			xscale=.9,
			yscale=4,
			normal/.style = {thick,black},
			forward/.style = {thick,green!50!black},
			backward/.style = {thick,orange!50!black},
		]
		\fill[black!10] (0.5,0) rectangle ++(15,-1) ;
		\draw[backward] (1,0) to (0.5,-0.166667) ;
		\draw[backward] (15+.5,-0.166667) to (13,-1) ;
		\node[above] at (1,0) {1} ;
		\node[below] at (13,-1) {1} ;
		\draw[normal] (2,0) to (3,-1) ;
		\node[above] at (2,0) {2} ;
		\node[below] at (3,-1) {2} ;
		\draw[normal] (3,0) to (9,-1) ;
		\node[above] at (3,0) {3} ;
		\node[below] at (9,-1) {3} ;
		\draw[normal] (4,0) to (7,-1) ;
		\node[above] at (4,0) {4} ;
		\node[below] at (7,-1) {4} ;
		\draw[normal] (5,0) to (4,-1) ;
		\node[above] at (5,0) {5} ;
		\node[below] at (4,-1) {5} ;
		\draw[backward] (6,0) to (0.5,-0.611111) ;
		\draw[backward] (15+.5,-0.611111) to (12,-1) ;
		\node[above] at (6,0) {6} ;
		\node[below] at (12,-1) {6} ;
		\draw[normal] (7,0) to (6,-1) ;
		\node[above] at (7,0) {7} ;
		\node[below] at (6,-1) {7} ;
		\draw[normal] (8,0) to (11,-1) ;
		\node[above] at (8,0) {8} ;
		\node[below] at (11,-1) {8} ;
		\draw[forward] (9,0) to (15+.5,-0.928571) ;
		\draw[forward] (0.5,-0.928571) to (1,-1) ;
		\node[above] at (9,0) {9} ;
		\node[below] at (1,-1) {9} ;
		\draw[normal] (10,0) to (14,-1) ;
		\node[above] at (10,0) {10} ;
		\node[below] at (14,-1) {10} ;
		\draw[forward] (11,0) to (15+.5,-0.750000) ;
		\draw[forward] (0.5,-0.750000) to (2,-1) ;
		\node[above] at (11,0) {11} ;
		\node[below] at (2,-1) {11} ;
		\draw[normal] (12,0) to (15,-1) ;
		\node[above] at (12,0) {12} ;
		\node[below] at (15,-1) {12} ;
		\draw[normal] (13,0) to (8,-1) ;
		\node[above] at (13,0) {13} ;
		\node[below] at (8,-1) {13} ;
		\draw[forward] (14,0) to (15+.5,-0.250000) ;
		\draw[forward] (0.5,-0.250000) to (5,-1) ;
		\node[above] at (14,0) {14} ;
		\node[below] at (5,-1) {14} ;
		\draw[forward] (15,0) to (15+.5,-0.050000) ;
		\draw[forward] (0.5,-0.050000) to (10,-1) ;
		\node[above] at (15,0) {15} ;
		\node[below] at (10,-1) {15} ;
		\draw[ultra thick] (0.5,0) -- ++(15,0);
		\draw[ultra thick] (0.5,-1) -- ++(15,0);
	\end{tikzpicture}%
	}}
	
	\plaincenter{\adjustbox{max width=\linewidth,scale=.85}{%
	\begin{tikzpicture}[
			scale=.4,
			point/.style = {draw,fill=white,circle,minimum size=12pt,inner sep=0pt,font=\tiny},
			copy point/.style = {point,fill=black!30,draw=black!80,text=black!80},
			graph edge/.style = {yellow,very thick},
		]
		\fill[black!30,opacity=.2] (0,0) rectangle (15+.3,45+1);
		\fill[black!30,opacity=.2] (2*15+.7,0) rectangle (3*15+1,45+1);
		\fill[orange,opacity=.1] (0,0) rectangle (45+1,15+.3);
		\fill[green!80!black,opacity=.1] (0,2*15+.7) rectangle (45+1,3*15+1);
		\draw[densely dotted] (0,0) grid (46,46) ;
		\draw[->,thick] (0,0) -- (46.500000,0) ;
		\draw[->,thick] (0,0) -- (0,46.500000) ;
			\node at (1,-.5) {\tiny 1} ;
			\node at (16,-.5) {\tiny 1} ;
			\node at (31,-.5) {\tiny 1} ;
			\node at (2,-.5) {\tiny 2} ;
			\node at (17,-.5) {\tiny 2} ;
			\node at (32,-.5) {\tiny 2} ;
			\node at (3,-.5) {\tiny 3} ;
			\node at (18,-.5) {\tiny 3} ;
			\node at (33,-.5) {\tiny 3} ;
			\node at (4,-.5) {\tiny 4} ;
			\node at (19,-.5) {\tiny 4} ;
			\node at (34,-.5) {\tiny 4} ;
			\node at (5,-.5) {\tiny 5} ;
			\node at (20,-.5) {\tiny 5} ;
			\node at (35,-.5) {\tiny 5} ;
			\node at (6,-.5) {\tiny 6} ;
			\node at (21,-.5) {\tiny 6} ;
			\node at (36,-.5) {\tiny 6} ;
			\node at (7,-.5) {\tiny 7} ;
			\node at (22,-.5) {\tiny 7} ;
			\node at (37,-.5) {\tiny 7} ;
			\node at (8,-.5) {\tiny 8} ;
			\node at (23,-.5) {\tiny 8} ;
			\node at (38,-.5) {\tiny 8} ;
			\node at (9,-.5) {\tiny 9} ;
			\node at (24,-.5) {\tiny 9} ;
			\node at (39,-.5) {\tiny 9} ;
			\node at (10,-.5) {\tiny 10} ;
			\node at (25,-.5) {\tiny 10} ;
			\node at (40,-.5) {\tiny 10} ;
			\node at (11,-.5) {\tiny 11} ;
			\node at (26,-.5) {\tiny 11} ;
			\node at (41,-.5) {\tiny 11} ;
			\node at (12,-.5) {\tiny 12} ;
			\node at (27,-.5) {\tiny 12} ;
			\node at (42,-.5) {\tiny 12} ;
			\node at (13,-.5) {\tiny 13} ;
			\node at (28,-.5) {\tiny 13} ;
			\node at (43,-.5) {\tiny 13} ;
			\node at (14,-.5) {\tiny 14} ;
			\node at (29,-.5) {\tiny 14} ;
			\node at (44,-.5) {\tiny 14} ;
			\node at (15,-.5) {\tiny 15} ;
			\node at (30,-.5) {\tiny 15} ;
			\node at (45,-.5) {\tiny 15} ;
			\node at (46,-.5) {\scriptsize $v$} ;
			\node at (-.5,1) {\tiny 1} ;
			\node at (-.5,16) {\tiny 1} ;
			\node at (-.5,31) {\tiny 1} ;
			\node at (-.5,2) {\tiny 2} ;
			\node at (-.5,17) {\tiny 2} ;
			\node at (-.5,32) {\tiny 2} ;
			\node at (-.5,3) {\tiny 3} ;
			\node at (-.5,18) {\tiny 3} ;
			\node at (-.5,33) {\tiny 3} ;
			\node at (-.5,4) {\tiny 4} ;
			\node at (-.5,19) {\tiny 4} ;
			\node at (-.5,34) {\tiny 4} ;
			\node at (-.5,5) {\tiny 5} ;
			\node at (-.5,20) {\tiny 5} ;
			\node at (-.5,35) {\tiny 5} ;
			\node at (-.5,6) {\tiny 6} ;
			\node at (-.5,21) {\tiny 6} ;
			\node at (-.5,36) {\tiny 6} ;
			\node at (-.5,7) {\tiny 7} ;
			\node at (-.5,22) {\tiny 7} ;
			\node at (-.5,37) {\tiny 7} ;
			\node at (-.5,8) {\tiny 8} ;
			\node at (-.5,23) {\tiny 8} ;
			\node at (-.5,38) {\tiny 8} ;
			\node at (-.5,9) {\tiny 9} ;
			\node at (-.5,24) {\tiny 9} ;
			\node at (-.5,39) {\tiny 9} ;
			\node at (-.5,10) {\tiny 10} ;
			\node at (-.5,25) {\tiny 10} ;
			\node at (-.5,40) {\tiny 10} ;
			\node at (-.5,11) {\tiny 11} ;
			\node at (-.5,26) {\tiny 11} ;
			\node at (-.5,41) {\tiny 11} ;
			\node at (-.5,12) {\tiny 12} ;
			\node at (-.5,27) {\tiny 12} ;
			\node at (-.5,42) {\tiny 12} ;
			\node at (-.5,13) {\tiny 13} ;
			\node at (-.5,28) {\tiny 13} ;
			\node at (-.5,43) {\tiny 13} ;
			\node at (-.5,14) {\tiny 14} ;
			\node at (-.5,29) {\tiny 14} ;
			\node at (-.5,44) {\tiny 14} ;
			\node at (-.5,15) {\tiny 15} ;
			\node at (-.5,30) {\tiny 15} ;
			\node at (-.5,45) {\tiny 15} ;
			\node[anchor=east] at (0,46) {\scriptsize $\pi^{-1}(v)$} ;
		\node[point] (v1) at (16,13) {1} ;
		\node[copy point] (r1) at (31,28) {1} ;
		\node[point] (v2) at (17,18) {2} ;
		\node[copy point] (l2) at (2,3) {2} ;
		\node[copy point] (r2) at (32,33) {2} ;
		\node[point] (v3) at (18,24) {3} ;
		\node[copy point] (l3) at (3,9) {3} ;
		\node[copy point] (r3) at (33,39) {3} ;
		\node[point] (v4) at (19,22) {4} ;
		\node[copy point] (l4) at (4,7) {4} ;
		\node[copy point] (r4) at (34,37) {4} ;
		\node[point] (v5) at (20,19) {5} ;
		\node[copy point] (l5) at (5,4) {5} ;
		\node[copy point] (r5) at (35,34) {5} ;
		\node[point] (v6) at (21,12) {6} ;
		\node[copy point] (r6) at (36,27) {6} ;
		\node[point] (v7) at (22,21) {7} ;
		\node[copy point] (l7) at (7,6) {7} ;
		\node[copy point] (r7) at (37,36) {7} ;
		\node[point] (v8) at (23,26) {8} ;
		\node[copy point] (l8) at (8,11) {8} ;
		\node[copy point] (r8) at (38,41) {8} ;
		\node[point] (v9) at (24,31) {9} ;
		\node[copy point] (l9) at (9,16) {9} ;
		\node[point] (v10) at (25,29) {10} ;
		\node[copy point] (l10) at (10,14) {10} ;
		\node[copy point] (r10) at (40,44) {10} ;
		\node[point] (v11) at (26,32) {11} ;
		\node[copy point] (l11) at (11,17) {11} ;
		\node[point] (v12) at (27,30) {12} ;
		\node[copy point] (l12) at (12,15) {12} ;
		\node[copy point] (r12) at (42,45) {12} ;
		\node[point] (v13) at (28,23) {13} ;
		\node[copy point] (l13) at (13,8) {13} ;
		\node[copy point] (r13) at (43,38) {13} ;
		\node[point] (v14) at (29,35) {14} ;
		\node[copy point] (l14) at (14,20) {14} ;
		\node[point] (v15) at (30,40) {15} ;
		\node[copy point] (l15) at (15,25) {15} ;
		\begin{pgfonlayer}{background}
			\draw[graph edge] (l3) to (l4) ; 
			\draw[graph edge] (l3) to (l5) ; 
			\draw[graph edge] (l3) to (l7) ; 
			\draw[graph edge] (l3) to (l13) ; 
			\draw[graph edge] (l4) to (l5) ; 
			\draw[graph edge] (l4) to (l7) ; 
			\draw[graph edge] (l8) to (l13) ; 
			\draw[graph edge] (l9) to (l10) ; 
			\draw[graph edge] (l9) to (l12) ; 
			\draw[graph edge] (l9) to (l13) ; 
			\draw[graph edge] (l9) to (v1) ; 
			\draw[graph edge] (l9) to (v6) ; 
			\draw[graph edge] (l10) to (l13) ; 
			\draw[graph edge] (l10) to (v1) ; 
			\draw[graph edge] (l10) to (v6) ; 
			\draw[graph edge] (l11) to (l12) ; 
			\draw[graph edge] (l11) to (l13) ; 
			\draw[graph edge] (l11) to (v1) ; 
			\draw[graph edge] (l11) to (v6) ; 
			\draw[graph edge] (l12) to (l13) ; 
			\draw[graph edge] (l12) to (v1) ; 
			\draw[graph edge] (l12) to (v6) ; 
			\draw[graph edge] (l14) to (v1) ; 
			\draw[graph edge] (l14) to (v2) ; 
			\draw[graph edge] (l14) to (v5) ; 
			\draw[graph edge] (l14) to (v6) ; 
			\draw[graph edge] (l15) to (v1) ; 
			\draw[graph edge] (l15) to (v2) ; 
			\draw[graph edge] (l15) to (v3) ; 
			\draw[graph edge] (l15) to (v4) ; 
			\draw[graph edge] (l15) to (v5) ; 
			\draw[graph edge] (l15) to (v6) ; 
			\draw[graph edge] (l15) to (v7) ; 
			\draw[graph edge] (l15) to (v13) ; 
			\draw[graph edge] (v1) to (v6) ; 
			\draw[graph edge] (v2) to (v6) ; 
			\draw[graph edge] (v3) to (v4) ; 
			\draw[graph edge] (v3) to (v5) ; 
			\draw[graph edge] (v3) to (v6) ; 
			\draw[graph edge] (v3) to (v7) ; 
			\draw[graph edge] (v3) to (v13) ; 
			\draw[graph edge] (v4) to (v5) ; 
			\draw[graph edge] (v4) to (v6) ; 
			\draw[graph edge] (v4) to (v7) ; 
			\draw[graph edge] (v5) to (v6) ; 
			\draw[graph edge] (v8) to (v13) ; 
			\draw[graph edge] (v9) to (v10) ; 
			\draw[graph edge] (v9) to (v12) ; 
			\draw[graph edge] (v9) to (v13) ; 
			\draw[graph edge] (v9) to (r1) ; 
			\draw[graph edge] (v9) to (r6) ; 
			\draw[graph edge] (v10) to (v13) ; 
			\draw[graph edge] (v10) to (r1) ; 
			\draw[graph edge] (v10) to (r6) ; 
			\draw[graph edge] (v11) to (v12) ; 
			\draw[graph edge] (v11) to (v13) ; 
			\draw[graph edge] (v11) to (r1) ; 
			\draw[graph edge] (v11) to (r6) ; 
			\draw[graph edge] (v12) to (v13) ; 
			\draw[graph edge] (v12) to (r1) ; 
			\draw[graph edge] (v12) to (r6) ; 
			\draw[graph edge] (v14) to (r1) ; 
			\draw[graph edge] (v14) to (r2) ; 
			\draw[graph edge] (v14) to (r5) ; 
			\draw[graph edge] (v14) to (r6) ; 
			\draw[graph edge] (v15) to (r1) ; 
			\draw[graph edge] (v15) to (r2) ; 
			\draw[graph edge] (v15) to (r3) ; 
			\draw[graph edge] (v15) to (r4) ; 
			\draw[graph edge] (v15) to (r5) ; 
			\draw[graph edge] (v15) to (r6) ; 
			\draw[graph edge] (v15) to (r7) ; 
			\draw[graph edge] (v15) to (r13) ; 
			\draw[graph edge] (r1) to (r6) ; 
			\draw[graph edge] (r2) to (r6) ; 
			\draw[graph edge] (r3) to (r4) ; 
			\draw[graph edge] (r3) to (r5) ; 
			\draw[graph edge] (r3) to (r6) ; 
			\draw[graph edge] (r3) to (r7) ; 
			\draw[graph edge] (r3) to (r13) ; 
			\draw[graph edge] (r4) to (r5) ; 
			\draw[graph edge] (r4) to (r6) ; 
			\draw[graph edge] (r4) to (r7) ; 
			\draw[graph edge] (r5) to (r6) ; 
			\draw[graph edge] (r8) to (r13) ; 
			\draw[graph edge] (r10) to (r13) ; 
			\draw[graph edge] (r12) to (r13) ; 
		\end{pgfonlayer}
	\end{tikzpicture}%
	}}}
	\ifspringer{\end{adjustbox}}{}%
	\caption{%
		A larger circular permutation graph with $n=15$ vertices, 
		represented by the permutation diagram with crossings (top) and 
		the grid representation of the thrice-unrolled PG (bottom).
		In the permutation diagram, noncrossing chords are drawn black, 
		forward crossing chords are green (vertices 9, 11, 14, 15) and 
		backward crossing chords are brown (vertices 1, 6).
	}
	\label{fig:circular-example}
\end{figure}

Mapping between vertex $v$ in $G$ and the $x$-coordinates of $\ell_v$, $c_v$, $r_v$ in $G_3$ is trivial.
To access the $y$-coordinate for a point $(x,y)$, $y(x)$, we consult the type of the corresponding vertex $v$:
\begin{align*}
		y(\ell_v) 
	&\wrel= 
		\begin{cases*}
			\Pi[v] & if $t[v] = N$\\
			\Pi[v] + \like[l]{2n}{n} & if $t[v] = F$
		\end{cases*}
\\
		y(c_v) 
	&\wrel= 
		\begin{cases*}
			\Pi[v] + n& if $t[v] = N$\\
			\Pi[v] + 2n & if $t[v] = F$\\
			\Pi[v]  & if $t[v] = B$
		\end{cases*}
\\
		y(r_v) 
	&\wrel= 
		\begin{cases*}
			\Pi[v] +2n& if $t[v] = N$\\
			\Pi[v] +n & if $t[v] = B$
		\end{cases*}
\end{align*}
All can be answered in $\Oh(1)$ time.
Based on that, we can answer the main queries.

\paragraph{Adjacency}
$u < v$ are adjacent (in $G_{\pi,t}$) iff
$y(c_u) > y(c_v) \vee y(\ell_v) > y(c_u) \vee y(c_v) > y(r_u)$;
if any of the involved copies does not exist, that part of the condition 
is considered unfulfilled. 

\paragraph{Neighborhood}
Given a vertex $v$, we use \wref{lem:cpg-neighborhood} to reduce the query
to neighborhood queries on $G_3$.
To compute the neighborhood of $c_v$ in the PG $G_3$, we use the same method as in 
\wref{sec:ds-array-based}; for that we store the range-minimum/maximum
index from \wref{lem:rmq-indexing} for the sequence of $y$-values
of all vertices in $G_3$ 
(filling empty slot from missing copies with $+\infty$, resp.\ $-\infty$, values).
Note that this index only requires access to individual values in the sequence of $y$-values
(which we can provide in constant time); 
it does not require the values to be stored explicitly in an array.
The additional space cost for constant-time range-min/max queries is only $\epsilon n$ bits.
The time stated for \GDegree follows from counting the neighbors one by one.

\paragraph{Distance and shortest paths}
As for neighborhood, we augment our data structure with the additional data structures
from \wref{sec:ds-distance} for the PG $G_3$, \ie,
we define $A$, $B$, $a^{\pm}(v)$, $b^\pm(v)$,
and $G_A$, $G_B$ as before for $G_3$. 
All now have up to $3n$ vertices instead of $n$, but only occupy $\Oh(n)$ bits in total.

By construction, two vertices $u$ and $v$ in $G_3$ are only adjacent if the corresponding
vertices in $G$ are adjacent.
Therefore, the distance between $u$ and $v$ can be found as the minimum over all
combinations of copies of $u$ and $v$ in $G_3$ (at most 9).

For (the first vertex on) a shortest path, the minimal distance pair of copies can be used
with the \GSPathFirst query on $G_3$.

\medskip\noindent
This concludes the proof of \wref{thm:succinct-cpg}.

%% file: permutation-graphs-semi-distributed.tex
\section{Semi-Distributed Graph Representations}
\label{sec:semi-distributed}
\label{app:semi-distributed}

While Bazzaro and Gavoille~\cite{BazzaroGavoille2009} report that no distance labeling scheme
for PGs exists with less than $3\lg(n)(1-o(1))$ bits per label,
our succinct data structure with overall $n\lg (n)(1+o(1))$ bits of space clearly demonstrates that
this lower bound can be overcome in ``centralized'' data structures.
An interesting question is whether this lower bound can also be circumvented using only a \emph{small amount of global memory} on top of the local labels.

More formally, a \emph{semi-distributed} (distance) oracle consists of a vertex labeling $\ell:V\to \{0,1\}^\star$
and a data structure $\mathcal D$, so that $\GDistance(u,v)$
can be computed from $(\ell(u), \ell(v), \mathcal D)$.
If we allow arbitrary data structures $\mathcal D$, this notion is not very interesting;
one could simply ask $\mathcal D$ to compute all queries.
But if we restrict~$\mathcal D$ to less space than necessary to simply encode the graph, we obtain
an interesting model of computation that interpolates between 
standard data structures and labeling schemes.

Let us call a representation an $\langle L(n), D(n)\rangle$-space semi-distributed representation
if for every $n$-vertex graph we have $|\ell(v)| \le L(n)$ for all vertices $v$ and
$|\mathcal D| \le D(n)$.
Our question can then be formulated as follows:
\textit{What is the smallest $D(n)$ that permits a $\langle(3-\epsilon)\lg n,D(n)\rangle $ space 
semi-distributed distance oracle for permutation graphs?}

The known distance labeling scheme from~\cite{BazzaroGavoille2009} 
implies a $\langle 9\lg n, 0 \rangle$-space semi-distributed representation,
and our succinct data structure constitutes a $\langle \lg n, n \lg(n)(1+o(1)) \rangle$-space 
semi-distributed representation.

A closer look at \wref{sec:data-structures} reveals that the dominant
space in our (array-based) data structure comes from storing $\pi^{-1}$.
In particular, all further data structures required to answer \GDistance queries
occupy only $O(n)$ bits of space.
Moreover, all computations to determine distances, and even the entire shortest path,
require only $\pi^{-1}$ of the original endpoints (cf.\ \wref{rem:pi-inv-of-A-B}).
We can thus move $\pi^{-1}(v)$ into the label of node $v$, thereby making it inaccessible 
from any other vertex without affecting the queries.
We hence obtain the following result.

\begin{theorem}[Semi-distributed PGs]
\label{thm:semi-distributed}
	Permutation graphs admit a $\langle 2\lg n, \Oh(n)\rangle$-space semi-distributed representation
	that allows to answer the following queries:
	\GAdjacent, \GDistance, and \GSPathFirst in $\Oh(1)$ time and
	$\GSPath(u,v)$ in $\Oh(\GDistance(u,v)+1)$ time.
\end{theorem}

\begin{proof}
	The label $\ell(v)$ consists of the pair of $(v, \pi^{-1}(v))$,
	\ie, the $x$- and $y$-coordinate in the grid representation of $G$.
	All remaining data structures from \wref{sec:data-structures} occupy $\Oh(n)$ bits of space.
	As discussed above, for the listed operations access to $\pi^{-1}$ 
	is only needed for the queried vertices.
\end{proof}

\begin{remark}[Who stores the labels]
	Note that in our succinct data structures, we identify vertices with the (left-to-right) ranks
	of the upper endpoints of their chords in the permutation diagram. 
	That means that the user of our data structure is willing to let
	(the construction algorithm of) our succinct data structure decide how to label vertices,
	and vertices are henceforth referred to using these labels.
	In a (semi-)distributed representation, we have to assign \emph{and store} a unique label 
	for each vertex, because queries are computed only from the \emph{labels} of the vertices
	(and potentially $\mathcal D$).
	The semi-distributed scheme derived from our succinct representation therefore
	takes up a total of $\sim 2n \lg n$ bits.
\end{remark}

This $\langle 2\lg n, O(n)\rangle$ scheme circumvents the lower bound for distance labelings 
in label length and overall space;
it thus gives a novel trade-off beyond the
fully distributed and fully centralized representations.
In particular, it shows that access to global storage, even a fairly limited amount,
is inherently more powerful than a fully-distributed labeling scheme.

%
%
